\documentclass[ejs]{imsart}

\RequirePackage{amsthm,amsmath,amsfonts,amssymb}
\RequirePackage{natbib}
\usepackage{algorithmic}
\usepackage{color}
\usepackage[ruled,vlined]{algorithm2e}
\usepackage{multirow}
\usepackage{comment}
\usepackage{ulem}

\usepackage{amsmath,amsfonts,amssymb}
\usepackage{graphicx,psfrag,epsf}
\usepackage{enumerate}
\usepackage[ruled,vlined]{algorithm2e}
\usepackage{algorithmic}
\usepackage{calrsfs}
\usepackage{algorithmic}
\usepackage{color}
\usepackage[ruled,vlined]{algorithm2e}
\usepackage{multirow}
\usepackage{comment}
\usepackage{ulem}
\newcommand{\parallelsum}{\mathbin{\!/\mkern-5mu/\!}}
\allowdisplaybreaks

\newcommand{\bbR}{\mathbb{R}}

\newcommand{\sC}{\mathcal{C}}
\newcommand{\sD}{\mathcal{D}}
\newcommand{\sE}{\mathcal{E}}
\newcommand{\sH}{\mathcal{H}}

\newcommand{\sR}{\mathcal{R}}
\newcommand{\sS}{\mathcal{S}}
\newcommand{\sT}{\mathcal{T}}
\newcommand{\sO}{\mathcal{O}}

\newcommand{\sU}{\mathcal{U}}
\newcommand{\sV}{\mathcal{V}}

\newcommand{\sZ}{\mathcal{Z}}

\newcommand{\ve}{\mathbf{e}}

\newcommand{\vx}{\mathbf{x}}

\newcommand{\vtheta}{\boldsymbol{\theta}}

\newcommand{\vxi}{\boldsymbol{\xi}}

\newcommand{\vepsilon}{\boldsymbol{\epsilon}}
\newcommand{\vzero}{\mathbf{0}}

\newcommand{\tx}{\tilde{x}}

\newcommand{\tp}{\tilde{p}}

\newcommand{\ttheta}{\tilde{\theta}}
\newcommand{\htheta}{\hat{\theta}}

\newcommand{\hPsi}{\widehat{\Psi }}
\newcommand{\hsigma}{\hat{\sigma}}
\newcommand{\tsigma}{\tilde{\sigma}}

\newcommand{\txi}{\tilde{\xi}}
\newcommand{\hxi}{\hat{\xi}}

\newcommand{\hsC}{\widehat{\mathcal{C} }}
\newcommand{\tsC}{\widetilde{\mathcal{C} }}
\newcommand{\tsP}{\widetilde{\mathcal{P} }}
\newcommand{\hsP}{\widehat{\mathcal{P} }}

\newcommand{\mA}{\mathbf{A}}

\newcommand{\mC}{\mathbf{C}}
\newcommand{\mD}{\mathbf{D}}

\newcommand{\mF}{\mathbf{F}}

\newcommand{\mI}{\mathbf{I}}

\newcommand{\mU}{\mathbf{U}}
\newcommand{\mV}{\mathbf{V}}

\newcommand{\mX}{\mathbf{X}}
\newcommand{\mY}{\mathbf{Y}}

\newcommand{\ma}{\mathbf{a}}
\newcommand{\mb}{\mathbf{b}}

\newcommand{\Expect}{\mathbb{E}}
\newcommand{\Prob}{\mathbb{P}}
\newcommand{\intd}{\,\mathrm{d}}

\newcommand{\tr}{\mathrm{tr}}

\newtheorem{theorem}{Theorem}
\newtheorem{lemma}[theorem]{Lemma} 
\newtheorem{proposition}[theorem]{Proposition} 
\newtheorem{remark}[theorem]{Remark}
\newtheorem{corollary}[theorem]{Corollary}

\newtheorem{condition}[theorem]{Condition}

\pubyear{2020}
\volume{0}
\issue{0}
\firstpage{1}
\lastpage{8}
\arxiv{2010.00000}

\startlocaldefs
\numberwithin{equation}{section}
\theoremstyle{plain}

\endlocaldefs

\begin{document}

\begin{frontmatter}
	\title{Functional Principal Subspace Sampling for Large Scale Functional Data Analysis}
	\runtitle{Randonmized FDA with FunPrinSS}
	
	\begin{aug}
		\author[A]{\fnms{Shiyuan} \snm{He}\ead[label=e1]{heshiyuan@ruc.edu.cn}}
		\and 
		\author[B]{\fnms{Xiaomeng} \snm{Yan}\ead[label=e2]{xiaomengyan@stat.tamu.edu}}
		\address[A]{Center for Applied Statistics, Institute of Statistics and Big Data, Renmin University of China, Bejing, 100872, China. \printead{e1}}
		\address[B]{Department of Statistics, Texas A\&M University, College Station, TX, 77840, USA. \printead{e2}}
		
			\runauthor{S. He and X. Yan.}
		
	\end{aug}

	\begin{abstract}
		Functional data analysis (FDA) methods have computational and theoretical appeals for some high dimensional data, but lack the scalability to modern large sample datasets. To tackle the challenge, we develop randomized algorithms for two important FDA methods: functional principal component analysis (FPCA) and functional linear regression (FLR) with scalar response. The two methods are connected as they both rely on the accurate estimation of functional principal subspace. The proposed algorithms draw subsamples from the large dataset at hand and apply FPCA or FLR over the subsamples to reduce the computational cost. To effectively preserve subspace information in the subsamples, we propose a functional principal subspace sampling probability, which removes the eigenvalue scale effect inside the functional principal subspace and properly weights the residual. Based on the operator perturbation analysis, we show the proposed probability has precise control over the first order error of the subspace projection operator and can be interpreted as an importance sampling for functional subspace estimation. Moreover, concentration bounds for the proposed algorithms are established to reflect the low intrinsic dimension nature of functional data in an infinite dimensional space. The effectiveness of the proposed algorithms is demonstrated upon synthetic and real datasets.
	\end{abstract}
	
	
\tableofcontents
	
\end{frontmatter}

\section{Introduction}
\label{sec:introduction}

The modern era brings us the opportunities and challenges of big data, where both sample size and dimension increase out of the  capacity of classical statistical methods.  Examples of the large scale datasets include:  millions of  spectra collected by   astronomical surveys (e.g. LAMOST of \cite{zhao2012lamost} and SDSS of \cite{eisenstein2011sdss});   petabyte amounts of hyperspectral images recorded by airborne or satellite remote sensing \citep{liu2018remote}; brain medical images through the functional magnetic resonance imaging  \citep[fMRI][]{tian2010functional, turkbrowne2013functional}.  A common characteristic of these datasets is that their observations  can be viewed as functions over some continuous domain (e.g. temporal or spatial domain), and functional data analysis \citep[FDA,][]{kokoszka2017introduction} can serve as a designated toolbox.  FDA treats functional observations as realizations of some stochastic process  with  low intrinsic dimension in an infinite dimensional Hilbert space.  This means observations can be effectively approximated by splines,  wavelets or linear combinations of a few functional principal components. In contrast to classical multivariate methods,  FDA takes these features into consideration to achieve  optimal convergence rates.   Despite their demonstrated value in theory and application, the current FDA  methods lack the scalability for modern large scale datasets \citep{hadjipantelis2018functional}.

Among the current approaches to making statistical methods scalable to large datasets, randomized algorithms have recently gained popularity.   The randomized algorithms approach the problem by constructing a  substantially smaller  sketch of the large scale data at hand. Then, existing  methods  are  applied to the sketch to reduce computation cost. When the sketch  keeps the most relevant information,  the result computed from the sketch should remain close to the result from the original dataset.  	There are a few ways to construct a sketch of a large data set. One approach is to draw subsamples  with respect to some carefully designed probability, which will select  informative samples with larger probability. These sampling probabilities include  the importance sampling in matrix multiplication \citep{drineas2006fast}, the leverage sampling for least squares regression \citep{drineas2012fast,ma2015statistical}, the subspace sampling for low rank matrix construction \citep{Drineas2006Subspace}, etc. Effective sampling probabilities have also been proposed for logistic regression \citep{wang2018optimal,wang2019more} and generalized linear models \citep{zhang2021optimal} to minimize the asymptotic  variance of the estimator.   	An alternative way towards sketching is to mix the original data with a random projection and  draw samples from the projected data  \citep{drineas2011faster,wang2017sketched}. 

Since the seminal works of \cite{drineas2006fast, drineas2006fast2, drineas2006fast3}  for matrix multiplication and approximation, the idea of randomized algorithm has  been successfully applied to optimization \citep{pilanci2015randomized,pilanci2017newton}, low rank matrix estimation \citep{halko2011finding, Drineas2006Subspace}, least squares regression \citep{drineas2012fast}, nonparametric kernel regression \citep{yang2017randomized}, etc.  These algorithms are able to yield comparatively  accurate results at reduced computational and storage costs. See the references \cite{drineas2018lectures,woodruff2014sketching}  for an overview.    Most theoretical analysis of these randomized algorithms is conducted from the algorithmic perspective, where the analysis is carried out conditionally on an arbitrarily fixed dataset.   Some recent works \citep{ma2015statistical,raskutti2016a,wang2017sketched} also  draw analysis from the statistical perspective, where statistical properties such as bias and average prediction error are considered.  

The current literature of randomized algorithms focuses on multivariate statistical methods which are not directly applicable to the functional data setting.  Their theoretical results also does not naturally extend to  functional data in an infinite dimensional space.  \cite{he2020randomized}  recently studied randomized estimation of functional covariance operator. Still, most functional data methods remain to be explored via the randomized algorithm approach for scalability. Among these methods, functional principal component analysis (FPCA)
\citep[e.g.][]{james2000principal,yao2005functional, peng2009geometric} is one of the most important.  Like the classical principal component analysis, FPCA identifies a few orthonormal functional principal components  (FPC)  explaining most of the variability of a dataset.  The computation and theoretical analysis of functional data in an infinite-dimensional  space are usually facilitated by FPCA. Indeed, a large amount of  FDA methods have been proposed based on FPC.  For example, when fitting functional linear regression,  \citep{hall2007methodology,kato2012estimation} first identify a few FPCs of the predictor,  and then the model is fitted  via the leading FPCs. For functional two-sample  test, the test statistic can be constructed based on FPCs \citep{horvath2013estimation}. The work of \cite{delaigle2010defining}  relies on FPCs to define probability density for random functions.  More examples can be found in~\cite{horvath2012inference}.

Due to the fundamental role of FPCA, we study a randomized FPCA algorithm and extend it to functional linear regression (FLR) with scalar response in this work. Our work focuses on the fully observed functional data case, where the function value is known at any point of its domain. In practical applications, for example, these could be functions recorded with a fixed high frequency over a time interval or recorded over a dense spatial grid. This type of data has been widely collected by astronomical spectral surveys, remote sensing, etc, as discussed at the beginning of this section.  Specifically, we will view  functional observations as indivisible elements in a Hilbert space, and the sketched data is constructed based on subsampling.  Our basic algorithm is motivated by the work of~\cite{he2020randomized} for which we now provide a review.

\subsection{Review of the Importance Sampling}
\label{subsec:intro:CovOperator}

Let  $\sH_X$  be a Hilbert space   equipped with an inner product $\langle\cdot, \cdot\rangle$ and an induced norm $\|\cdot \|$. 
For example, when the space $\sH_X = L_2(T)$ is the set of all square integrable functions over a compact interval $T$, the inner product is $\langle x, x'\rangle = \int_T x(t) x'(t)\intd t$ and the   norm is $\|x\| = \big(\int_T x^2(t) \intd t\big)^{1/2}$ for $x,x'\in\sH_X$. Given a random element $x\in\sH_X$ with zero mean and finite second moment $\Expect \|x\|^2$, its \textit{covariance operator} $\sC_{XX}$: $\sH_X \mapsto \sH_X$ is a mapping such that $u\in\sH_X$ is mapped to $\Expect \langle u,x\rangle x$. Equivalently, we can write $\sC_{XX}= \Expect\left( x\otimes x\right)$, where $\otimes$ is  tensor product such that $\left(x\otimes x\right) u = \langle x, u\rangle \times x$ for any $u\in \sH_X$.  Covariance operator generalizes the concept of covariance matrix in multivariate statistics.
For any $u,v \in \sH_X$, it holds that $\langle\sC_{XX} u, v\rangle = \Expect  (\langle x, u\rangle \langle x, v\rangle)$. We can see that $\langle\sC_{XX} u, v\rangle$ quantifies the correlation between the random scalars $\langle x, u\rangle$ and $\langle x, v\rangle$. A comprehensive introduction to these FDA concepts can be found in \cite{hsing2015theoretical}.

\begin{algorithm}[t]
	\caption{Randomized Covariance Operator Estimation\label{alg:sampleCov}}
	\textbf{Input}: Dataset $\{x_n\}_{n=1}^N$; sampling probability $\{p_n\}_{n=1}^N$;	subsample size $C$. \\
	\textbf{Output}: $\tsC_{XX}$. 
	\begin{algorithmic}[1]
		\FOR{$c = 1,\cdots, C$}
		\STATE Sample $\tx_c$ from $\{x_n\}_{n=1}^N$ and get 
		$\tp_c$ with probability $\Prob\big((\tx_c, \tp_c)= (x_n, p_n)\big) = p_n.$
		\ENDFOR
		\STATE Compute $\tsC_{XX} = \frac{1}{C}\sum_{c=1}^C \frac{1}{N\tp_c}\tx_c\otimes \tx_c$.
	\end{algorithmic}
\end{algorithm}

Suppose we have  $N$ (not necessarily independent) realizations $\{x_n\}_{n=1}^N$ of  the random element $x$. The  covariance operator $\sC_{XX}$ can be estimated by the \textit{empirical covariance operator} 
$\hsC_{XX} = (1/N)\sum_{n=1}^N x_n\otimes x_n$. The empirical covariance operator involves the calculation of $N$ tensor  products and their summation. This computational cost  is overwhelming for large sample size. For example, when each $x_n\in L_2(T)$ is digitally stored as a high dimensional vector of length $L$ over a dense grid of $T$, the cost of computing $\hsC_{XX}$ scales as $O(N L^2)$. In a typical astronomical survey, $L$ is in the  order of thousands, and $N$ has the magnitude of millions. Under this scenario, computing a covariance operator using all available data for scientific research  is unrealistic.

Based on  the randomized matrix multiplication of~\cite{drineas2006fast}, 
the  algorithm of  \cite{he2020randomized}   estimates the covariance operator from a sketch of  much smaller  size.  The procedure is summarized in Algorithm~\ref{alg:sampleCov}. Given a full dataset $\{x_n\}_{n=1}^N$ and a sampling probability  $\{p_n\}_{n=1}^N$, it draws with replacement a subsample $\{\tx_c\}_{c=1}^C$ of size $C$ from $\{x_n\}_{n=1}^N$.  	Then, the subsampled empirical covariance operator is computed as  
\begin{equation}\label{eqn:subsampleCXX}
	\tsC_{XX} = \frac{1}{C}\sum_{c=1}^C \frac{1}{N\tp_c}\tx_c\otimes \tx_c.
\end{equation}
Suppose $c$-th subsample is indexed by $i_c$ in the original dataset, we have set $\tx_c = x_{i_c}$ and $\tp_c = p_{i_c}$  correspondingly in the above expression.  It is obvious that, conditional on the full dataset $\{x_n\}_{n=1}^N$,  the subsampled $\tsC_{XX}$ is an unbiased estimator of the full sample covariance operator $\hsC_{XX}$ with any strictly positive probability $\{p_n\}_{n=1}^N$.  It is shown in \cite{he2020randomized} that the optimal  sampling probability, which minimizes the expected squared Hilbert-Schmidt norm of the subsampling error $\tsC_{XX} - \hsC_{XX}$,   is of the form $p_n\propto \|x_n\|^2$. This  sampling probability will be referred as the \textit{importance sampling} (IMPO) throughout this work.

\subsection{Limitation of the Importance Sampling}
\label{subsec:intro:toy}
Recall the purpose of this work is to develop randomized functional principal component analysis (FPCA) and its extension to functional linear regression (FLR). 
 At first glance,   the importance sampling (IMPO)  of \cite{he2020randomized} is already a plausible sampling strategy. By direct application of perturbation theory, the eigenfunctions (functional principal components) of $\tsC_{XX}$ and $\hsC_{XX}$ should not be far away from each other, provided the related eigengap is positive and the difference $\tsC_{XX} - \hsC_{XX}$ is small. However, for  FPCA and FLR, this importance sampling probability is  far from being optimal.

To see the reasons,  consider a toy  principal component analysis example in $\bbR^4$. We generate  $N=1000$ observations by $\vx_n = \sigma_1 \xi_{n1} \vtheta_1 +  \sigma_2 \xi_{n2} \vtheta_2  +  \sigma_3 \xi_{n3} \vtheta_3$, where the principal component directions $\vtheta_1 = (1,1,0,0)^T/\sqrt{2}$,  $\vtheta_2 = (1,-1,0,0)^T/\sqrt{2}$ and  $\vtheta_3 = (0,0,1,1)^T/\sqrt{2}$  are scaled by $\sigma_1 = 10$, $\sigma_2=2$, $\sigma_3=0.1$, respectively. The scores $\xi_{n1}, \xi_{n2},\xi_{n3}\in\bbR$  are drawn from the standard normal distribution. 	Figure~\ref{fig:toyFPCA} plots the generated points by their first two coordinates. Suppose we want to estimate the first two principal components $\vtheta_1,\vtheta_2$  from  subsamples, and  the sampling is taken with respect to  the IMPO sampling probability~\citep{he2020randomized}.  To illustrate the effect of IMPO, we color points with larger $p_n\propto \|x_n\|^2= \sigma_1^2\xi_{n1}^2+\sigma_2^2\xi_{n2}^2+\sigma_3^2\xi_{n3}^2$ by darker blue in  the left penal of Figure~\ref{fig:toyFPCA}. It is evident this  probability over-emphasizes the observations along the first direction  $\vtheta_1$ while neglects the observations along    $\vtheta_2$.  When $C$ subsamples are drawn according to IMPO,  the resulting subsamples  will likely contain abundant information for estimating $\vtheta_1$, but  fail to determine $\vtheta_2$ accurately. 	

In the above example, the first two scaling factors only have a moderately large ratio $\sigma_1^2/\sigma_2^2 = 25$, but their effect in determining $p_n$ is not negligible.   To resolve the issue, we may consider to remove the scaling factor to get a new sampling probability $p_n\propto \xi_{n1}^2 + \xi_{n2}^2$. The  points  in the right penal of Figure~\ref{fig:toyFPCA} are colored accordingly. Notice this  probability puts  equal sampling weights for observations along $\vtheta_1$ and $\vtheta_2$.  The subsamples, which are drawn by this probability, are more likely to contain balanced information for estimating the first two principal components.   Moreover, this idea   can be generalized to infinite dimensional space to estimate the leading $R(>0)$ functional principal components.

\begin{figure}[t]
	\centering
	\includegraphics[width = 0.7\textwidth]{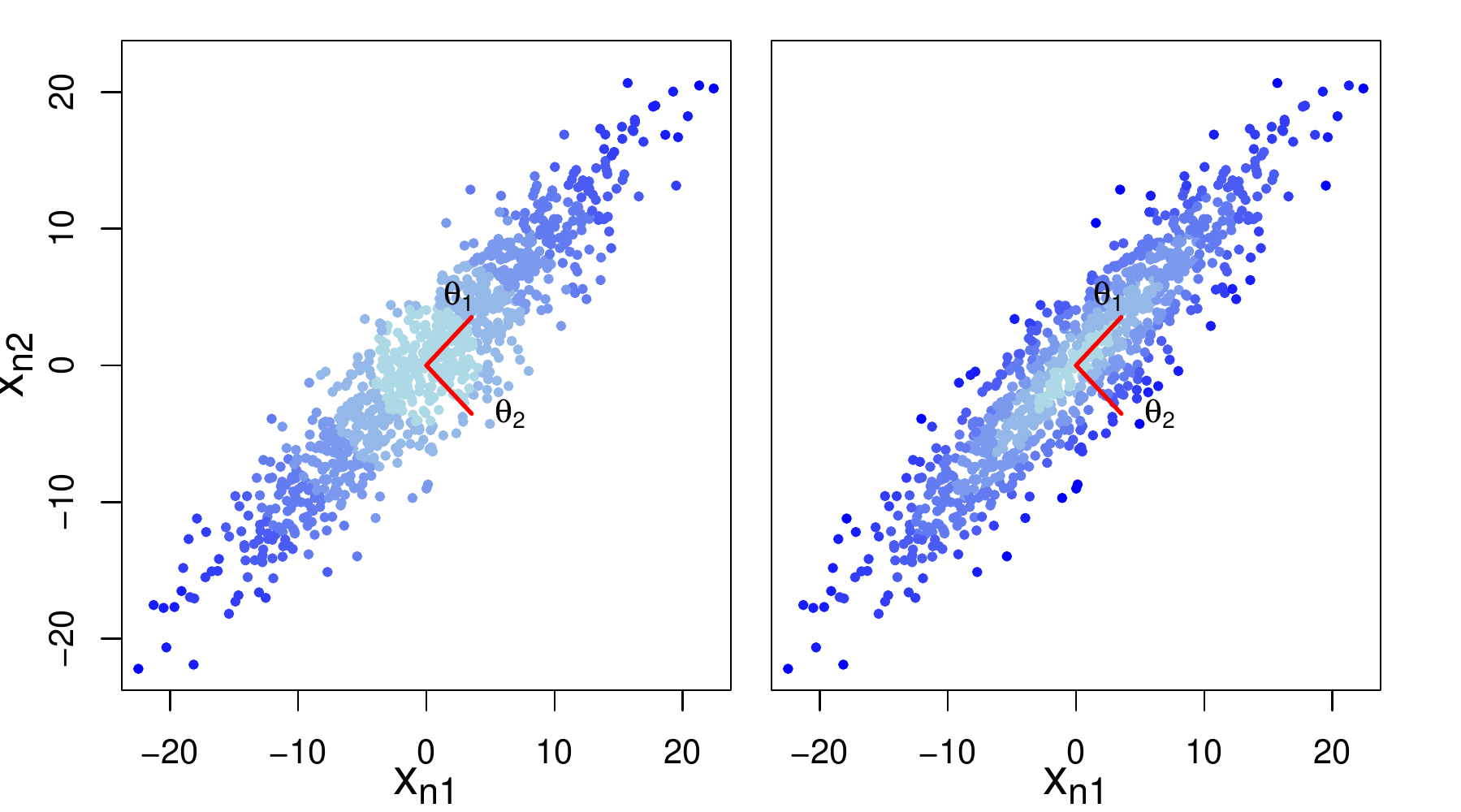} 
	\caption{An toy example for principal component analysis in $\bbR^4$. The $N = 1000$ points $\vx_n = (x_{n1}, x_{n2}, x_{n3}, x_{n4})^T\in\bbR^2$ are plotted by their first two coordinates $(x_{n1}, x_{n2})$.. The red lines mark the first and second principal components $\vtheta_1$ and $\vtheta_2$. The left panel colors the data points by $p_n \propto \| x_n\|^2$, while the right panel by $p_n \propto (\xi_{n1}^2 + \xi_{n2}^2)$. Points with larger sampling probability has darker blue color.} \label{fig:toyFPCA}
\end{figure}

\subsection{Contributions of This Work}
\label{subsec:intro:Contributions}

This work develops randomized algorithms for functional principal component analysis (FPCA) and functional linear regression (FLR) with scalar response. Our methodological and theoretical contributions are summarized in the following paragraphs.

In Section~\ref{sec:fpca}, a randomized algorithm is developed for FPCA. Motivated by the  example from Section~\ref{subsec:intro:toy}, we propose the \textit{functional principal subspace sampling probability} (FunPrinSS)   to effectively preserve   subspace information in  subsamples. The  FunPinSS  probability removes the  eigenvalue scale effect inside the functional principal subspace; and at the same time, it properly weights the complementary subspace  residual.  	The latter part is necessary because we need to control the interaction between the leading and the remaining eigenfunctions in the infinite dimensional space.	 As the exact value of the FunPrinSS probability is unknown unless we execute full sample computation, a fast and theoretically justifiable algorithm is proposed for practical implementation.  
  
  For the theoretical analysis of the randomized FPCA algorithm,   the task of precisely estimating the leading $R$ eigenfunctions  is translated into estimating the projection operator of their spanned  subspace.   The analysis is conducted from an algorithmic perspective, which is conditional on the full dataset and imposes minimal assumption (on eigengap).  With  operator perturbation theory, we are able to decompose the  subsample estimation error  into a first order term and a higher order term.  The proposed FunPrinSS probability is then justified from two aspects:   
  \begin{enumerate}
 \item[(i)] From a heuristic  aspect, we show the proposed probability enjoys the  property of precisely controlling  the first order error.  In fact, the FunPrinSS probability can be interpreted as an importance sampling  for the  first order error. This means FunPrinSS for FPCA plays a  role similar to that of the IMPO sampling~\citep{he2020randomized} in minimizing the  difference $\tsC_{XX} - \hsC_{XX}$ for covariance operator.  
 \item[(ii)] From a  rigorous aspect,  we develop  concentration bounds for the proposed algorithm. The constants of the derived bounds reflect the low intrinsic dimension nature of functional data. That is, the subsample uncertainty is characterized by the dimension of the functional principal subspace and the intrinsic dimension of the residual covariance operator. These dimension parameters are usually small for functional data, and as a result, our theory indicates the subsample  size need not be very large to control the subsample error within a desired accuracy level.
  \end{enumerate}

The randomized FPCA algorithm and its analysis naturally extend to the  functional linear regression  setting in Section~\ref{sec:flr}. For FLR as an inverse problem  in the Hilbert space,  the  involved operator inverse is not well defined due to its infinitely many eigenvalues decreasing to zero.  We study an estimator whose operator inverse is truncated after its leading $R$ functional principal components. It follows that the accuracy of the FLR estimator relies on the precision of the subsampled estimates of functional principal components.  Based on the FunPrinSS probability, a randomized algorithm for FLR is proposed.

For the theoretical justification  the randomized FLR algorithm, we conduct a careful analysis of the  subsample prediction error and decompose it into three components: the functional principal subspace estimation error, the truncated inverse operator estimation error, and an orthogonal residual error.   The error decomposition reveals  distinct sources of uncertainty during subsampling.  Like FPCA, our FLR theoretical  bound is  further established via the operator perturbation expansion  of the related error components. 

The proposed randomized algorithms are closely related to the survey based functional data analysis methods  \citep[see e.g.][]{cardot2011horvitz,cardot2013confidence,degras2014rotation,lardin2014analysing}. These works study  the estimation of the   mean (or total) function, covariance function, or eigenfunctions for a  population (with possibly unknown size).  Classical survey sampling strategies (e.g. simple sampling without replacement and stratified sampling) are employed to draw functional samples. Most of the existing works focus on estimating the mean or total function. When auxiliary variables are available, \cite{cardot2013uniform}  construct a model-assisted estimator of the mean function. For  eigenfunction estimation,  \cite{cardot2010properties}  establish the asymptotic design unbiasedness and asymptotic variance of the  estimator for general survey sampling probabilities.  However, no specific   strategy is  proposed  to reduce the asymptotic variance of the eigenfunction estimator.   To our best knowledge, functional linear regression has not been considered and analyzed yet in the survey based functional data analysis framework.

\vspace{10pt}

We conclude this section by introducing notations used throughout this work.
Let $\mathcal{J}: \sH_X \mapsto \sH_X$ be a bounded linear operator defined on $\sH_X$. The \textit{operator norm} is defined by $\Vert\mathcal{J}\Vert :=  \sup_{u\in\sH_X:\, \|u\|\le 1} \|\mathcal{J}u\|$. Assume that $\{e_{j}\}_{j=1}^\infty$ is a complete orthonormal system for $\sH_X$, the \textit{Hilbert--Schmidt norm} for $\mathcal{J}$ is defined by $\Vert\mathcal{J} \Vert_{HS}^2 :=  \sum_{j=1}^\infty\langle \mathcal{J}e_{j},e_{j}\rangle^2$.  Besides, its \textit{intrinsic dimension} is computed via $\text{intdim}(\mathcal{J}) := \tr(\mathcal{J}) / \Vert \mathcal{J} \Vert$, where the \textit{trace} of $\mathcal{J}$ is defined as $\tr(\mathcal{J}) := \sum_{j=1}^\infty  \langle \mathcal{J} e_{j},e_{j}\rangle$. In this work, the $N$-dimensional Euclidean space $\bbR^N$  is equipped with a normalized inner product $\langle \mathbf{a}, \mathbf{b} \rangle_N = (1/N)\sum_{i=1}^{N} a_i b_i$ for any vectors $\mathbf{a}, \mathbf{b}\in \bbR^{N}$ and its induced norm is defined as $\|\mathbf{a}\|_N=\big(\sum_{i=1}^{N} a_i^2/N\big)^{1/2}$.

The rest of the article is organized as follows. In Section~\ref{sec:fpca}, we develop our randomized algorithm for FPCA. We propose the functional principal subspace probability and a fast two-step algorithm to estimate it. The theoretical analysis of the randomized FPCA algorithms is conducted in Section~\ref{sec:fpca:theory}. In Section~\ref{sec:flr}, the algorithm and theoretical results are  extended to functional linear regression with scalar response. Simulation experiments are conducted in Section~\ref{sec:simu} to compare distinct sampling proposals. The algorithm performance is further examined over an astronomical real dataset in Section~\ref{sec:realdata}.

\section{Functional Principal Subspace Sampling for FPCA}\label{sec:fpca}

Functional principal component analysis \citep[FPCA,][]{horvath2012inference,ramsay2004functional} is the cornerstone for many functional data analysis methods.   To reduce the computation cost of FPCA over a large dataset, we develop a randomized algorithm for FPCA  in Section~\ref{subsec:fpca:algorithm}.   To retain the most relevant principal subspace information after subsampling, we propose the functional principal subspace sampling (FunPrinSS)  probability in Section~\ref{subsec:fpca:prob}.  A practical implementation of our algorithm is discussed in Section~\ref{subsec:fpca:implementation}.

\subsection{Randomized FPCA}
\label{subsec:fpca:algorithm}

FPCA generalizes the classical multivariate principal component analysis to identify the modes of variations among a  dataset.   Classical multivariate principal components can be computed from  the eigen-decomposition of a covariance matrix.  Similarly, functional principal components can be estimated via the \textit{eigenvalue-eigenfunction decomposition}  of the empirical covariance operator  $\hsC_{XX} = \sum_{r=1}^\infty \hsigma_r^2 \htheta_{r}\otimes \htheta_{r}$, where $\hsigma_1^2 \geq \hsigma_2^2\geq\cdots$ are non-negative  eigenvalues, and 
$\htheta_r$'s are the corresponding eigenfunctions (or called \textit{functional principal components}).   According to the Karhunen-Lo\`{e}ve (KL) expansion, each sample $x_n$ can be represented  by 
\begin{equation}\label{eq:KL}
	x_n =\sum_{r = 1}^{\infty}\hsigma_r\hxi_{nr}\htheta_r,
\end{equation} 
where $\hxi_{nr}$ is the $r$-th  \textit{principal component score} of the $n$-th sample. It is computed by $\hxi_{nr} = \langle x_n, \htheta_r\rangle/\hsigma_r$ when  $\hsigma_r^2 > 0$, and we set  $\hxi_{nr} = 1$ when $\hsigma_r^2 =0$.  Let $\vxi_{r} = (\hxi_{1r},\cdots, \hxi_{Nr})^T$ denote the vector of the $r$-th score for all samples.  It satisfies the normalizing property $\| \vxi_{r}\|_N = 1$.  From the KL expansion and the  orthonormality of the eigenfunctions, we can also find that $\|x_n\|^2 = \sum_{r = 1}^\infty\sigma_r^2\hxi_{nr}^2$.

\begin{algorithm}[t]
	\caption{Randomized FPCA\label{alg:sampleFPCA}}
	\textbf{Input}: Dataset $\{x_n\}_{n=1}^N$; sampling probability $\{p_n\}_{n=1}^N$;	subsample size $C$. \\
	\textbf{Output}:  $\tsigma_{r}$ and $\ttheta_r$ for $r=1,\cdots, R$
	\begin{algorithmic}[1]
		\FOR{$c = 1,\cdots, C$}
		\STATE Sample $\tx_c$ from $\{x_n\}_{n=1}^N$ and get 
		$\tp_c$ with probability $\Prob\big((\tx_c, \tp_c)= (x_n, p_n)\big) = p_n.$
		\ENDFOR
		\STATE Compute $\tsC_{XX} = \frac{1}{C}\sum_{c=1}^C \frac{1}{N\tp_c}\tx_c\otimes \tx_c$.
		\STATE Compute the eigenvalue-eigenvector decomposition $\tsC_{XX} = \sum_{r=1}^\infty \tsigma_r^2 \ttheta_r\otimes \ttheta_r$.
	\end{algorithmic}
\end{algorithm}

When the eigenvalues $\hsigma_r^2$'s  have a fast decay toward zero, the covariance operator can be well approximated by the leading eigenpairs,   $\hsC_{XX} \approx \sum_{r=1}^R \hsigma_r^2 \htheta_{r}\otimes \htheta_{r}$ for some $R$.  Besides,  each sample can be approximated	 by $x_n \approx \sum_{r = 1}^{R}\hsigma_r\hxi_{nr}\htheta_r$. In practice,  obtaining the leading $R$  functional principal components $\htheta_{1}, \cdots, \htheta_{R}$ relies on the computation and decomposition of  the empirical $\hsC_{XX}$. As noted in Section~\ref{sec:introduction},  computing $\hsC_{XX}$ is an expensive procedure for a large dataset. We propose to estimate the full sample $\htheta_{1}, \cdots, \htheta_{R}$ based on a subsampled dataset. In particular,  we can apply Algorithm~\ref{alg:sampleCov} to get a subsampled covariance operator $\tsC_{XX}$, and perform eigenvalue-eigenfunction decomposition of $\tsC_{XX}$ instead of $\hsC_{XX}$.  This is summarized in  Algorithm~\ref{alg:sampleFPCA} for some sampling probability $\{p_n\}_{n=1}^N$, where the probability satisfies $p_n \ge 0$ and $\sum_{n=1}^N p_n = 1$.	The algorithm will output subsampled eigenfunction estimates $\ttheta_1,\cdots, \ttheta_R$. The effectiveness of Algorithm~\ref{alg:sampleFPCA} relies on suitably choosing the sampling probability, which will be discussed in details in Section~\ref{subsec:fpca:prob}.

\subsection{Functional Principal Subspace Sampling Probability}
\label{subsec:fpca:prob}

Algorithm~\ref{alg:sampleFPCA} leaves the sampling probability $\{p_n\}_{n=1}^N$ unspecified. As noted in Section~\ref{subsec:intro:toy}, though $p_n\propto \| x_n\|^2$ is optimal in controlling the error $\tsC_{XX} - \hsC_{XX}$ in terms of Hilbert-Schmidt  norm, it is far from being optimal for the principal subspace estimation. The problem arises because the eigenvalues $\hsigma_{r}^2$'s place  unequal weights on distinct principal component directions for $p_n\propto \| x_n\|^2= \sum_{r = 1}^\infty\sigma_r^2\hxi_{nr}^2$.  To remove the influence of eigenvalues in determing the sampling probability, we propose  the  \textit{functional principal subspace sampling (FunPrinSS) probability}, which is
\begin{equation} \label{eqn:prob:FPCAProptoExact}
	p_n^{\mathrm{Exact}} =
	\frac{\sum_{r=1}^R \hxi_{nr}^2 + \|(I - \hsP_R) x_n \|^2  / \hsigma_{R}^2}{\sum_{m=1}^N \big[ \sum_{r=1}^R \hxi_{mr}^2 + \|( I - \hsP_R) x_m \|^2  / \hsigma_{R}^2\big] },
\end{equation}
where $I$ is the identity operator (i.e., $I x = x$ for any $x\in\sH_X$), and 
$\hsP_R = \sum_{r=1}^R \htheta_{r} \otimes \htheta_{r}$ is the projection operator of the subspace spanned by $\htheta_{1},\cdots, \htheta_{R}$.
For the $n$-th sample, the numerator of~\eqref{eqn:prob:FPCAProptoExact} includes two parts: (i) the sum of squared scores in the principal subspace $\sum_{r=1}^R \hxi_{nr}^2$ ; and (ii) the squared norm of $(I-\hsP_R) x_n$, which is the  observation projected out of the principal subspace.  The second part of the numerator is necessary because, in addition to the first $R$ scores, all the remaining scores will affect the operator perturbation  $\tsP_R - \hsP_R$ (see Lemma~\ref{lemma:pcaLinearSum} of Section~\ref{subsec:fpca:maintheory} and the followed discussion).  However, it is not practical to compute  all the remaining scores in the infinite dimensional function space for a large dataset. 
For this reason, we adopt the squared norm of the projected residual $\|( I - \hsP_R) x_n \|^2  / \hsigma_{R}^2$ in~\eqref{eqn:prob:FPCAProptoExact} to account for the effect of the remaining scores.

The proposed sampling probability~\eqref{eqn:prob:FPCAProptoExact} is denoted as ``Exact'' because   its value is determined by  the full sample eigenfunctions $\htheta_{1},\cdots, \htheta_R$ and the related scores $\hat{\xi}_{nr}$. In practice, its exact   value is unknown  unless we carry out  full sample FPCA computation.  Fortunately, practitioners only need an approximate value of~\eqref{eqn:prob:FPCAProptoExact} as  input for  Algorithm~\ref{alg:sampleFPCA}.  A  probability $\{p_n\}_{n=1}^N$ is regarded as an approximation to~\eqref{eqn:prob:FPCAProptoExact} if
	\begin{equation} \label{eqn:prob:NearlyExact}
		p_n \ge p_n^{\mathrm{Exact}} /\beta
	\end{equation}
	for some \textit{fixed} $\beta\ge 1$. We call  a sampling probability $\{p_n\}_{n=1}^N$ satisfying~\eqref{eqn:prob:NearlyExact}  \textit{nearly exact FunPrinSS probability}. Section~\ref{subsec:fpca:implementation} will develop a fast algorithm to obtain an estimate $\{\hat{p}_n\}_{n=1}^N$ fulfilling the lower bound requirement~\eqref{eqn:prob:NearlyExact} with high probability. Our theoretical analysis  of the randomized FPCA  will also be based on any sampling probability satisfying~\eqref{eqn:prob:NearlyExact}.

The normalizing constant  for the sampling probability~\eqref{eqn:prob:FPCAProptoExact} has an interesting interpretation. Define $\mathcal{R}:=\sum_{r=R+1}^\infty \hsigma_r^2 \big(\htheta_r\otimes \htheta_r\big) $ as the residual operator obtained by removing the first $R$ eigenpairs  from $\hsC_{XX}$, i.e. $\mathcal{R}= \hsC_{XX} - \sum_{r=1}^R \hsigma_r^2 \big(\htheta_r\otimes \htheta_r\big)$. Let $\Delta_R$ be the intrinsic dimension of $\mathcal{R}$   such that
\begin{equation} \label{eqn:residualIntrinsic}
	\Delta_R:= \mathrm{int dim}(\mathcal{R}) = 	\tr(\mathcal{\sR}) / \Vert \sR \Vert = 
	\sum_{s=R+1}^\infty \hsigma_s^2    / \hsigma_{R+1}^2,
\end{equation}
if $\hsigma_{R+1}^2 > 0$;  and $	\Delta_R = 0$ if $\hsigma_{R+1}^2 = 0$.
Then, for the normalizing constant in the denominator of~\eqref{eqn:prob:FPCAProptoExact}, we have
\begin{align}
	\sum_{m=1}^N\big[\sum_{r=1}^R  \hxi_{mr}^2 + \|( I - \hsP_R) & x_m \|^2  / \hsigma_{R}^2\big] 
	=\sum_{m=1}^N\big[\sum_{r=1}^R \hxi_{mr}^2 +\sum_{s=R+1}^N \hsigma_s^2 \hxi_{ms}^2   / \hsigma_{R}^2\big] \nonumber \\
	&=N\big(R +\sum_{s=R+1}^N \hsigma_s^2    / \hsigma_{R}^2\big) 
	\le N\big(R + \Delta_R\big), \label{eqn:fpca:dimension}
\end{align}
where the second equality uses the score  normalizing property $\|\vxi_r\|_N = 1$ for  $\vxi_{r} = (\hxi_{1r},\cdots, \hxi_{Nr})^T$. The last inequality uses that $\hsigma_{R}^2\ge \hsigma_{R+1}^2$. The above states that the normalizing constant in~\eqref{eqn:prob:FPCAProptoExact} is upper bounded by   $ N(R + \Delta_R)$. In particular, the quantity $R+\Delta_R$  is a summation of the principal subspace dimension $R$ and the residual  intrinsic dimension $\Delta_R$.

\subsection{Implementation of FunPrinSS}\label{subsec:fpca:implementation}

\begin{algorithm}[t]
	\caption{Estimating the Functional Principal Subspace Probability\label{alg:sampleProb}}
	\textbf{Input}: The dataset $\{x_n\}_{n=1}^N$;
	the subsample size $C$; the number  $R$. \\
	\textbf{Output}: The estimated $\{p_n \}_{n=1}^N$. 
	\begin{algorithmic}[1]
		\STATE Compute $p_n' = \alpha \cdot (1/N) + (1-\alpha) \cdot \| x_n\|^2  /\big(\sum_{m=1}^N  \| x_n\|^2\big)$ for each $n$.
		\STATE Apply Algorithm~\ref{alg:sampleFPCA}  with  
		$\{p'_n\}_{n=1}^N$ and $C$ to get the pilot estimates $(\tsigma'_{r})^2$ and $\ttheta'_r$ for $r\le R$.
		\FOR{$n=1,\cdots, N$}
		\STATE Compute $\hxi'_{nr} = \langle\ttheta'_r, x_n\rangle / \tsigma'_{r}$ for $r=1,\cdots, R$.
		\STATE Compute $\hat{p}_n = \sum_{r=1}^R (\hxi'_{nr})^2 + \|(I - \tsP'_R) x_n \|^2  / (\tsigma'_{R})^2$.
		\ENDFOR
		\STATE Normalize $\{\hat{p}_n \}_{n=1}^N$ to become a proper probability.
	\end{algorithmic}
\end{algorithm}

Recall the exact value of~\eqref{eqn:prob:FPCAProptoExact} is not available unless we obtain $\htheta_1, \cdots, \htheta_R$ with full sample computation. This would contradict the goal of reducing computational cost via subsampling. In practice, a two-step procedure can be carried out:

\begin{quote}
\begin{enumerate}
	\item[Step 1.]	Obtain an estimate $\{\hat{p}_n\}_{n=1}^N$  of the exact sampling probability~\eqref{eqn:prob:FPCAProptoExact}  via the randomized procedure in Algorithm~\ref{alg:sampleProb}. 
	\item[Step 2.] Plug the  estimated probability $\{\hat{p}_n\}_{n=1}^N$   into Algorithm~\ref{alg:sampleFPCA}, which will gives the final estimator of the functional principal components. 
\end{enumerate} 
\end{quote}
Specifically, in the first step, Algorithm~\ref{alg:sampleProb} estimates the  sampling probability~\eqref{eqn:prob:FPCAProptoExact}  as follows. At first (Line 1--2), it executes Algorithm~\ref{alg:sampleFPCA} with the  sampling probability 
$$p_n' = \alpha  (1/N) + (1-\alpha)  \| x_n\|^2  / \Big(\sum_{m=1}^N  \| x_n\|^2\Big),$$ for some $\alpha\in[0,1]$.  The probability $\{p_n'\}_{n=1}^N$ is a mixture of the uniform sampling probability and the  importance sampling probability \citep{he2020randomized}.  This provides a pilot estimate of the first $R$ eigenvalues  $(\tsigma'_{r})^2$ and eigenfunctions  $\ttheta'_r$  in Line~2 of Algorithm~\ref{alg:sampleProb}. At the same time, a pilot estimation of the projection operator $\hsP'_R = \sum_{r=1}^R \htheta'_r \otimes \htheta'_r$ is obtained. With these pilot estimates, we compute the score estimates $\hxi'_{nr} = \langle\ttheta'_r, x_n\rangle / \tsigma'_{r}$ for each sample in Line~4 of Algorithm~\ref{alg:sampleProb}. 
Then, the  probability estimator  $\hat{p}_n \propto \big(\sum_{r=1}^R (\hxi_{nr}')^2 + \|(I - \hsP'_R) x_n \|^2  / (\hsigma'_{R})^2\big)$ is computed for each sample  in Line~5. Algorithm~\ref{alg:sampleProb} will output the normalized $\{\hat{p}_n\}_{n=1}^N$. In Section~\ref{subsec:fpca:theoryprob}, it will be shown that the  probability $\{\hat{p}_n\}_{n=1}^N$ computed by Algorithm~\ref{alg:sampleProb} satisfies the lower bound requirement~\eqref{eqn:prob:NearlyExact}  with high probability.

\begin{remark}
	Two-step procedures can often be found in the literature of randomized algorithms. In~\cite{drineas2012fast}, the leverage scores are approximately computed by sketching the design matrix. In the work of~\cite{wang2018optimal,zhang2021optimal} for generalized linear regression, the sampling probabilities are also computed from  pilot regression estimators.
\end{remark}

\begin{remark} \label{remark:complexity:fpca} 
	The computational complexity of the above two-step procedure can be analyzed when each functional observation is digitally recorded as a high dimensional vector of length $L$. In this case, the eigenfunction can be obtained by directly applying SVD to the data matrix. The full sample computational complexity  is $\mathcal{O}(NL\min\{N,L\})$. On the other hand, the proposed  Algorithm~\ref{alg:sampleFPCA} costs a complexity of order $\mathcal{O}( NRL + CL\min\{C,L\})$, where  computing the sampling probabilities via Algorithm~\ref{alg:sampleProb} costs $\mathcal{O}(NRL)$. 
	
	Compared with the full sample computation,  the proposed algorithm has much smaller complexity when $C\ll N$ and $R\ll L$.  Specifically, the relation $R\ll L$ is a characteristic of functional data, whose intrinsic dimension is much  smaller than the ambient dimension.  Recall from Section~\ref{sec:introduction}, in practical applications such as astronomical spectral processing,   $L$ is in the  order of thousands, and $N$ has the magnitude of millions. Some work \citep{connolly1994spectral}  has found $R=3$  principal components are enough to capture most variability of the dataset.  Theorem~\ref{thm:projectionOperator}   suggests the  subsample size $C$ need not to be very large for accurate subsample estimation, due to the small  dimension parameter $R+\Delta_R$ for functional data. 
\end{remark}

\section{Theoretical Analysis of Randomized FPCA with FunPrinSS}\label{sec:fpca:theory}

This section develops a theoretical justification of Algorithm~\ref{alg:sampleFPCA} when the employed sampling probability satisfies~\eqref{eqn:prob:NearlyExact}. The derivation is based on the operator perturbation theory reviewed in Section~\ref{subsec:fpca:operator}, and  the main results follow in Section~\ref{subsec:fpca:maintheory}.  After that,   the theoretical justification of Algorithm~\ref{alg:sampleProb} is presented in Section~\ref{subsec:fpca:theoryprob}. We will show its estimated probability satisfies the lower bound requirement~\eqref{eqn:prob:NearlyExact} with high probability. Finally, we  discuss practical choice of the parameter $R$  in Section~\ref{subsec:fpca:chooseR}.

Our analysis   formulates estimating  the leading $R$ eigenfunctions as estimating the projection operator $\hsP_R = \sum_{r=1}^R \htheta_{r} \otimes \htheta_{r}$.  Algorithm~\ref{alg:sampleFPCA} delivers a randomized estimate $\tsP_R =  \sum_{r=1}^R \ttheta_r \otimes \ttheta_r$, which is computed from the sketched data.   We will focus on bounding their difference $ \tsP_R - \hsP_R$.  The theoretical analysis is carried out from an algorithmic perspective.  In other words, the concentration bounds are established conditional on the full dataset, and the derivation exploits  the  independence of the  subsampling draws  within our algorithms. The  full dataset can be almost arbitrary and the  theoretical bounds provide worst-case guarantee. Neither specific distribution nor independence assumption is imposed on the full dataset.  Only a single condition on the eigenvalues is required.

\begin{condition} \label{assumption:eigenvalue}
	The  empirical eigenvalues of $\hsC_{XX}$ satisfy the  following relation:
	$$\hsigma_{1}^2 > \hsigma_{2}^2 > \cdots > \hsigma_{R}^2 >
	\hsigma_{R+1}^2 \geq \hsigma_{R+2}^2 \geq \cdots.$$
\end{condition}

\begin{remark}
	A strictly positive  \textit{eigengap} $g_R  = \hsigma_{R}^2-\hsigma_{R+1}^2 > 0$ is the minimal assumption  for the identifiability  of $\hsP_R$.  To simplify presentation, we further assume the leading $R$ eigenvalues $\hsigma_{1}^2, \cdots, \hsigma_{R}^2$ are distinct with multiplicity one. If there exists repeated eigenvalues, the same theoretical results can be obtained using a similar argument but with complications in notation. Lastly,  note we always have $\hsigma_r^2  = 0$ for $r>N$, 	because the empirical  $\hsC_{XX}$ is computed from  $N$ samples, i.e. $\text{rank}\big(\hsC_{XX}\big)\le N$.
\end{remark}

\subsection{Operator Perturbation Theory} 
\label{subsec:fpca:operator}

We  set up the tools from  the operator perturbation theory~\citep{hsing2015theoretical} for further analysis of the randomized FPCA. For the operator $ \hsC_{XX}$,  its  \textit{resolvent}  is defined as $\sR_{\hsC_{XX}}(\eta): = ( \hsC_{XX}-\eta I)^{-1}$, where $I$ is the identity mapping. Let $\Gamma_R:=\{\eta\in \mathbb{C}:\ \mathrm{dist}(\eta,\,[\hsigma_R^2, \hsigma_1^2]) = g_R/2\}$ represent the boundary of a disk in a complex plane. Every point on  $\Gamma_R$ has equal distance $g_R/2=(\hsigma_{R}-\hsigma_{R+1})/2$ to the interval $[\hsigma_R^2, \hsigma_1^2]\subset \bbR$  on the positive real axis.
Then, the subspace projection operator can be expressed as a contour integration
$\hsP_R = -\frac{1}{2\pi i} \oint_{\Gamma_R} \sR_{\hsC_{XX}}(\eta) \intd \eta$.

Denote $\sE = \tsC_{XX} - \hsC_{XX}$ as the covariance operator approximation error. Then, using the result of~\cite{koltchinskii2016asymptotics}, the difference of the projection operators can be decomposed as  $\tsP_R - \hsP_R = L_R(\sE) + S_R(\sE) $, where
\begin{align}
	L_R(\sE): & = \frac{1}{2\pi i} \oint_{\Gamma_R} \sR_{\hsC_{XX}}(\eta) \sE \sR_{\hsC_{XX}}(\eta) \mathrm{d}\eta,\, \label{eqn:projerror:L}\\
	S_R(\sE):&=  \tsP_R - \hsP_R - L_R(\sE). \label{eqn:projerror:S}
\end{align}
Notice the first term $L_R(\sE)$ is  linear in $\sE$, while $S_R(\sE)$ is the residual  error.  Because $\Expect (\sE)=0$ where the expectation is taken with respect to the subsampling procedure, the  error $L_R(\sE)$ in~\eqref{eqn:projerror:L} also has zero expectation  $\Expect L_R(\sE)=0$.

Recall   $g_R  = \hsigma_{R}^2-\hsigma_{R+1}^2 (> 0)$ is the eigengap.
In the case  $\|\tsC_{XX}-\hsC_{XX}\|<g_R/3$, it holds that $$\max_{r}|\tsigma_r^2-\hsigma_r^2| \leq \|\tsC_{XX}-\hsC_{XX}\|<g_R/3.$$ This implies that the first $R$ eigenvalues of $\tsC_{XX}$ are strictly in the interior of the disk $\Gamma_R$. Therefore, the relation $\tsP_R = -\frac{1}{2\pi i} \oint_{\Gamma_R} \sR_{\tsC_{XX}}(\eta) \intd \eta$ holds for the subsampled estimator with $\sR_{\tsC_{XX}}(\eta) = ( \tsC_{XX}-\eta I)^{-1}$. In this case, the residual approximation error $S_R(\sE)$  has the  integral expression 
$$  S_R(\sE)= -\frac{1}{2\pi i} \oint_{\Gamma_R}
\sum_{k\ge 2}  [-\sR_{\hsC_{XX}}(\eta) \sE]^k \sR_{\hsC_{XX}}(\eta) \mathrm{d}\eta.$$
From the above, we can see that $L_R(\sE)$ can be regarded as the first order approximation error for 
$\tsP_R - \hsP_R$, while $S_R(\sE)$ as the higher order error. In fact, according to Lemma 2 of \cite{koltchinskii2016asymptotics}, the two terms $L_R(\sE)$ and $S_R(\sE)$ can be bounded as in the following lemma.

\begin{lemma} \label{lemma:lsbound}
	For the full sample covariance operator $\hsC_{XX}$ and its subsampled counterpart $\tsC_{XX}$, denote $\sE = \tsC_{XX} - \hsC_{XX}$ as their difference. Consider the decomposition $\tsP_R - \hsP_R = L_R(\sE) + S_R(\sE)$ with $L_R(\sE)$ and $S_R(\sE)$ defined in~\eqref{eqn:projerror:L} and~\eqref{eqn:projerror:S}, respectively.  It holds that
	\begin{equation*}
		(i)\quad \Vert L_R(\sE)  \Vert \le \big[ 1 + (\hsigma_1^2 - \hsigma_R^2)/(\pi g_R) \big] \times  \Vert \sE\Vert ; \quad 	(ii)\quad \Vert S_R(\sE)  \Vert \le K_R\left( \Vert \sE\Vert  / g_R  \right)^2. 
	\end{equation*}
	where $K_R := 15[1 + 2(\hsigma_1^2 - \hsigma_R^2)/ (\pi g_R )]$.
\end{lemma}

\subsection{Main Results} \label{subsec:fpca:maintheory}
We are now ready to  derive our main theorem for the randomized FPCA in Algorithm~\ref{alg:sampleFPCA} with the nearly exact sampling probability~\eqref{eqn:prob:NearlyExact}.  Before presenting the formal result, we  provide a heuristic justification  of  the proposed functional principal subspace sampling probability~\eqref{eqn:prob:FPCAProptoExact} and~\eqref{eqn:prob:NearlyExact}.   The decomposition in~\eqref{eqn:projerror:L} and~\eqref{eqn:projerror:S} suggests that, to minimize the loss $\| \tsP_R - \hsP_R\|$,  we need to have precise control over the first order error $L_R(\sE)$; at the same time, the magnitude of  $S_R(\sE)$ should also be contained with high probability.   Our proposed FunPrinSS probability backs up this intuition. In fact, the  probability~\eqref{eqn:prob:FPCAProptoExact} or~\eqref{eqn:prob:NearlyExact} can be interpreted as an importance sampling probability to control the first order error  $L_R(\sE)$. The  argument is based on the following fact.

\begin{lemma}  \label{lemma:pcaLinearSum}
	The first order error term  $L_R(\sE)$	can be expressed as  $L_R(\sE) = (1/C)\sum_{c=1}^C \sZ_c / (N\tp_c)$. Each  $\sZ_c$ is  related to one subsampled $\tx_c$ and that
	$$ \sZ_c =  \sum_{r=1}^R
	\sum_{s=R+1}^\infty \frac{\hsigma_r\hsigma_s}{\hsigma_r^2 - \hsigma_s^2} \times
	(\txi_{cr}\txi_{cs}) \times 
	\big[\htheta_r \otimes \htheta_s +
	\htheta_s \otimes \htheta_r\big].$$
	In the above,   $\txi_{cs}  = \langle \tx_c, \htheta_s\rangle/\hsigma_{s}$ when $\hsigma_s>0$ and $\txi_{cs}  =1$ when $\hsigma_s=0$.
\end{lemma}

\begin{remark}
	Note the quantity $\txi_{cs}$  associated with $\tx_{c}$ is  computed with the full sample eigenvalue $\hsigma_s$ and  the full sample eigenfunction $\htheta_{s}$. Suppose the $c$-th subsample is the $i_c$-th observation in the original  dataset (i.e., $\tx_{c} = x_{i_c}$), and the original  $x_{i_c}$ has full sample scores		$\hxi_{i_c,1},\hxi_{i_c,2},\cdots$.  Recall we have set $\tp_c=p_{i_c}$ as the sampling probability of  the original $i_c$-th observation.  Similarly,  we  denote for the $c$-th subsample $\txi_{cs}=\hxi_{i_c,s}$  ($s=1,2,\cdots$) as the corresponding full sample score. 
\end{remark}

Lemma~\ref{lemma:pcaLinearSum} reveals that $L_R(\sE)$ can be expressed as a summation of $C$ terms, and each summand $\sZ_c$ is related to one subsampled $\tx_c$.  For the importance sampling, 
the sampling probability is proportional to the size of each summand.  
We can show the proposed probability~\eqref{eqn:prob:FPCAProptoExact} or~\eqref{eqn:prob:NearlyExact} serve this purpose by computing the squared norm of each summand $\sZ_c$ in Lemma~\ref{lemma:pcaLinearSum} as
\begin{align} \label{eqn:ZcHSnorm}
	\| \sZ_c\|_{HS}^2  = 
	2 \sum_{r=1}^R &
	\sum_{s=R+1}^\infty \frac{\hsigma_r^2\hsigma_s^2}{[\hsigma_r^2 - \hsigma_s^2]^2} \cdot
	\txi_{cr}^2\txi_{cs}^2 \nonumber \\
	= &\quad 2 \sum_{r=1}^R
	\sum_{s=R+1}^\infty 
	\underbrace{ \frac{\hsigma_r^4}{[\hsigma_r^2 - \hsigma_s^2]^2}}_{f_{rs}^2} \cdot \txi_{cr}^2 (\hsigma_s^2\txi_{cs}^2 /\hsigma_r^2) .
\end{align}
Define $f_{rs}:= (1-q_{r,s})^{-1}$ with $q_{r,s} = \hsigma_s^2/\hsigma_r^2$ for the above equation.
For $r<s$, it is obvious  that $q_{r,s} \le 1$ and $f_{rs} \ge 1$.
Consider the special case where the leading $R$ principal components dominate the overall signal, i.e., $\hsigma_{R}^2 \gg \hsigma_{R+1}^2$. In this case, we have the approximation $f_{rs}  \approx 1$ for  $r(\le R)$ and $s(>R)$. 	It follows approximately that
\begin{align} 
	\| \sZ_c\|_{HS}^2 	&\approx 
	2 \sum_{r=1}^R
	\sum_{s=R+1}^\infty  \txi_{cr}^2 \big(\hsigma_s^2\txi_{cs}^2 /\hsigma_r^2\big) \stackrel{(i)}{\le} 2  \Big( \sum_{r=1}^R\txi_{cr}^2\Big) 
	\Big( \sum_{s=R+1}^\infty\hsigma_s^2 \txi_{cs}^2 /\hsigma_R^2 \Big) \nonumber\\
	&\stackrel{(ii)}{\le} \frac{1}{2} \Big( \sum_{r=1}^R\txi_{cr}^2
	+ \sum_{s=R+1}^\infty\hsigma_s^2 \txi_{cs}^2 /\hsigma_R^2 \Big)^2 \nonumber\\
	&=\frac{ 1}{2} \Big( \sum_{r=1}^R\txi_{cr}^2
	+  \|(I - \hsP_R) \tx_c \|^2  / \hsigma_{R}^2 \Big)^2. \label{eqn:intuitive}
\end{align}
In the above, (i) uses that $\hsigma_{r}^2 \ge \hsigma_{R}^2$ for $r\le R$, and (ii) uses the  inequality $4ab\le (a+b)^2$ for $a,b\in\bbR$. As expected,  the last line of~\eqref{eqn:intuitive}  indicates the approximate proportional relation $\tp_c^{\mathrm{Exact}} \propto \| \sZ_c\|_{HS}$ for the FunPrinSS probability~\eqref{eqn:prob:FPCAProptoExact}. This resembles the importance sampling $\tp_c \propto \| \tx_{c} \otimes \tx_{c} \|_{HS} = \|\tx_{c}\|^2$ for the covariance operator estimation~\eqref{eqn:subsampleCXX} in~\cite{he2020randomized}.	 From here, we find the proposed FunPrinSS probability  can be viewed as   an importance sampling probability  to control $L_R(\sE)$ when $\hsigma_{R}^2 \gg \hsigma_{R+1}^2$.

The above heuristic justification is developed in the  case $\hsigma_{R}^2 \gg \hsigma_{R+1}^2$. Under the general eigenvalue setting of Condition~\ref{assumption:eigenvalue}, to  control the first order term $L_R(\sE)$, we  define $$G_R := (1 - q_{R,R+1})^{-1} = \hsigma_{R}^2/(\hsigma_{R}^2-\hsigma_{R+1}^2),$$ and find
it as an upper bound for the ratio factor $f_{rs}$ in~\eqref{eqn:ZcHSnorm}. That is, for $r\le R < s$, we have
$$ f_{rs}  = (1- q_{r,s} )^{-1} \le  (1-  q_{R,R+1} )^{-1}= G_R.$$ This relation holds because  $f(q) = 1/(1-q)$ is an increasing function for $q\in(0,1)$.  Theorem~\ref{thm:projectionOperator} below formalizes this idea and provides theoretical guarantees for the subsampled  FPCA in Algorithm~\ref{alg:sampleFPCA} with the nearly exact sampling probability~\eqref{eqn:prob:NearlyExact}.  Note all the  probability bounds  in this work are derived conditional on the full dataset and with respect to the subsampling.

\begin{theorem} \label{thm:projectionOperator}
	Under Condition~\ref{assumption:eigenvalue}, define  $G_R = \hsigma_{R}^2/(\hsigma_{R}^2-\hsigma_{R+1}^2)$ and define 
	$$V = \max\{G_R^2 Z^2/(2\beta),\, Z\}\quad \text{and}\quad L=\max\{G_R Z /\sqrt{2},\, Z+1\},$$  with 
	$Z = \beta (R + \Delta_R)$.  Then,  for sampling probability 
	$\{p_n\}_{n=1}^N$ satisfying $p_n \ge p_n^{\text{Exact}} / \beta$
	with $\beta\ge 1$, and for $\epsilon$ satisfying $\epsilon\cdot C \ge \sqrt{CV} + L/3$, it holds that 
	\begin{equation} \label{thm:projectionOperator:bound}
		\|\tsP_R - \hsP_R\| \le \epsilon +K_R\cdot \hsigma_1^4 \epsilon^2/ g_R ^2,
	\end{equation}
	with probability at least $1 - 12(R+\Delta_R)\exp \Big( -\frac{C\epsilon^2/2}{ V+ L \epsilon /3} \Big).$ 
	
\end{theorem}

In the above theorem, $G_R$ can be viewed as a measurement of the eigengap $\hsigma_{R}^2-\hsigma_{R+1}^2$. The dimension $R+\Delta_R$ discussed  in~\eqref{eqn:fpca:dimension} also appears in multiple places above. On the right hand side of the error bound~\eqref{thm:projectionOperator:bound}, the first term $\epsilon$ bounds the first order error $L_R(\sE)$ in~\eqref{eqn:projerror:L}, while the second term $K_R\cdot \hsigma_1^4 \epsilon^2/ g_R ^2$ bounds  the high order $S_R(\sE)$ in~\eqref{eqn:projerror:S}.   

According to~\eqref{thm:projectionOperator:bound} of Theorem~\ref{thm:projectionOperator}, the subsample error  $\|\tsP_R - \hsP_R\|$ can be controlled within a  precision level $\epsilon$ with high probability. For a  large enough $C$ and with probability at least $0.9$, the bound~\eqref{thm:projectionOperator:bound} holds  for
the precision level $\epsilon$  of magnitude order $ \sO \big((R+\Delta_R) \log^{1/2}[120(R+\Delta_R)]/C^{1/2} \big)$. The precision level $\epsilon$ is inversely proportional to $C^{1/2}$, and is proportional to the  dimension parameter $R+\Delta_R$ with an additional logarithm factor. The nature of functional data usually indicates the dimension parameter $R+\Delta_R$ is small, though the ambient space dimension is infinite. This means  the subsample size $C$ need not be very large to control the subsample error toward a small accuracy level  $\epsilon$.

\subsection{Theoretical Justification of Algorithm~\ref{alg:sampleProb} }
\label{subsec:fpca:theoryprob}

 Algorithm~\ref{alg:sampleProb} can  be justified for providing the pilot probability estimate. We show that its estimated probability $\{\hat{p}_n\}_{n=1}^N$ remains close to the target~\eqref{eqn:prob:FPCAProptoExact} in the sense of~\eqref{eqn:prob:NearlyExact}. Its proof is derived in Section~\ref{proposition:problowerbound:proof} in the Appendix.

\begin{proposition} \label{proposition:problowerbound}
	Set $\Delta_0 = \mathrm{intdim}(\hsC_{XX})$,
	$\gamma_0 = \sqrt{(\Delta_0 +1)\log(120\Delta_0)}$  and $G_R = (1-q_{R,R+1})^{-1}$.	
	For large enough subsample size $C$,
	with probability at least $0.9$, the  pilot sampling probability $\{\hat{p}_n\}_{n=1}^N$  computed by Algorithm~\ref{alg:sampleProb} with $\alpha=0.5$ satisfy
	$$
	\min\Big\{\frac{p^{\mathrm{Exact}}}{\hat{p}_n} ,\, \frac{ \hat{p}_n}{p^{\mathrm{Exact}}} 	\Big\} \,\ge\, \frac{1 - \gamma_1 C^{-1/2}-\gamma_2 C^{-1}}{1 + \gamma_1 C^{-1/2}+\gamma_2C^{-1}},
	$$
	where  $\gamma_1 =35(1+G_R )(R+\Delta_R) +	8\hsigma_{1}^2 \gamma_0/ \hsigma_{R}^2$
	and $\gamma_2 = 32K_R\hsigma_{1}^6 G_R^2 /(g_R^2\hsigma_R^2)$.
\end{proposition}

One  implication of Proposition~\ref{proposition:problowerbound} is that, when the subsample size $C$ is large enough, it holds with high probability that $\hat{p}_n \ge p^{\mathrm{Exact}}/\beta$ for $\beta = \frac{1 + \gamma_1 C^{-1/2}+\gamma_2 C^{-1}}{1 - \gamma_1 C^{-1/2}-\gamma_2C^{-1}}$. This  corresponds to our definition of the nearly exact FunPrinSS probability~\eqref{eqn:prob:NearlyExact}. 
Besides, the value of $\beta(\ge 1)$ will converge to one as $C$ increases to infinity.

\subsection{Choice of $R$}
\label{subsec:fpca:chooseR}
	
For functional principal component analysis, the subspace dimension $R$ is usually chosen based on 
the fraction of  the  variance explained (FVE)   by the leading $R$ eigenfuctions. That is, $R$ is chosen such that the FVE is greater than a pre-specified threshold value (e.g. $90\%$). 
Given the full sample estimate $\hsP_R$ , its fraction of  variance  explained can be computed by
$$
\widehat{\text{FVE}} = \frac{\sum_{r=1}^{R} \hsigma_{r}^2}{\sum_{r=1}^{\infty} \hsigma_{r}^2}=\frac{\sum_{n=1}^{N} \| \hsP_R x_n\|^2}{\sum_{n=1}^{N} \|x_n\|^2};
$$
and correspondingly, the  fraction of  variance explained by the subsampled estimate $\tsP_R$ is
$\widetilde{\text{FVE}} = \sum_{n=1}^{N} \| \tsP_R x_n\|^2 / \sum_{n=1}^{N} \|x_n\|^2$.
As a direct consequence of Theorem~\ref{thm:projectionOperator}, we can bound the difference between the $\widehat{\text{FVE}}$ of the full sample estimator and the  
$\widetilde{\text{FVE}}$ of the subsampled estimator.

\begin{corollary} \label{corollary:pev}
Under the conditions of Theorem~\ref{thm:projectionOperator}, it holds that
	\begin{equation}  \label{eqn:pevdiff}
\big|\widehat{\mathrm{FVE}}  -\widetilde{\mathrm{FVE}} \big| \le \epsilon +K_R\cdot \hsigma_1^4 \epsilon^2/ g_R ^2,
\end{equation}
with probability at least $1 - 12(R+\Delta_R)\exp \big( -\frac{C\epsilon^2/2}{ V+ L \epsilon /3} \big).$
\end{corollary}
\begin{proof}
With  Theorem~\ref{thm:projectionOperator}, the result~\eqref{eqn:pevdiff} is straightforward by observing 
	\begin{align*}
		&	\Big|\frac{\sum_{n=1}^{N} \| \hsP_R x_n\|^2}{\sum_{n=1}^{N} \|x_n\|^2} -
		\frac{\sum_{n=1}^{N} \| \tsP_R x_n\|^2}{\sum_{n=1}^{N} \|x_n\|^2}\Big| \\
		= &\Big|	\frac{\sum_{n=1}^{N} \langle x_n, (\hsP_R-\tsP_R) x_n\rangle}{\sum_{n=1}^{N} \|x_n\|^2}\Big| 
		\le  \epsilon +K_R\cdot \hsigma_1^4 \epsilon^2/ g_R ^2.
	\end{align*}
The last inequality employs $|\langle x_n, (\hsP_R-\tsP_R) x_n\rangle| \le \|\hsP_R-\tsP_R\| \cdot \|x_n\|^2$.
\end{proof}

Corollary~\ref{corollary:pev} indicates  $\widetilde{\text{FVE}}$ is a reliable estimate of 
$\widehat{\text{FVE}}$. When applying the randomized algorithm, we can use the criterion  $\widetilde{\text{FVE}}$ to select $R$.  After obtaining the subsampled estimate $\tsP_R$,  computing
$\widetilde{\text{FVE}} $ requires one additional scan over the whole dataset and have a computational complexity $\sO(NRL)$. This complexity is no more than that of Algorithm~\ref{alg:sampleProb}, see Remark~\ref{remark:complexity:fpca}.

\section{Extenstion: Randomized Functional Linear Regression}\label{sec:flr}

A considerable number of   functional data models rely on functional principal components for dimension reduction. 
One of these models is functional linear regression (FLR) with scalar response. In this section, we extend the results from Section~\ref{sec:fpca} and Section~\ref{sec:fpca:theory}  to FLR and develop its theoretical guarantees.

\subsection{Randomized FLR} \label{subsec:flr:algorithm}
FLR is a generalization of the classical linear regression to the functional setting.  For  a functional predictor  $x_n\in \sH_X$ and a scalar response $Y_n\in\bbR$, FLR assumes the relation $Y_n = \alpha + \langle x_n, \Psi\rangle + \epsilon_n$ for $n=1,\cdots, N$. In this model,  $\alpha\in\bbR$  is  the intercept term, $\Psi\in \sH_X$ is the regression function and $\epsilon_n$ is some random error with zero mean. When both $x_n$ and $\Psi$ belong to $L_2(T)$ for some compact interval $T$, the regression model is  written as $Y_n = \alpha + \int_T x_n(t) \Psi(t) \intd t + \epsilon_n$.
Given any estimator $\hPsi$, the  intercept is directly available via $\hat{\alpha} = (1/N) \sum_{n=1}^N (Y_n - \langle x_n, \hPsi\rangle)$. In the following, we  focus on the estimation of the regression function $\Psi$ by assuming that $Y_n$ and $x_n$ have been centered with zero mean, such that $\hat{\alpha} = 0$.  

Considerable amount of works has been proposed for estimating $\Psi$ \citep[e.g.][]{yao2005regression,yuan2010reproducing}. One of the  commonly used estimators for $\Psi$ is based on the truncated inverse of  the empirical covariance operator $\hsC_{XX} = \frac{1}{N}\sum_{i=1}^N x_i\otimes x_i$; see  \cite{hall2007methodology, cardot2007clt}. 	Suppose the covariance operator admits the eigenvalue-eigenfunction decomposition	$\hsC_{XX} =  \sum_{r=1}^\infty \hsigma_r^2\cdot  \htheta_r \otimes \htheta_r$. 	Its rank-$R$ \textit{truncated inverse} is defined as 	$\hsC_{XX}^+ = \sum_{r=1}^R (1/\hsigma_r^2)\cdot  \htheta_r \otimes \htheta_r$. Based on this, the regression function is estimated by  $\hPsi = \hsC_{XX}^{+} \hat{z}$ with $\hat{z} = (1/N) \sum_{n=1}^N Y_n x_n$.

\begin{algorithm}[t]
	\caption{Base Randomized FLR\label{alg:sampleFLR}}
	\textbf{Input}: Dataset $\{(x_n, Y_n)\}_{n=1}^N$; sampling probability $\{p_n\}_{n=1}^N$;	subsample size $C$. \\
	\textbf{Output}: the regression function $\widetilde{\Psi}$ 
	\begin{algorithmic}[1]
		\FOR{$c = 1,\cdots, C$}
		\STATE Sample $(\tx_c, \tilde{Y}_c)$ from $\{(x_n, Y_n)\}_{n=1}^N$ and get 
		$\tp_c$ with  $\Prob\big((\tx_c, \tilde{Y}_c, \tp_c)= (x_n, Y_n, p_n)\big) = p_n.$
		\ENDFOR
		\STATE Compute $\tsC_{XX} = \frac{1}{C} \sum_{c=1}^C \frac{1}{N\tp_c}\tx_c\otimes \tx_c $ and $\tilde{z} = \frac{1}{C} \sum_{c=1}^C \frac{1}{N\tp_c}\tilde{Y}_c \tx_c$.
		\STATE Compute $\widetilde{\Psi} = \tsC^+_{XX} \tilde{z} $
	\end{algorithmic}
\end{algorithm}

To reduce the cost of computing $\hPsi$ from the full sample, 
we adapt  the randomized algorithm to estimate it based on a subset of the data. The algorithm will sample a subset $\{(\tilde{x}_c, \tilde{Y}_c)\}_{c=1}^C$ of observation pairs  from the full sample $\{(x_n, Y_n)\}_{n=1}^N$. The sampling is taken  with replacement and  according to some probability distribution $\{p_n\}_{n=1}^N$. Given the subsampled  pairs, we can compute two quantities,
$\tsC_{XX} = \frac{1}{C} \sum_{c=1}^C \frac{1}{N\tp_c}\tx_c\otimes \tx_c $ and $\tilde{z} = \frac{1}{C} \sum_{c=1}^C \frac{1}{N\tp_c}\tilde{Y}_c \tx_c$,  which are unbiased estimators of $\hsC$ and $\hat{z}$, respectively. Finally, the full sample regression function is estimated by its subsampled counterpart $\widetilde{\Psi} = \tsC^+_{XX} \tilde{z}$. This procedure is listed in Algorithm~\ref{alg:sampleFLR}. 

In Algorithm~\ref{alg:sampleFLR}, notice the truncated inverse $\tsC^+_{XX}$ is computed from the eigenvalue-eigenfunction decomposition of the subsampled $\tsC_{XX}$. It is evident that the accuracy of the randomized estimator $\widetilde{\Psi}$ relies on the precision of the  principal subspace estimated by $\tsC_{XX}$. Our proposed FunPrinSS probability~\eqref{eqn:prob:FPCAProptoExact} is a proper choice for the randomized FLR to capture the principal subspace information.  

In practice, the	FunPrinSS probability~\eqref{eqn:prob:FPCAProptoExact} is unknown without full sample computation. We can carry out a two-step procedure similarly to that of Section~\ref{subsec:fpca:implementation}. Firstly, the FunPrinSS probability is estimated via Algorithm~\ref{alg:sampleProb}. Secondly, the estimated probability is  supplied to Algorithm~\ref{alg:sampleFLR} to obtain $\widetilde{\Psi}$.  As noted in Proposition~\ref{proposition:problowerbound}, the  probability  obtained in the first step is close to the target FunPrinSS probability with high probability when the subsample size $C$ is large enough. 

Application of  the randomized FLR algorithm also requires the specification of the principal subspace dimension $R$. Practitioners can choose $R$ such that the fraction of variance explained (FVE) by the  principal subspace is close to $100\%$. In this way, most of the variability of the predictor $x$ is captured and the remaining variability can be regarded as being negligible.

\begin{remark} \label{remark:complexity:flr}
When  the functional observation is recorded as a high dimensional vector of length $L$, 
the computational complexity of such two-step procedure can be similarly analyzed as in Remark~\ref{remark:complexity:fpca}. 
 In this case, the computational complexity of FLR with full sample is $\mathcal{O}(NL\min\{L,N\}+ L(R+N))$. The proposed two-step procedure for randomized FLR requires  $\mathcal{O}( CL\min\{C,L\} + (C+NR)L )$.  	Compared with the full sample computation,  the proposed algorithm has much smaller complexity when $C\ll N$ and $R\ll L$. 
\end{remark}

\subsection{Analysis of Randomized FLR with FunPrinSS}\label{subsec:flr:theory}

We continue to theoretically analyze the subsampled regression estimator $\widetilde{\Psi}$ obtained from Algorithm~\ref{alg:sampleFLR} with the  nearly exact sampling probability~\eqref{eqn:prob:NearlyExact}.  Toward this target, we express the full sample estimator $\hPsi$ and the subsampled estimator $\widetilde{\Psi}$  in compact expressions. These expressions resemble the commonly used ones  for the classical linear regression. First of all, define an operator $\sT: \sH_X\to \bbR^N$ such that  an element $u\in\sH_X$ is mapped to
\begin{equation}
	\sT u = (\langle x_1, u\rangle, \cdots, \langle x_N, u\rangle)^T. \label{eqn:sTMap}
\end{equation}
Let  $\mY = (y_1,\cdots, y_N)^T$ and $\vepsilon = (\epsilon_1, \cdots, \epsilon_N)^T$ be  $N$-dimensional vectors of the responses and the residuals, respectively. Then, the regression model for all $N$  observations can be expressed as $\mY = \sT \Psi + \vepsilon$. The operator $\sT$ plays a similar role as the design matrix  for classical linear regression.

Recall from the end of Section~\ref{sec:introduction}, the Euclidean space $\bbR^N$ is equipped with the normalized inner product $\langle\cdot, \cdot\rangle_N$.  We  denote $\sT^*:\bbR^N\to \sH_X$ as 
the adjoint operator of $\sT$. For any $\mathbf{a}\in\bbR^N$ and $u\in\sH_X$, it holds that
$$
\langle \sT^* \mathbf{a}, u\rangle \stackrel{(i)}{=}\langle \mathbf{a}, \sT u\rangle_N  \stackrel{(ii)}{=} \frac{1}{N} \sum_{n=1}^N a_n \langle x_n, u\rangle
= \big\langle \sum_{n=1}^N a_n x_n/N,\ u \big\rangle.
$$
In the above, (i) uses the definition of adjoint operator, and (ii) employs~\eqref{eqn:sTMap}.
From the above equation, we find that the adjoint operator $\sT^*$ maps any $\mathbf{a}\in\bbR^N$ to $(1/N)\sum_{n=1}^N a_n x_n $. Moreover, for any $u\in\sH_X$, we can find that
$$
\sT^*\sT u = \sT^*\begin{pmatrix}
	\langle x_1, u \rangle\\\vdots \\\langle x_N, u \rangle
\end{pmatrix}
=\frac{1}{N} \sum_{n=1}^N \langle x_n, u \rangle x_n = \hsC_{XX} u.
$$
The above equation implies that $\hsC_{XX} = \sT^*\sT$. Similarly, we can verify $\hat{z} = \sT^*\mY$. Combining these results, the full sample estimator of the regression function can be equivalently expressed as $\hPsi =\hsC_{XX}^{+} \hat{z} = (\sT^*\sT)^+(\sT^*\mY)$.

To find a closed form expression for the subsampled $\widetilde{\Psi}$,
we represent the subsampling process by an operator $\sD:\bbR^N\to\bbR^C$.  As a mapping between Euclidean spaces, $\sD$ can be  represented by a matrix $\mD \in\bbR^{C\times N}$. The $(c,n)$-th entry of the matrix $\mD$ is $D_{cn}=1/\sqrt{Np_n}$ if $x_n$ is selected as the $c$-th subsample; all the other entries are zero. The adjoint operator of $\sD$ is denoted as
$\sD^* :\bbR^C\to\bbR^N$. Note the Euclidean space $\bbR^C$ is also equipped with the normalized inner product $\langle\cdot, \cdot\rangle_C$, such that
for $\mathbf{a}', \mathbf{b}' \in\bbR^C$ we have $\langle\mathbf{a}', \mathbf{b}'\rangle_C = (1/C)\sum_{c=1}^C a'_cb'_c$. Now for any  $\mathbf{a}' \in\bbR^C$
and  $ \mathbf{b} \in\bbR^N$, we have
$$
\langle \mathbf{a}', \mD \mathbf{b}\rangle_C = \frac{1}{C} \langle \mathbf{a}', \mD \mathbf{b}\rangle_2  = \frac{1}{C} \langle \mD^T\mathbf{a}',  \mathbf{b}\rangle_2 =
\langle (N/C)\mD^T \mathbf{a}', \mathbf{b} \rangle_N,
$$
where $\langle \cdot, \cdot\rangle_2$ is the  classical (unnormalized) inner product for Euclidean space, which is the summation of element-wise products. From the above equation, we can see that the matrix representation of $\sD^*$ is $(N/C)\mD^T$, which is the transpose of $\mD$ and multiplied by a ratio factor $N/C$. With some computation, we can verify that for any $u\in\sH_X$
$$
\sT^*\sD^*\sD\sT u = (N/C) \sT^*\mD^T\mD\begin{pmatrix}
	\langle x_1, u \rangle\\\vdots \\\langle x_N, u \rangle
\end{pmatrix} = \frac{1}{C} \sum_{c=1}^C \frac{1}{N\tp_c}
\langle \tx_c,u\rangle \tx_c = \tsC_{XX} u.
$$
It follows that $\tsC_{XX} = \sT^*\sD^*\sD\sT$, and
$\tilde{z} = \sT^*\sD^*\sD\mY$ can be verified due to a similar reasoning. 
This leads to a compact expression for the randomized regression estimator 
$\widetilde{\Psi} =  \tsC^+_{XX} \tilde{z} = ( \sT^*\sD^*\sD\sT)^+ ( \sT^*\sD^*\sD\mY)$. The above results are summarized in Lemma~\ref{lemma:reg:compact} below, which gives an alternative expression of the full sample estimator and the subsampled estimator obtained via Algorithm~\ref{alg:sampleFLR}.

\begin{lemma} \label{lemma:reg:compact}
	The full sample regression estimator can be expressed as $$\hPsi = (\sT^*\sT)^+(\sT^*\mY),$$ and 
	the subsampled estimator  can be expressed as $$\widetilde{\Psi} = ( \sT^*\sD^*\sD\sT)^+ ( \sT^*\sD^*\sD\mY).$$
\end{lemma}

Now we show that the subsampled estimator provides a good approximation of the full sample estimator under suitable criterion. Our aim is to  characterize the closeness between the full sample estimator $\hPsi$ and the subsampled estimator  $\widetilde{\Psi}$. In the literature of functional linear regression, strong consistency of the estimated regression function is not generally possible. The estimators are only guaranteed to converge in some weak norm such as the prediction error on a new sample \citep[see][]{cardot2007clt, yuan2010reproducing}.
We choose to measure the difference $\widetilde{\Psi}-\hPsi$ based on the  full sample prediction error  $\|\sT\widetilde{\Psi}-\sT\hPsi \|_N^2
= \frac{1}{N} \sum_{n=1}^N \langle x_n,  \widetilde{\Psi}-\hPsi \rangle^2$. When only a fraction of the whole dataset is employed to estimate $\widetilde{\Psi}$, this metric partially gauges the quality of 
$\widetilde{\Psi}$ by the samples out of the subsampled (training) dataset. 
Our theoretical target boils down to the analysis of the error $ \sT\widetilde{\Psi}-\sT\hPsi$.  Lemma~\ref{lemma:regErrorDeco} below provides a decomposition of the subsampling error, the proof of which can be found in Section~\ref{appendix:flr:proof} in the Appendix. Recall that $x_n$ can be represented as $x_n = \sum_{r = 1}^{\infty}\hsigma_r\hxi_{nr}\htheta_r$ due to the Karhunen-Lo\`{e}ve expansion~\eqref{eq:KL}. The following decomposition involves an operator $\sT^+:\bbR^N\to \sH_X$, which maps $\mathbf{a}\in\bbR^N$  to $
\sT^+\mathbf{a} = \frac{1}{N}\sum_{n=1}^N a_n x_n^+ \in\sH_X$, with 
$x_n^+ = \sum_{r=1}^R \hxi_{nr} \htheta_r/\hsigma_r$  corresponding to each $x_n$.

\begin{lemma} \label{lemma:regErrorDeco}
	The subsampling error in full sample prediction has the following decomposition
	\begin{align}
	 \sT\widetilde{\Psi} - \sT\hPsi  \nonumber 
		=\underbrace{ \sT(\tsP_R-	\hsP_R)\sT^+\mY }_{\mathrm{I}} &
	 	+ \underbrace{ \sT(\tsC^+_{XX} - \hsC^+_{XX})\sT^*(\sD^*\sD)\mY^{\perp} }_{\mathrm{II}} \\
	&\qquad	+  \underbrace{ \sT \hsC^+_{XX}\sT^*(\sD^*\sD)\mY^{\perp}}_{\mathrm{III}} . \label{eqn:regression:errorterm23}
	\end{align}
	In the above, $\mY^\perp = \mY -\sT \hPsi$
	is the residual of the full sample estimator. In addition, 
	we have that 
	\begin{equation} \label{eqn:regression:perp}
		\sT \hsC_{XX}^+\sT^* \mY^{\perp} = \vzero.
	\end{equation}
\end{lemma}

Lemma~\ref{lemma:regErrorDeco} shows that the subsampling error $ \sT\widetilde{\Psi} -\sT\hPsi $ can be decomposed into three interpretable terms. The first term (I)  accounts for the subsampling estimation error of the principal subspace $\tsP_R-\hsP_R$. The second term (II) arises  due to the difference between
the truncated inverses  $\tsC^+_{XX} - \hsC^+_{XX}$. 
The third term (III) can be interpreted as follows.
Notice we have $\sT \hsC_{XX}^+\sT^* \mY^{\perp} = \vzero$ by the last conclusion of Lemma~\ref{lemma:regErrorDeco}. This relation holds due to the similar reason  that the residual is orthogonal to the column space of the design matrix for classical linear regression. In fact, \eqref{eqn:regression:perp} holds because our $\mY^\perp $ is 
orthogonal to the space spanned by the first $R$ score vectors , i.e., $\text{span}(\vxi_1,\cdots, \vxi_R)$
with  $\vxi_{r} = (\hxi_{1r},\cdots, \hxi_{Nr})^T$. This orthogonality is violated due to subsampling, which explains the appearance of the last term (III) in~\eqref{eqn:regression:errorterm23}. With Lemma~\ref{lemma:regErrorDeco} and the operator perturbation expansion techniques in Section~\ref{subsec:fpca:operator}, we can establish the following   result for the randomized FLR. 

\begin{theorem} \label{thm:regression}
	Under Condition~\ref{assumption:eigenvalue},
	set $Z_1 =\hsigma_1^6K_R / ( g_R ^2\hsigma_R^2) $ and $Z_2 = \beta(R+\Delta_R)$. In addition, let $V=  2 +   (2+G_R^2)Z_2^2/\beta$ and 
	$L= \sqrt{2R}+ ( \sqrt{2}+ G_R/\sqrt{2} )Z_2$.
	Suppose $C$ subsamples are obtained according to the  probability 
	$\{p_n\}$ satisfying
	$p_n \ge  p_n^{\mathrm{Exact}}/\beta$
	with some $\beta\ge 1$. Then with probability at least
	\begin{align} \label{thm:regression:eventprob}
		&1 - 16(R + \Delta_R)  \exp\Big(-\frac{C\epsilon^2}{V +L\epsilon/3 }\Big)
		- \frac{3Z_2}{\epsilon^2 C},
	\end{align}
	it holds that
	\begin{align}
		\| \sT\widetilde{\Psi} - \sT\hPsi  \|_N & \le 
		\underbrace{ \big( \epsilon+ \hsigma_RZ_1/\hsigma_1 \epsilon^2\big)  \|\mY \|_N }_{a} 
		+\underbrace{4\big(\epsilon+  Z_1\epsilon^2 \big)   \|\mY^{\perp}\|_N}_{b}.
		\label{thm:regression:bound}
	\end{align}
	In the above, $\epsilon$ satisfy $\epsilon  \ge   \sqrt{V/C} + L/(3C)$ and $\epsilon\le  \min\{1, g_R/3\}$. 
\end{theorem}

Theorem~\ref{thm:regression} directly implies that, given a  subsample size $C$, we can control  $\|\sT\widetilde{\Psi}- \sT\hPsi \|_N$ within a  precision level $\epsilon$ in~\eqref{thm:regression:bound} with high probability. 	For a  large enough $C$ and with probability at least $0.9$, the bound~\eqref{thm:regression:bound} holds for the precision level $\epsilon$  of magnitude order $\sO\big((R+\Delta_R) \log^{1/2}[320(R+\Delta_R)]/C^{1/2}\big)$.

The right hand side of~\eqref{thm:regression:bound} contains two parts.  The first term~(a) bounds the subspace difference~(I) in~\eqref{eqn:regression:errorterm23}, and it is an error term relative to $\|\mY\|_N$. The second term~(b) bounds~(II) and~(III) in~\eqref{eqn:regression:errorterm23} together, and its magnitude is relative to  $\|\mY^\perp\|_N$.  

To help get a better understanding of the two terms on the right hand side of ~\eqref{thm:regression:bound}, we can  relate them to similar quantities in finite dimensional problems. For the randomized classical linear regression \citep{drineas2011faster}, the error bound can be simplified to one term containing $\|\mY^\perp\|_N$ . This is because their corresponding predictors $x_n$ span a finite dimensional space. Their work agrees with our result when $\mathrm{dim}(\sH_X)$  is finite and relatively small, and we directly take $R$ to be $\mathrm{dim}(\sH_X)$ without truncation. Indeed, with large enough subsample size $C$, the finite dimensional space can be easily recovered, and   $\hsP_R=\tsP_R$ is guaranteed with high  probability.  This implies that the term~(I) will  vanish in~\eqref{eqn:regression:errorterm23} for linear regression, and consequently   the first term~(a) will disappear in~\eqref{thm:regression:bound}.  The result~\eqref{thm:regression:bound} would become  an error bound whose magnitude is only related to the size of the residual, $\|\mY^\perp\|_N$. 

When $\mathrm{dim}(\sH_X)$  is finite and $R<\mathrm{dim}(\sH_X)$, our  randomized FLR  corresponds  to  randomized multivariate principal component regression (PCR). In the literature,  additive error bound relative to $\|\mY\|_2$  has been established for the randomized PCR   \citep[see. Theorem~5 of][]{mor2019sketching}.  Our theoretical bound~\eqref{thm:regression:bound} reveals different sources of uncertainties  whose sizes are proportional to $\|\mY^\perp\|_N$ or $\|\mY\|_N$. This is a consequence of the  subsample error decomposition  in Lemma~\ref{lemma:regErrorDeco}.

Notice the precision parameter $\epsilon$ has an upper bound $g_R/3$ in the theorem. 	This is because we need to control the difference between the truncated inverses $\tsC^+_{XX} - \hsC^+_{XX}$. The difference is generally unbounded unless 	we require $\|\tsC_{XX} - \hsC_{XX}\| $ to be smaller than the eigengap $g_R$. Thereby, the theorem is derived conditional on the event
$\|\tsC_{XX} - \hsC_{XX}\| \le (1/3)  g_R$, 	which holds at least with the stated probability~\eqref{thm:regression:eventprob}.

\section{Simulation Study} \label{sec:simu}

This section conducts a simulation study for assessing the  randomized functional principal component analysis in Algorithm~\ref{alg:sampleFPCA} and the randomized functional linear regression in Algorithm~\ref{alg:sampleFLR}. Specifically, for these randomized algorithms, we will consider three sampling probabilities: 1) the uniform sampling probability (UNIF), $p_n = 1/N$; 2) the importance sampling probability \citep[IMPO,][]{he2020randomized}, $p_n  = \| x_n\|^2 / \sum_{m = 1}^N\Vert x_m\Vert^2$; 3) the proposed  functional principal subspace sampling probability (FunPrinSS)  estimated by Algorithm~\ref{alg:sampleProb} with $\alpha = 0.5$.  These sampling probabilities will be supplied to the randomized algorithms, and the accuracy of the subsampled estimation will be compared.

\subsection{Data Generation}
\label{sec:simu:data}

In each replication of the simulation study, we will generate one synthetic dataset $\{x_n\}_{n = 1}^N$ as follows. Each observation $x_n$ is a continuous function defined over the compact interval $T = [0,1]$. In particular, it is generated according to $x_n(t) = \sum_{r = 1}^{50} \sigma_r\xi_{nr}\theta_r(t)$, where the $\theta_r$ is set as the  Fourier function which is $\theta_r(t)=\sqrt{2}\sin(2\pi r t)$ when $r$ is odd, and  $\theta_r(t)=\sqrt{2}\cos(2\pi r t)$ when $r$ is even. The scores $\xi_{nr}$'s are independent random variables with zero  mean  and unit variance.  	For the random function $x_n$, its covariance operator $\sC_{XX}$  is determined by the integral 	$(\sC_{XX} u)(t') = \int_0^1 k(t, t') u(t) \intd t$ for any $u\in L_2([0,1])$, with integral kernel $k(t, t') = \sum_{j=1}^p \sigma_r^2 \theta_r(t) \theta_r(t')$. We can see  the $r$-th eigenvalue and eigenfunction of $\sC_{XX}$ are $\sigma_{r}^2$ and $\theta_r$, respectively.

The settings  to specify the eigenvalues $\{\sigma_r^2\}_{r = 1}^{50}$ and the distribution of $\xi_{nr}$ is described in Table~\ref{tab:simu:paras}.  
The eigenvalues  $\{\sigma^2_r\}_{r = 1}^{50}$ are specified as two cases: \texttt{Exponential Decay (ED)} and \texttt{Polynomial Decay (PD)}. The decline rate of the PD eigenvalues is slower than the ED eigenvalues. Hence, the eigengap for the PD setting is smaller, and the estimation of the subspace spanned by the first R eigenfunctions is more difficult. 

Similar to the simulation setting of~\cite{ma2015statistical}, the random scores, $\xi_{nr}$'s, are independently drawn from one of the three distributions: \texttt{Nearly Uniform (NU)} distribution, for which $\xi_{nr}$'s are drawn independently from standard normal distribution $N(0,1)$;  \texttt{Moderately Nonuniform (MN)} distribution,  where the scores $\xi_{nr}$'s are independently generated by t-distribution with three degree of freedom and scaled to unit variance; and  \texttt{Very Nonuniform (VN)} distribution, where $\xi_{nr}$ follows t-distribution with one degree of freedom and scaled to unit variance. Notice among these three distributions, the score distribution for \texttt{NU} has the lightest tail, while has the heaviest tail for \texttt{VN}.  As a result, the functions in a synthetic dataset from \texttt{NU} will behave more homogeneously.  On the other hand,  the synthetic datasets from \texttt{MN}  or \texttt{VN}  are likely to contain some observations with extreme magnitude. \cite{ma2015statistical} note  observations in the VN setting could have drastically different leverages, adversely affecting the naive uniform sampling.

\begin{table}[t]
	\centering
	\caption{Simulation Parameter Specification.  \label{tab:simu:paras}}
	\resizebox{1\textwidth}{!}{
		\begin{tabular}{c|c|c}
			{\bf Parameter} & {\bf Type} & {\bf Description} \\
			\hline
			\hline
			\multirow{2}{*}{Eigenvalue $\{\sigma^2_r\}_{r = 1}^{50}$} & \texttt{Exponential Decay (ED)} & $\sigma^2_r = c\cdot\kappa^{r}$ with $c=2^{51}$ and $\kappa=0.5$\\
			& \texttt{Polynomial Decay (PD)} &  $\sigma^2_r = c\cdot r^{-\kappa}$ with $c=100$ and $\kappa=1.5$\\
			\hline
			\multirow{2}{*}{Score $\xi_{nr}$'s} & \texttt{Nearly Uniform (NU)} & standard normal distribution $N(0,1)$\\
			& \texttt{Moderately Nonuniform (MN)} & t-distribution with three degree of freedom and scaled to unit variance\\
			& \texttt{Very Nonuniform (VN)} & t-distribution with one degree of freedom and scaled to unit variance\\
			\hline
			
		\end{tabular}
	}
\end{table}

By combining two different choices of eigenvalues (ED or PD) and three different choices of score distributions (NU, MN, or VN), we have totally six distinct ways to generate a full sample dataset $\{x_n\}_{n=1}^N$ of size $N = 10,000$. In each replicate of the simulation, one dataset will be generated in a specified way.  The procedure is repeated  $1000$ times to report the average loss and associated standard errors based on the subsample size $C \in \{100,300,500,700,1000,3000,5000,7000\}$.	We will apply the randomized algorithms to estimate the leading functional principal components in Section~\ref{sec:simu:fpca}  and solve the functional linear regression  in  Section~\ref{sec:simu:flr}.  We present here the results with the ED eigenvalue setting. See Section~\ref{sec:supp:simulation} in the Appendix for the results with the PD eigenvalue setting.

\begin{figure}[t]
	\centering
	\includegraphics[width = 0.9\textwidth, height = 0.45\textwidth]{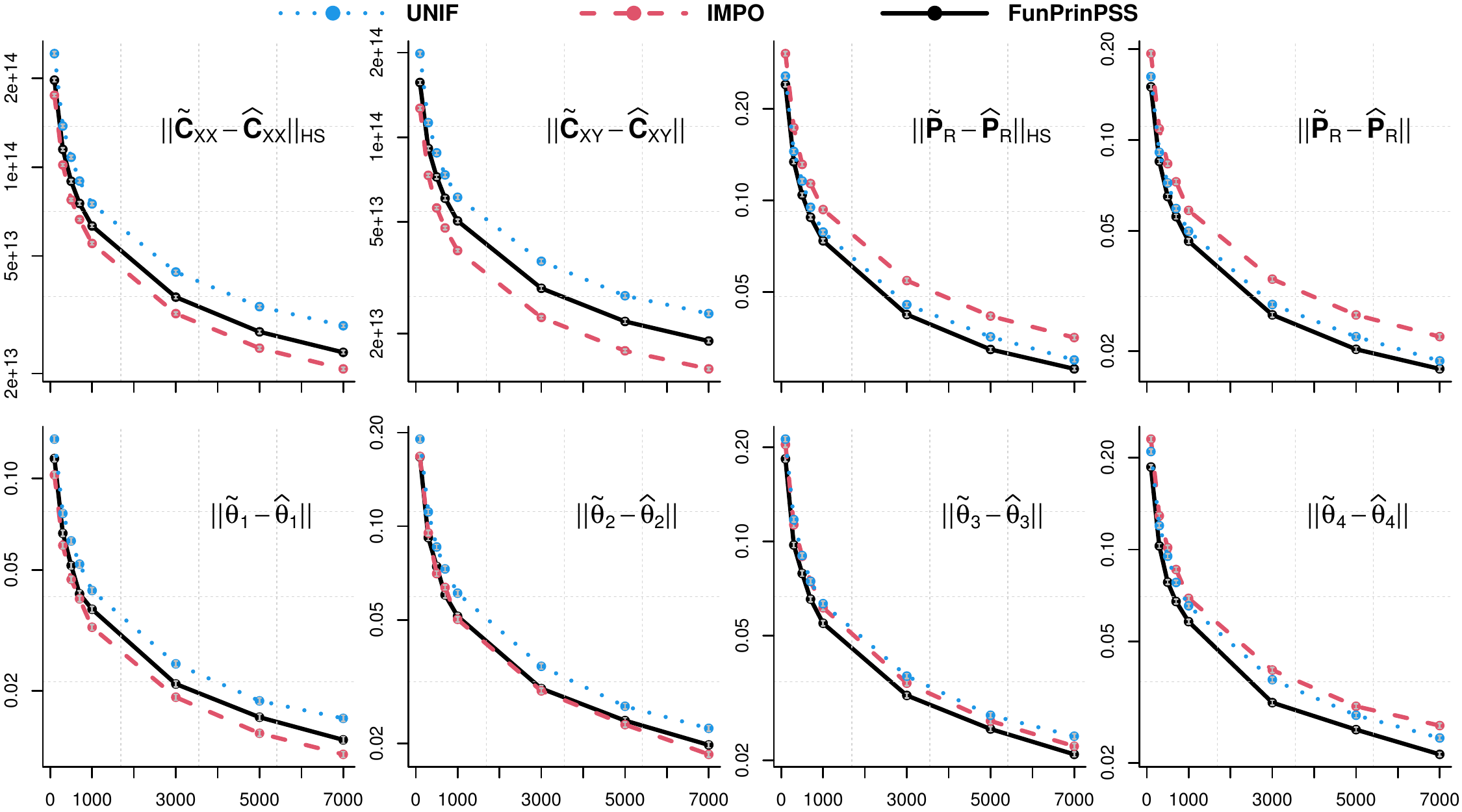} 
	\caption{Randomized covariance operator and FPCA estimation for the ED and NU setting. The vertical axes represent comparison metrics in  $\log_{10}$ scale. The blue dotted, red dashed  and black solid lines correspond to the results of the  UNIF, IMPO and FunPrinSS sampling, respectively.}\label{fig:ErrorNUexp}
\end{figure}
\begin{figure}[ht]
	\centering
	\includegraphics[width = 0.9\textwidth, height = 0.45\textwidth]{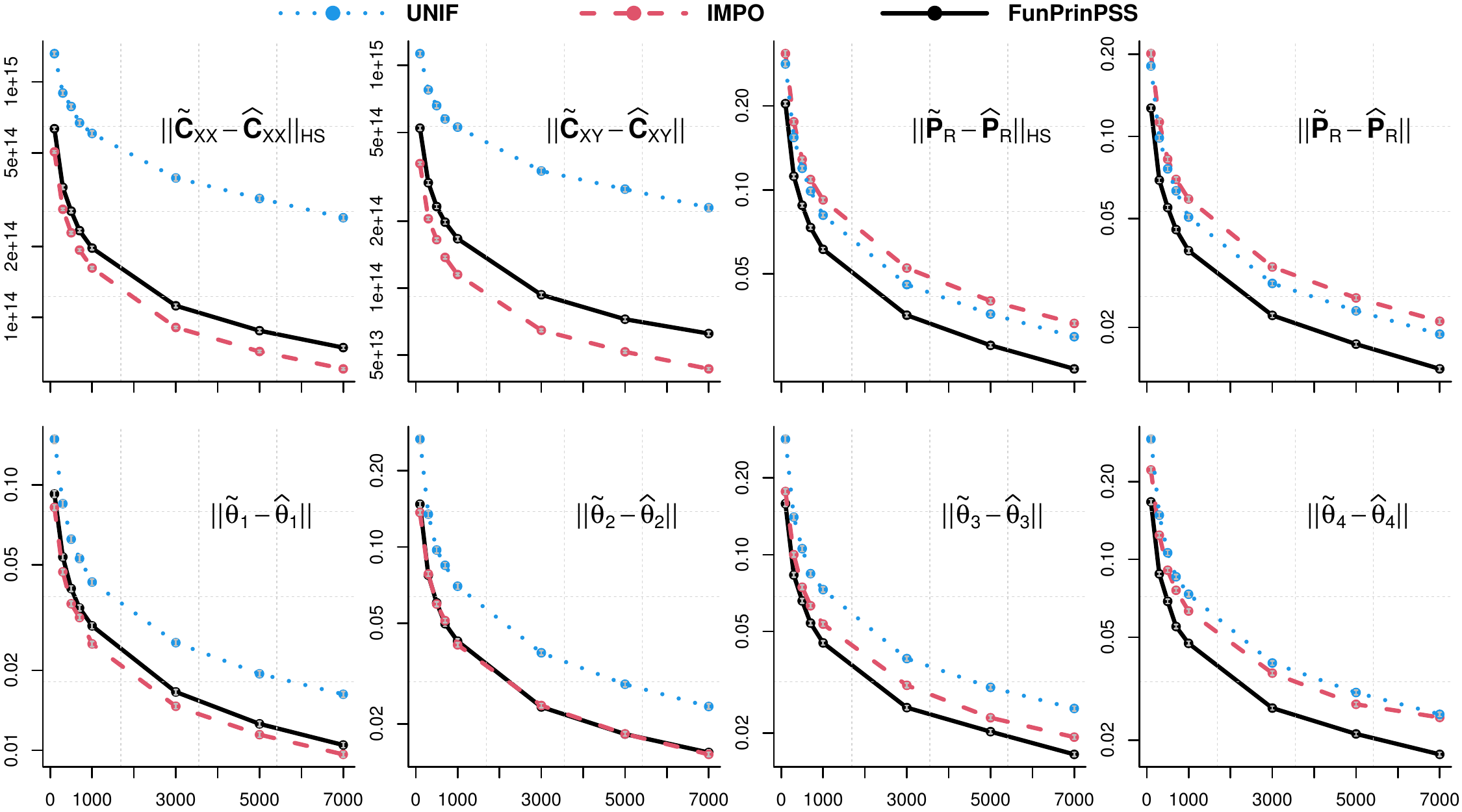} 
	\caption{Randomized covariance operator and FPCA estimation for the ED and MN setting. }\label{fig:ErrorMNexp}
\end{figure}
\begin{figure}[ht]
	\centering
	\includegraphics[width = 0.9\textwidth, height = 0.45\textwidth]{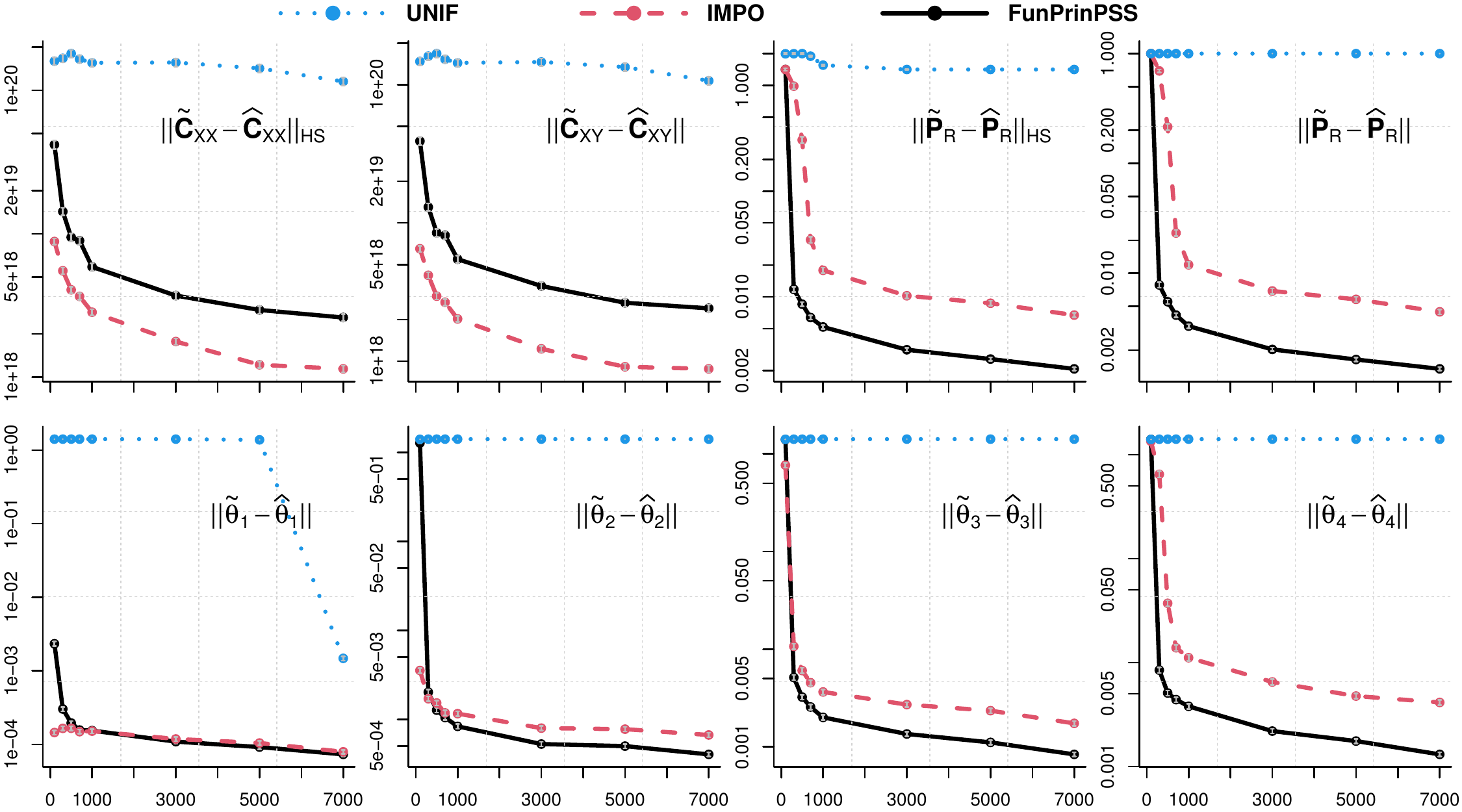} 
	\caption{Randomized covariance operator and FPCA estimation  for the ED and VN setting.}\label{fig:ErrorVNexp}
\end{figure}

\subsection{Randomized FPCA}
\label{sec:simu:fpca}

To assess the randomized FPCA algorithm, we  supply Algorithm~\ref{alg:sampleFPCA} with the three sampling probabilities discussed at the beginning of this section.   In particular, we will	estimate  the first $R=5$ functional principal components ($\theta_1,\cdots, \theta_5$) and their spanned subspace. A variety of criteria are used to assess the algorithm performance. The functional principal subspace estimation error is assessed via the  Hilbert-Schmidt norm and the operator norm  $\| \tsP_R - \hsP_R\|_{HS}$ and $\| \tsP_R - \hsP_R\|$ for the functional principal subspace projection operator. The estimation error for each principal component function will be computed by $\|\ttheta_r - \text{sign}(\langle\ttheta_r,\htheta_r\rangle)\htheta_r\|$. As a byproduct,  the subsampled covariance operator errors $\|\tsC_{XX}-\hsC_{XX}\|_{HS}$ and $\|\tsC_{XX}-\hsC_{XX}\|$ are also reported.
Notice these  criteria are assessing the differences between the subsampled estimates and the full sample estimates, as developed in our theoretical results.

The simulation results with exponentially decaying (ED) eigenvalues are presented in  Figures~\ref{fig:ErrorNUexp}--\ref{fig:ErrorVNexp}.  The three figures correspond to the cases with NU, MN and VN distributions, respectively. In each figure, the first two panels in the first row  demonstrate the estimation error of $\|\tsC_{XX}-\hsC_{XX}\|_{HS}$ and  $\|\tsC_{XX}-\hsC_{XX}\|$. The third and fourth panels in the first row report the error of functional subspace projection operator  $\| \tsP_R - \hsP_R\|_{HS}$ and $\| \tsP_R - \hsP_R\|$. The estimation error for the first four eigenfunctions are reported in the second row.  In addition, for each panel,  the horizontal axes represent the subsample size $C$, while the vertical axes are the comparison metrics in the logarithm scale with base 10. The blue dotted, red dashed  and black solid lines correspond to the results of UNIF, IMPO and FunPrinSS probability, respectively.

From the figures, we can see that as the subsample size increases, all errors tend to decrease monotonically. There are more points worth emphasizing.  Firstly, UNIF has the worst performance in most of the cases. It almost fails completely for the VN datasets in Figure~\ref{fig:ErrorVNexp}, where the some principal component scores have extreme magnitude. Secondly, the sampling strategy~IMPO performs  best in estimating  $\hsC_{XX}$ and $\htheta_1$, which agrees with the result of \cite{he2020randomized}. It also confirms our observations from Figure~\ref{fig:toyFPCA} in Section~\ref{sec:introduction} that this sampling strategy over-emphasizes the eigenfunction $\htheta_1$ with the most dominant eigenvalues. As for the estimation of $\htheta_3$ and $\htheta_4$, the performance of IMPO is no longer the best.  Especially for $\htheta_4$ with a small eigenvalue, its accuracy can even be    worse  than UNIF, as shown in the last penal of Figure~\ref{fig:ErrorNUexp}. Thirdly, our proposed FunPrinSS has the smallest  error for estimating $\hsP_R$, which provides empirical support for our theoretical results of Theorem~\ref{thm:projectionOperator}. It also leads in the accuracy of estimating $\htheta_2, \cdots, \htheta_4$, and the accuracy remains appealing on the VN dataset with extreme observations in Figure~\ref{fig:ErrorVNexp}. The above shows the result for the ED eigenvalue setting.  The randomized FPCA results have similar interpretation for  the PD eigenvalue setting. See  Section~\ref{sec:supp:simulation} in the Appendix.

\subsection{Randomized FLR}
\label{sec:simu:flr}

\begin{figure}[t]
	\centering
	\includegraphics[width = 0.8\textwidth, height = 0.45\textwidth]{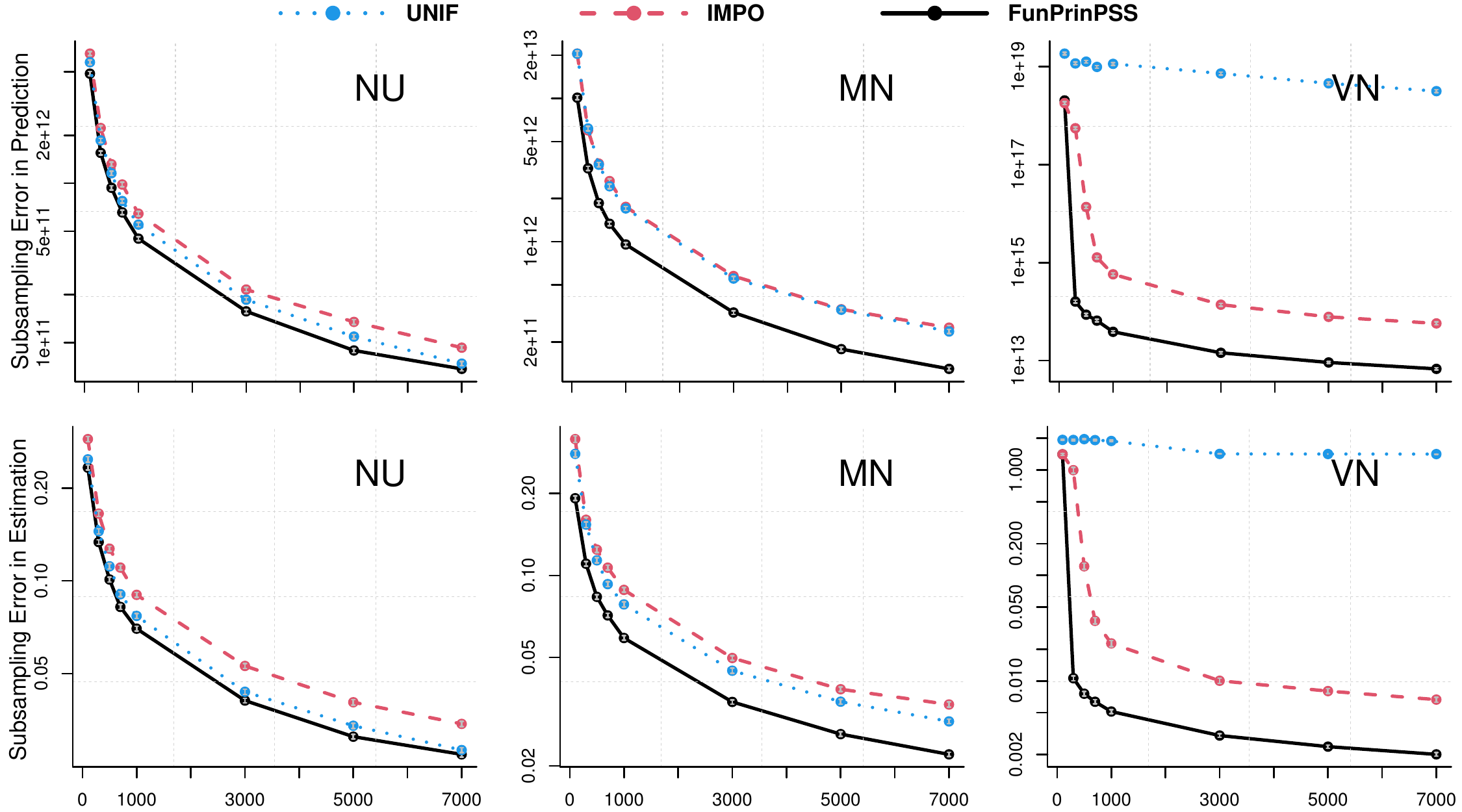} 
	\caption{Randomized FLR for predictors $x_n$ generated with the exponentially decaying eigenvalues. The blue dotted, red dashed  and black solid lines correspond to the results of the UNIF, IMPO and FunPrinSS sampling, respectively.}\label{fig:Regressionexp}
\end{figure}

To examine the performance of the proposed randomized FLR algorithm, we generate data pairs $\{(x_n, Y_n)\}_{n=1}^N$ with $N = 10,000$ as follows. The functional predictor 
$x_n$ is simulated by  the procedure in Section~\ref{sec:simu:data}. The additional scalar response $Y_n$ is set by $Y_n = \langle x_n, \Psi\rangle + \epsilon_n$ with noise $\epsilon_n\sim N(0,1)$.  Besides, the regression function is specified as $\Psi(t) = \sum_{r = 1}^{50}  \theta_r(t)$, where  $\theta_r$'s are the Fourier bases specified in Section~\ref{sec:simu:data}. Recall by combining different choices of eigenvalues (ED or PD), and different choices of score distributions (NU, MN, or VN), six types of  simulation settings are considered.

As in Section~\ref{sec:simu:fpca}, we  consider the same sampling probabilities for Algorithm~\ref{alg:sampleFLR} to obtain the estimator $\widetilde{\Psi}$.  	We consider two metrics to compare these sampling schemes. 	The first is the    prediction error  	$\| \sT\hPsi - \sT\widetilde{\Psi} \|_N^2 = (1/N) \sum_{n=1}^N \langle x_n, \hPsi - \widetilde{\Psi} \rangle^2$ discussed in Section~\ref{subsec:flr:theory}.  	The second is direct estimation error $\|\hPsi-\widetilde{\Psi}\|^2 = \int_0^1 [\hPsi(t)-\widetilde{\Psi}(t)]^2\intd  t$ in $L_2$ norm. The latter  is a stronger norm on the difference between the full sample empirical estimate $\hPsi$ and subsampled estimator $\widetilde{\Psi}$.

The simulation is carried out with 1000 replicates and  the operator inverse is truncated at $R=5$.  Figure~\ref{fig:Regressionexp}  presents the prediction error and the estimation error for the ED eigenvalue setting and different score distributions (NU, MN, VN).  The first row depicts the prediction error, while the second row shows the estimation error. The three columns correspond to the NU, MN and VN score distribution, respectively. For each panel, the horizontal axis is the subsample size $C$, and the vertical axis plots the error in the logarithm scale with base 10. In addition, the blue dotted, red dashed  and black solid lines correspond to the results of UNIF, IMPO and FunPrinSS sampling probabilities, respectively. It is remarkable that the proposed FunPrinSS has the smallest errors in all cases and its performance is stable in the VN datasets with extreme observations.  Between different versions of sampling probabilities, their performance difference enlarges as the tails of the score distributions grow heavier. In some cases (e.g. when the scores has NU distribution), the non-informative uniform sampling (UNIF) is slightly better than the the IMPO sampling. However, UNIF has the worst performance for the VN score distribution.   Figure~\ref{fig:Regressionpoly} in  Appendix~\ref{sec:supp:simulation} presents the simulation results for the PD eigenvalue setting. It also conveys a similar message.

\section{Real Data Analysis}
\label{sec:realdata}

\begin{figure}[t]
	\centering
	\includegraphics[width = 0.64\textwidth]{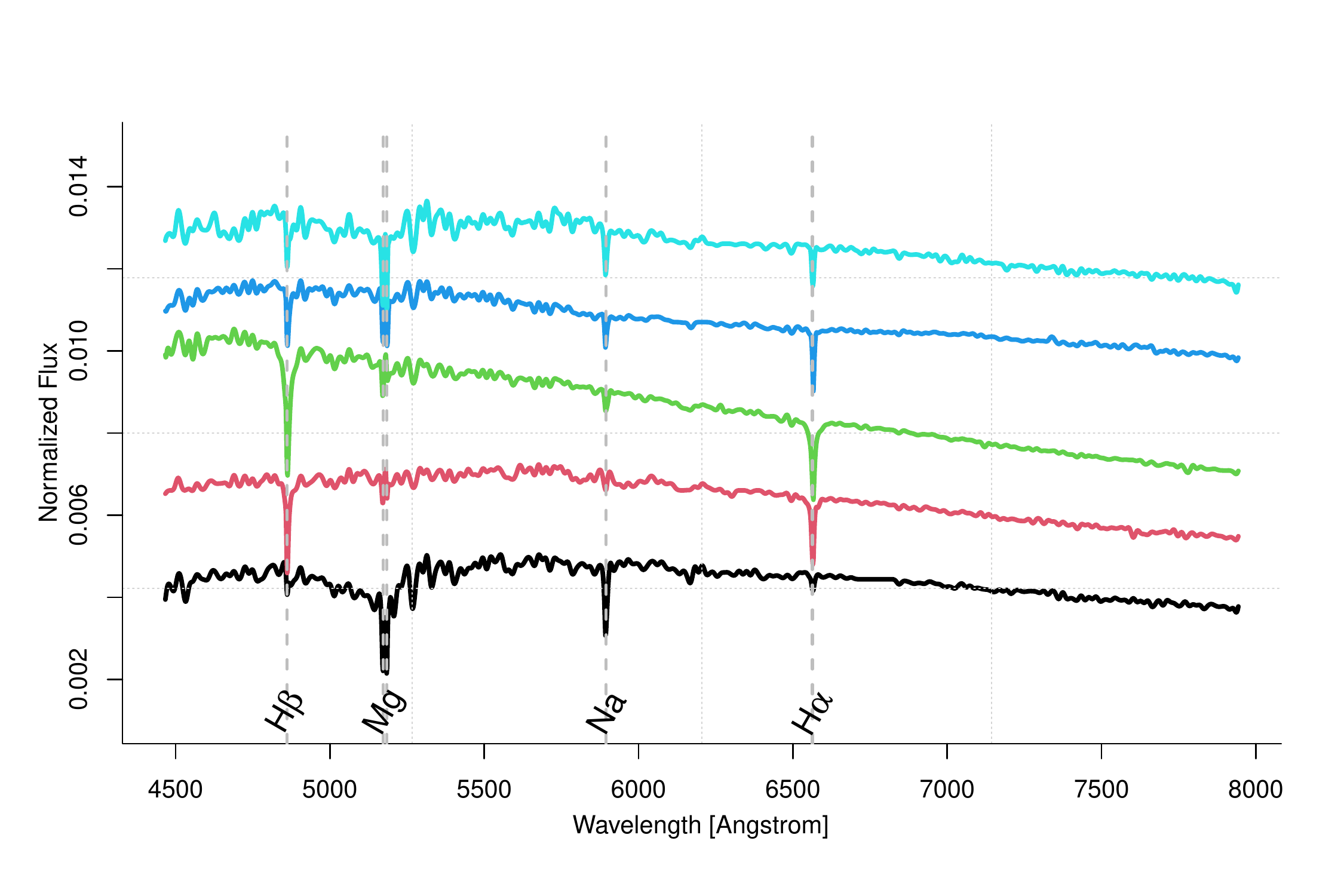} 
	\caption{Five stella spectra from LAMOST with 0.002 offsets in the vertical direction. The horizontal axis is  wavelength (\r{A}) and the vertical axis is  normalized flux.}\label{fig:FiveSpectrum}
\end{figure}

We now illustrate the effectiveness of our randomized algorithm on a stellar spectrum dataset. The dataset contains observations from  the Large sky Area Multi-Object fiber Spectroscopic Telescope \citep[LAMOST,][]{zhao2012lamost}. The LAMOST project began carrying out its  spectroscopic survey of millions of stars and galaxies in 2012. The current LAMOST data release includes more than nine million spectra. Aided by this rapid  progress in the availability of spectra data, we can conduct scientific research on many advanced astrophysical topics, such as exploring  galaxy star generation, tracing the structure and evolution history of our Milky way galaxy and searching for signatures of dark matter distribution and sub-structures in the Milky Way halo.

\begin{table}[t]
	\centering
	\caption{Randomized covariance operator and FPCA  estimation with  $R = 5$ for the stellar spectra dataset. The first column is the comparison metric and the scale of the reported numbers.  In Column 2--10 are the Monte Carlo averages and the standard errors in the parentheses. }
	\label{table:RealDataFPCAError}
	\resizebox{\columnwidth}{!}{%
		\begin{tabular}{  c | c c  c| c c c| c c c  }
			&  \multicolumn{3}{c|}{$C=1,000$} &  \multicolumn{3}{c|}{$C=5,000$}  &  \multicolumn{3}{c}{$C=10,000$}    \\
			\hline 
			\hline 
			
			& UINF & IMPO & FunPrinSS & UINF & IMPO & FunPrinSS & UINF & IMPO & FunPrinSS  \\
			
			\hline 
			$\|\tsC_{XX}-\hsC_{XX}\|_{HS}$ $(\times 10^{-9})$ & 20.75 & {\bf 4.62} & 12.55 & 9.58 & {\bf 2.05} & 5.64 & 7.09 & {\bf 1.44} & 3.91 \\
			&(0.45)&(0.04)&(0.24)&(0.17)&(0.02)&(0.11)&(0.12)&(0.01)&(0.08)\\
			$\|\tsC_{XX}-\hsC_{XX}\|$ $(\times 10^{-9})$ & 19.72  & {\bf 3.52} & 12.30 & 9.05 & {\bf 1.56} & 5.54 & 6.72 & {\bf 1.10} & 3.84 \\
			&(0.46)&(0.03)&(0.25)&(0.18)&(0.01)&(0.11)&(0.12)&(0.01)&(0.08)\\
			$\|\tsP_R-\hsP_R\|_{HS}$ $(\times 10^{-1})$  & 11.65 & 6.45 & {\bf 2.31} & 7.11 & 2.87 & {\bf 1.03} & 5.35 & 2.05 & {\bf 0.73} \\
			&(0.11)&(0.10)&(0.04)&(0.13)&(0.05)&(0.01)&(0.12)&(0.03)&(0.01)\\
			$\| \tsP_R - \hsP_R\|$ $(\times 10^{-1})$ & 7.72 & 4.34 & {\bf 1.52} & 4.76 & 1.93 & {\bf 0.68} & 3.59 & 1.37 & {\bf 0.48} \\
			&(0.08)&(0.07)&(0.03)&(0.10)&(0.04)&(0.01)&(0.08)&(0.03)&(0.01)\\
			$\|\ttheta_1 - \text{sign}(\langle\ttheta_1,\htheta_1\rangle)\htheta_1\|$ $(\times 10^{-2})$ & 1.79 & {\bf 0.77} & 0.91 & 0.89 & {\bf 0.34} & 0.40 & 0.64 & {\bf 0.24} & 0.28 \\
			&(0.04)&(0.01)&(0.01)&(0.01)&(0.00)&(0.00)&(0.01)&(0.00)&(0.00)\\
			$\|\ttheta_2 - \text{sign}(\langle\ttheta_2,\htheta_2\rangle)\htheta_2\|$ $(\times 10^{-2})$& 21.70 & 7.94 & {\bf 5.15} & 7.70 & 3.54 & {\bf 2.31} & 5.35 & 2.49 & {\bf 1.62} \\
			&(0.81)&(0.09)&(0.07)&(0.10)&(0.04)&(0.03)&(0.06)&(0.03)&(0.02)\\
			$\|\ttheta_3 - \text{sign}(\langle\ttheta_3,\htheta_3\rangle)\htheta_3\|$ $(\times 10^{-2})$ & 39.39 & 12.37 & {\bf 6.46} & 16.69 & 4.53 & {\bf 2.88} & 9.55 & 3.20 & {\bf 1.99} \\
			&(1.40)&(0.36)&(0.07)&(0.86)&(0.05)&(0.03)&(0.52)&(0.03)&(0.02)\\
			$\|\ttheta_4 - \text{sign}(\langle\ttheta_4,\htheta_4\rangle)\htheta_4\|$ $(\times 10^{-2})$ & 99.42 & 34.68 & {\bf 9.91} & 53.67 & 14.53 & {\bf 4.29} & 37.89 & 9.93 & {\bf 3.03} \\
			&(0.92)&(0.83)&(0.20)&(1.10)&(0.35)&(0.07)&(0.91)&(0.23)&(0.05)\\
			$\|\ttheta_5 - \text{sign}(\langle\ttheta_5,\htheta_5\rangle)\htheta_5\|$ $(\times 10^{-2})$ & 90.89 & 55.06 & {\bf 18.11} & 61.87 & 24.30 & {\bf 7.93} & 47.71 & 17.08 & {\bf 5.62} \\
			&(1.10)&(0.99)&(0.29)&(1.28)&(0.46)&(0.11)&(1.15)&(0.29)&(0.08)\\
			\hline
		\end{tabular}
	}	
\end{table}

From the the current LAMOST data release, we randomly select a dataset containing $N = 110,000$ stellar spectra  with $r$-band signal-to-noise ratio larger than 35. Each spectrum is normalized to have unit norm and get denoised by a spline wavelet method. 	In addition, the average spectrum is subtracted from each observation  such that the dataset has zero mean. 	We treat this as our full dataset and apply the subsampling algorithms over this dataset. 	Figure~\ref{fig:FiveSpectrum} presents five example spectra from the pre-processed dataset.   In the figure,  the horizontal axis  is  wavelength, and the vertical axis is flux which measures the brightness of an object.  Each spectrum observation can be viewed as functional data, where the flux value $x_n(\lambda)$ varies as a function of  the wavelength $\lambda$.  In Figure~\ref{fig:FiveSpectrum}, several important spectral features are marked as vertical dashed lines. These spectral features are adopted by \cite{liu2015spectral} for spectral classification.  

Characterizing the diversity of  stars and extracting  stellar atmospheric parameters are of key astrophysical importance for astronomical observational studies.  We will apply the randomized FPCA in Algorithm~\ref{alg:sampleFPCA} and the randomized FLR in Algorithm~\ref{alg:sampleFLR} to this dataset.

\subsection{Results for Randomized FPCA}

Dimensional reduction  has been an effective  tools for astronomical spectrum data analysis \citep{bermejo2013pca,McGurk:2010aw}. We  apply the randomized FPCA in  Algorithm~\ref{alg:sampleFPCA} with the three sampling strategies (UNIF, IMPO and FunPrinSS) in Section~\ref{sec:simu}.  We will estimate the leading $R=5$ functional principal components, together with the covariance operator. The parameter $R=5$ is chosen  such that FVE is approximately $98\%$. The improvement of FVE with $R=6$ is marginal (see Figure~\ref{fig:realdata:FVE}).

Based on the whole selected 110,000 spectra, the full sample estimators are computed with brute force.  Besides,	we  apply the randomized algorithm with the subsample size $C$ varying among $1,000$, $5,000$ and $10,000$.  The randomized estimates are compared against  the full sample results by the metrics introduced in Section~\ref{sec:simu:fpca}.  	For each subsample size, the  procedure is repeated $1,000$ times to compute the average estimation loss. 	Table~\ref{table:RealDataFPCAError} presents the resulting average  accuracy and related standard errors. The method performance in Table~\ref{table:RealDataFPCAError} can be interpreted  similarly as that from  Section~\ref{sec:simu:fpca}. The IMPO sampling strategy delivers the smallest error in estimating $\htheta_1$ and $\hsC_{XX}$. However, FunPrinSS has the best performance in estimating the projection operator $\hsP_R$ and the other functional principal components.

\begin{figure}[p]
	\centering
	\includegraphics[width=0.95\textwidth]{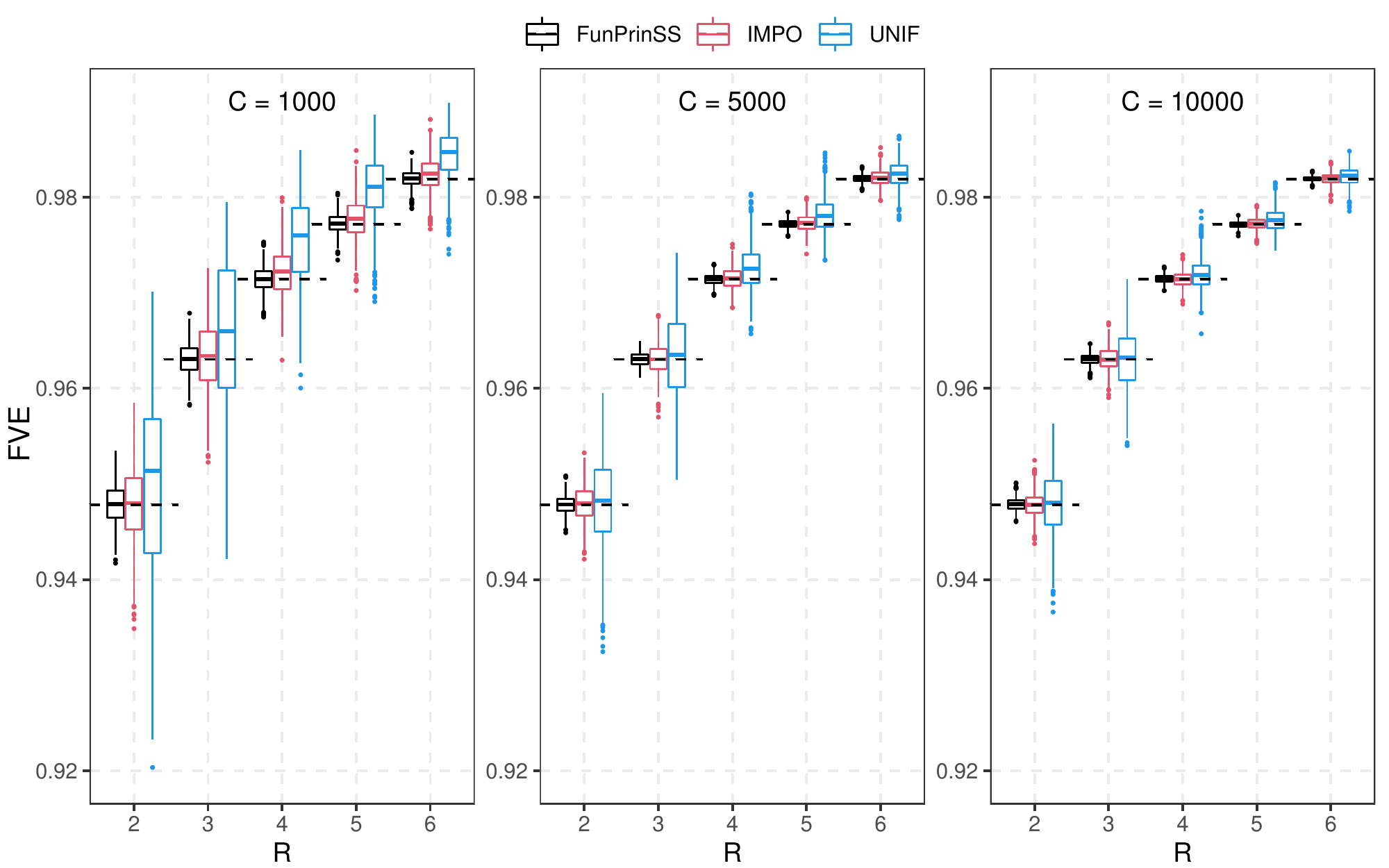}
	\caption{The boxplot of the FVE of  the subsampled functional principal component estimates  under 1000 replicates. The three panels  correspond to the subsample size $C=1000$, $5000$ and $10000$, respectively. In each panel, different sampling probabilities (UNIF, IMPO, FunPrinSS) are compared over various $R$. The FVEs of the full sample estimator are shown as dashed horizontal lines.  }  \label{fig:realdata:FVE}
\end{figure}

\begin{figure}[p]
	\centering
	\includegraphics[width = \textwidth,height = 0.5\textwidth]{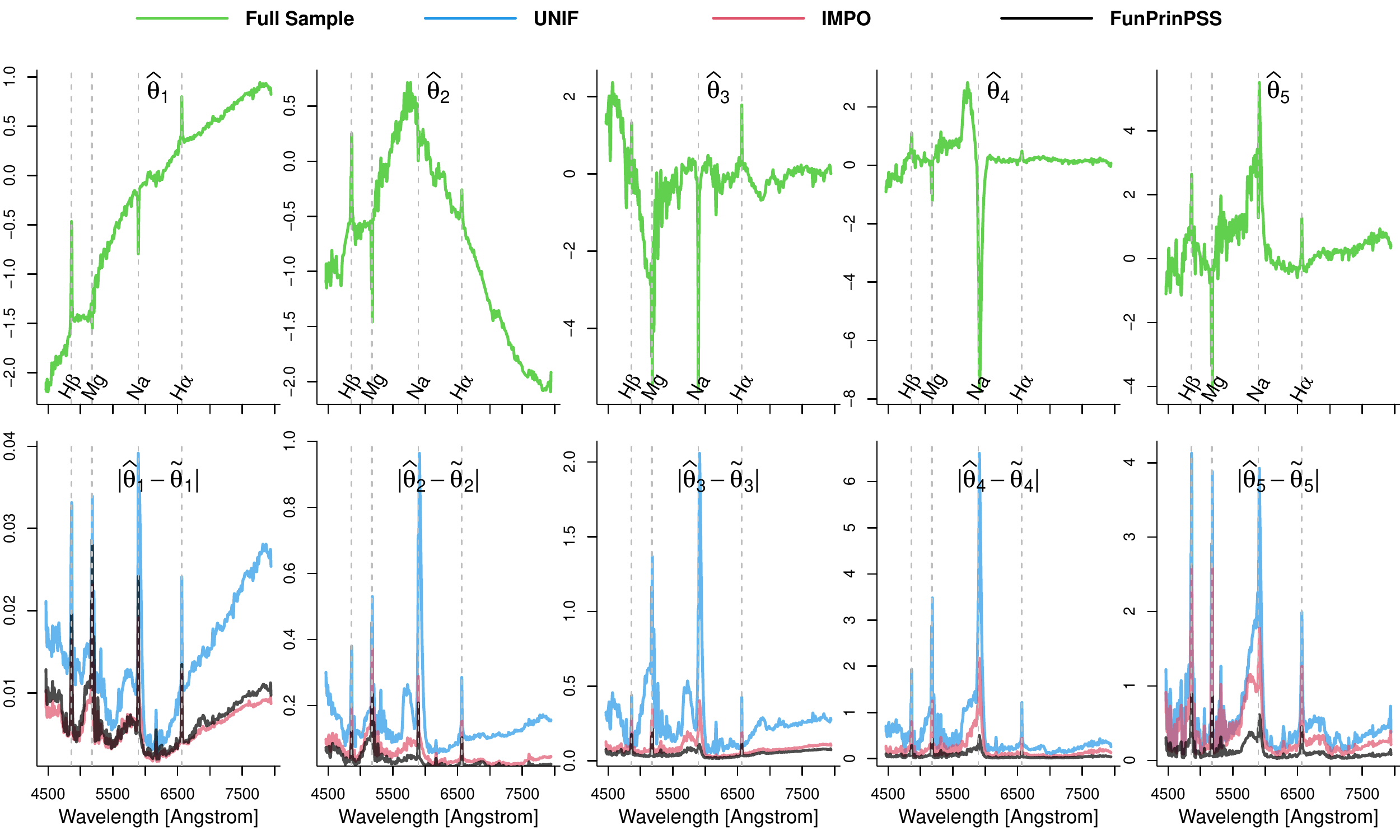} 
	\caption{ The green solid lines (in the first row) represent the leading five eigenfunctions of full sample covariance operator. The blue, red and black solid lines (in the second row) represent the averaged absolute difference between the full sample eigenfunctions  and the subsampled eigenfunctions obtained by UNIF, IMPO and FunPrinSS sampling, respectively. }\label{fig:RealDataEigenFun}
\end{figure}

In Table~\ref{table:RealDataFPCAError}, for the  randomized FPCA with the  FunPrinSS probability, we can see the  Monte Carlo average error $\|\tsP_R - \hsP_R\|\approx 0.152$, $0.068$, $0.048$ for $C=1000$,  $5000$ and $10000$, respectively. Neglecting the second order term ($\epsilon^2$) in~\eqref{thm:projectionOperator:bound} of Theorem~\ref{thm:projectionOperator}, we can find   the expected subsample error $\Expect\|\tsP_R - \hsP_R\|$  is  upper bounded by  a first-order   term with magnitude order $\sO\big((R+\Delta_R) \log^{1/2} (R+\Delta_R)/C^{1/2} \big)$. For $C=1000$,  $5000$ and $10000$, the values of this theoretical order are $\sO\big(0.469)$, $\sO\big( 0.210)$ and $\sO\big( 0.148)$, respectively. The ratio between the Monte Carlo average error and the theoretical first-order approximate magnitude is stable at a constant around $0.32$.

In Figure~\ref{fig:RealDataEigenFun}, we plot the leading five full sample eigenfunctions in the first row of panels,  and the averaged absolute difference between the full sample and  subsampled eigenfunctions in the second row of panels.  The absolute  difference $|\htheta_{r}(\lambda) - \ttheta_{r}(\lambda)|$ at each wavelength $\lambda$ is averaged over   the $1000$ replicates when the  subsample size is $C=1,000$. From the figure, we can compare the ability  of capturing  the spectral features ($H_{\beta}$, $Na$, $Mg$ and $H_{\alpha}$ marked by the vertical dashed lines)  by various sampling probabilities. As expected, the  UNIF sampling (colored by blue) has the largest error for characterizing the eigenfunctions and spectral features. The IMPO  probability (red) is competitive in recovering spectral lines for the first  eigenfunctions  but not so for the rest of eigenfunctions. For the other eigenfunctions,  the  error of  estimating the spectral features by the  FunPrinSS  probability (black) is the smallest.

We  examine the accuracy of 	the FVE of  the subsampled estimator  (see $\widetilde{\mathrm{FVE}}$ in Section~\ref{subsec:fpca:chooseR}). Figure~\ref{fig:realdata:FVE} draws the boxplots for the FVE of the subsampled  functional principal component estimates under 1000 replicates. The three panels correspond to the subsample size $C=1000$, $5000$ and $10000$, respectively. In each panel, the sampling probabilities (UNIF, IMPO, FunPrinSS) are compared over various $R$. The FVEs of the full sample estimator are shown as dashed horizontal lines. From the figure, we can see the  FVE of the subsampled estimator   serves as an estimate of the FVE of the full sample estimator. In particular, the  FVE of the subsampled estimator   with the  FunPrinSS probability is very accurate with small variance.

\subsection{Results for Randomized FLR}

\begin{figure}[p]
	\centering
	\includegraphics[height = 0.55\textwidth,width = 0.85\textwidth]{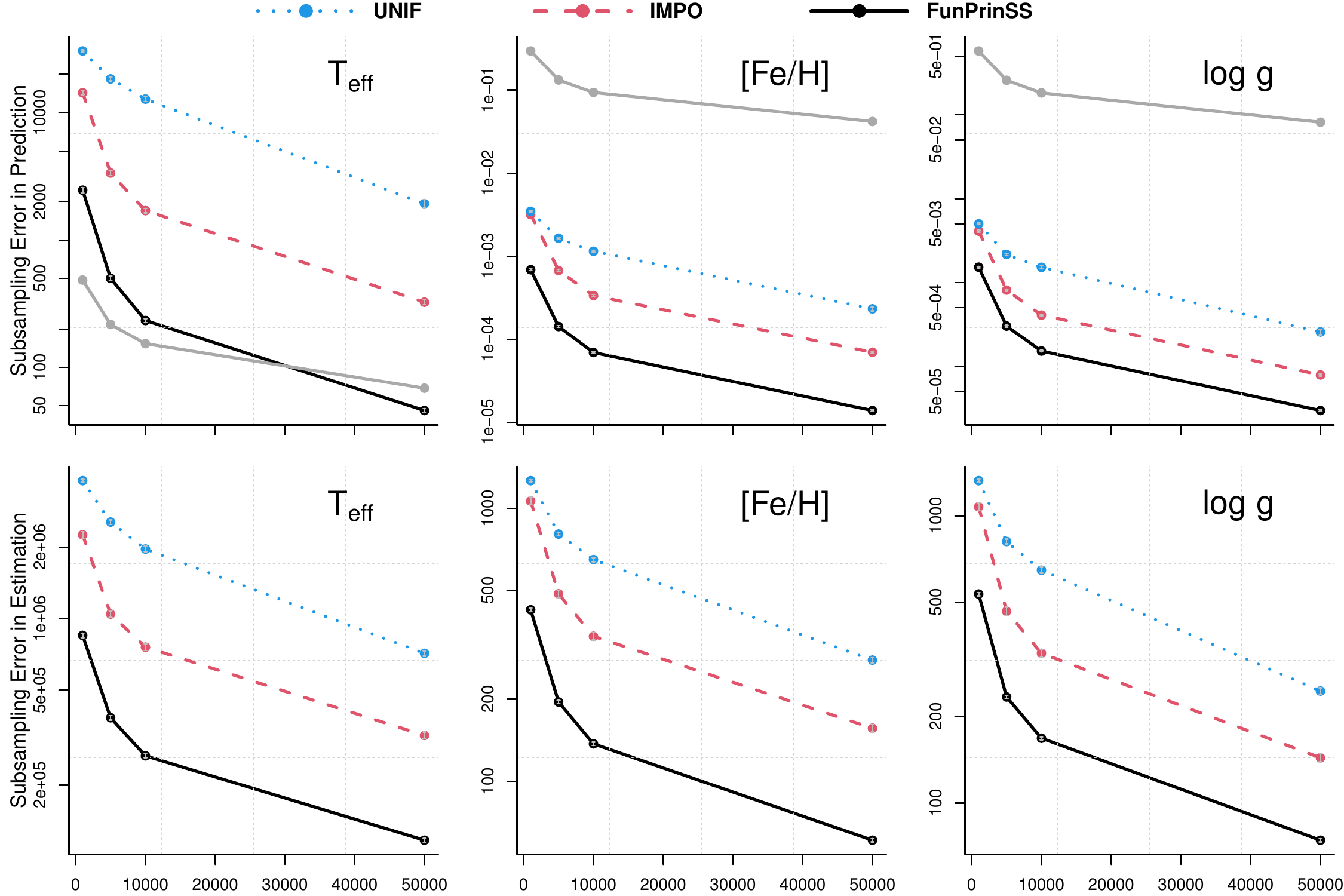} 
	\caption{Subsampling error in estimation and in prediction for the randomized FLR for the LAMOST spectra dataset. The blue dotted, red dashed  and black solid lines correspond to the results of the UNIF, IMPO and FunPrinSS sampling, respectively. The gray solid curves in the first row show  the magnitude $\sO\big((R+\Delta_R) (\|\mY\|_N+\|\mY^{\perp}\|_N) \log^{1/2} (R+\Delta_R)/C^{1/2} \big)$ of  the first-order  upper bound  for the expected subsample error.}\label{fig:RealDataRegression}
\end{figure}
\begin{figure}[p]
	\centering
	\includegraphics[height = 0.55\textwidth,width = 0.95\textwidth]{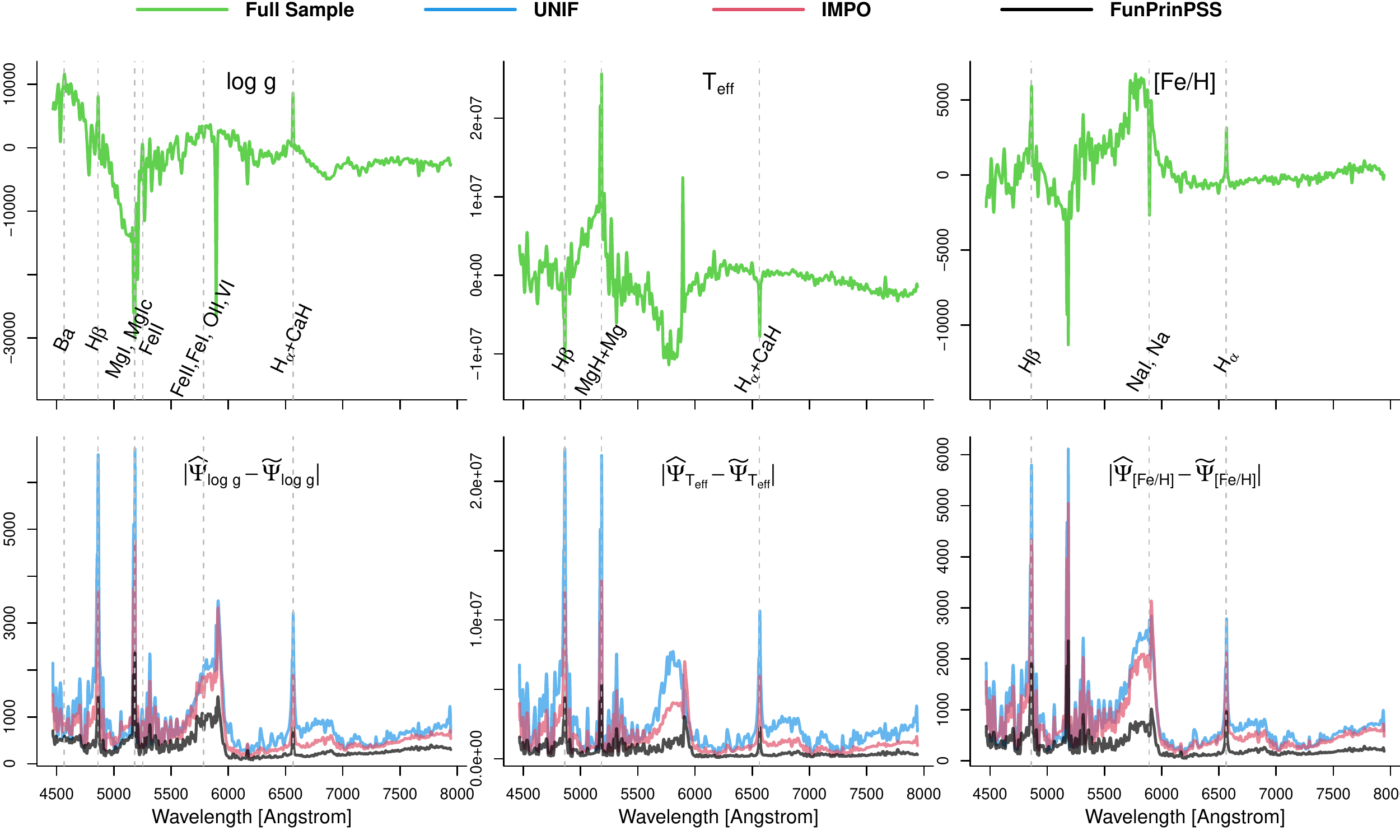} 
	\caption{The comparison of $\widetilde{\Psi}$ obtained by the randomized FLR with different sampling probabilities. The first row (green lines) presents the full sample regression functions $\widehat{\Psi}$ and the second row shows the average absolute difference between the full sample and the subsampled estimates.  The blue, red and black solid lines (in the second row) correspond to the subsampled estimator $\widetilde{\Psi}$'s  obtained by UNIF, IMPO and FunPrinSS, respectively. The left, middle and right  column have the results for the  $\log\, g$, $T_{\text{eff}}$ and [Fe/H] response, respectively.}\label{fig:RealDataRegressionPhi}
\end{figure}

The fundamental stellar atmospheric parameters such as effective temperature ($T_{\text{eff}}$), surface gravity (log $g$), and metallicity ([Fe/H]) are useful probes for revealing the kinematic and chemical properties of  stars. 	We conduct three functional linear regression tasks to compare the different sampling schemes (UNIF, IMPO and FunPrinSS). For each task, the  response $Y_n$ is selected  as one of the  stellar atmospheric parameters: the effective temperature ($T_{\text{eff}}$),  the surface gravity ($\log g$), and the metallicity ([Fe/H]).  In particular, 
the regression model is $Y_n = \int_\Lambda x_n(\lambda) \Psi(\lambda) \intd \lambda + \epsilon_n$. Each spectrum $x_n$ is integrated with the regression function $\Psi$ over the observed wavelength range. Both the response and the spectrum predictor are centered with zero mean, so we do not need to estimate the intercept term.

The accuracy metrics in Section~\ref{sec:simu:flr} are also employed here for comparison
between the full sample and subsampled estimator.  The full sample estimator  $\widehat{\Psi}$ is  obtained via brute force computation on the full dataset.   In Figure~\ref{fig:RealDataRegression}, the subsampling prediction error and estimation error are shown in the first and second row, respectively. 
The three columns of Figure~\ref{fig:RealDataRegression} corresponds to the results with the response $T_{\text{eff}}$, $\log g$, and [Fe/H], respectively.
In each penal, the vertical axis is the error in the scale of $\log_{10}$; and the horizontal axis is the subsample size $C$. These plots suggest that the proposed FunPrinSS method outperforms the other two sampling schemes under all circumstances.  Theorem~\ref{thm:regression} indicates 	the expected subsample error $\Expect\| \sT\widetilde{\Psi} - \sT\hPsi  \|_N $  can be approximately upper bounded by a first-order  term with  magnitude order $\sO\big((R+\Delta_R) (\|\mY\|_N+\|\mY^{\perp}\|_N) \log^{1/2} (R+\Delta_R)/C^{1/2} \big)$. The values of  the theoretical order are plotted as the gray solid curves in  the first row of Figure~\ref{fig:RealDataRegression}. The vertical axis is in the scale of $\log_{10}$, and the theoretical order predicts the  Monte Carlo performance up to an offset.  Note the theoretical order underpredicts  the actual Monte Carlo average error for small $C$, as the higher  order error terms  in Theorem~\ref{thm:regression} have been discarded.

Figure~\ref{fig:RealDataRegressionPhi} plots the full sample regression function estimates $\widehat{\Psi}$ in the first row.  The second row shows  the averaged absolute difference between the full sample estimates $\widehat{\Psi}$ and  the subsampled estimates $\widetilde{\Psi}$ over $1000$ replicates with $C=1,000$.
The work of \cite{li2014sdss} has identified a few spectral features  in determining the stellar atmosphere parameters. Some of these key spectral features are also marked by the vertical dashed lines. The proposed FunPrinSS has generally smaller error in identifying these spectral features (e.g. $H_{\alpha}$ $\lambda 6563$) and provides a better estimation of the feature strength.

\section*{Acknowledgment}
	The authors thank the editor, the associate editor, and an anonymous reviewer for constructive comments that helped significantly improve this work.  The authors also thank Professor Jianhua Z. Huang, Professor Xiangyu Chang and Professor Yixuan Qiu for comments on the initial draft of this work.  Shiyuan He's research was  supported by National Natural Science Foundation of China (No.11801561). This research was supported by Public Computing Cloud, Renmin University of China.


\appendix

\clearpage

\section{Comparison with  sampling probabilities in the multivariate context}
\label{sec:samplecompare}

There are  some underlying connections between	the FunPrinSS probability~\eqref{eqn:prob:FPCAProptoExact} for functional data and some sampling probabilities  for multivariate statistics in the literature. In this section, we discuss  some of the similarities and differences.
	
	The FunPrinSS probability~\eqref{eqn:prob:FPCAProptoExact}  is closely related to the leverage sampling probability~\citep{drineas2012fast} for classical linear regression. Let $\mathbf{X} = \left(\mathbf{x}_1,\ldots,\mathbf{x}_N\right)^T$ be a $N\times p$ fixed design matrix with $N$ observations and $p$ variables (with $N>p$). 
	Suppose the singular value decomposition (SVD) of $\mathbf{X}$ is $\mX = \mU\mD\mV^T$, where $\mD$ is a diagonal matrix with singular values, and the columns of $\mU$ and $\mV$ contain the  left and right singular vectors, respectively. Denote $U_{nj}$ as the $(n,j)$-th element of $\mU$.
	The leverage sampling probability $p_n$ is proportional to the \textit{statistical leverage} $h_{nn}$ of $n$-th observation, i.e. 
	$p_n\propto h_{nn} = \mathbf{x}_n^T\left(\mathbf{X}^T\mathbf{X}\right)^{-}\mathbf{x}_n
	= \sum_{j = 1}^p U_{nj}^2.$ When the samples $\vx_n$'s have zero mean,  $U_{nj}$ can be interpreted as  multivariate principal component scores, and the statistical leverage  $h_{nn}$ is the summation of all squared scores in the $p$-dimensional space. In contrast, our FunPrinSS probability~\eqref{eqn:prob:FPCAProptoExact} makes a rank-$R$ truncation to the Karhunen-Lo\`{e}ve expansion~\eqref{eq:KL} of the infinite dimensional functional data, and takes into account the effect of the remaining scores via the projected residual $(I - \hsP_R) x_n$.

	The FunPrinSS probability is also connected to the subspace sampling  of \cite{Drineas2006Subspace}. To perform the CUR  decomposition for a matrix $\mX\in\bbR^{N\times p}$, they compute the row sampling probabilities as follows. A matrix $\mC\in\bbR^{N\times q}$ is formed by sampling columns from $\mX$, and its SVD $\mC = \mU_C\mD_C\mV_C^T$ is computed. Then, the sampling probability for $n$-th row of $\mX$ is computed from the $n$-th row of $\mU_C$, the $n$-th row of $\mX^{\perp}_C = (\mI-\mU_C\mU_C^T)\mX$ and their product
	$$
	p_n = \frac{(1/3)\|(\mU_C)_{(n)}\|_2^2}{\sum_{m = 1}^N \|(\mU_C)_{(m)}\|_2^2} +  
	\frac{(1/3)\|(\mU_C)_{(n)}\|_2\|(\mX^{\perp}_C)_{(n)}\|_2}{\sum_{m = 1}^N \|(\mU_C)_{(m)}\|_2\|(\mX^{\perp}_C)_{(m)}\|_2} 
	+
	\frac{(1/3)\|(\mX^{\perp}_C)_{(n)}\|_2^2}{\sum_{m = 1}^N \|(\mX^{\perp}_C)_{(m)}\|_2^2}.
	$$
	In the above, $(\mA)_{(n)}$ represents the $n$-th row of a matrix $\mA$ and $\|\cdot\|_2$ is the classical unnormalized Euclidean norm. When we view the rows of $\mX$ as samples (observations) of a zero mean random vector and the columns of $\mX$ as variables, the matrix $\mU_C$ serves as an estimate of the principal component scores, and $(\mI-\mU_C\mU_C^T)\mX$  is the residual matrix.

\section{Concentration Inequality for Compact Operators}

Theorem 7.3.1 of \cite{tropp2015introduction} developed intrinsic matrix Bernstein concentration inequality. Together with the techniques in Section 3.2 of \cite{minsker2017some}, it can be easily extended to compact operators. 
Similar results can also be found in Lemma~5 of \cite{dicker2017kernel}.
In the below, $\sH_1$ and $\sH_2$ are two  Hilbert spaces.

\begin{lemma} \label{lemma:mainConcentration}
	For a finite sequence of random operators $\sZ_i$ mapping from $\sH_1$ to $\sH_2$, they satisfies $\Expect \sZ_i = 0$  and $\Vert \sZ_i \Vert \le L$.
	Suppose their summation is $\sS = \sum_i \sZ_i$ and $\sS^*$ is its adjoint. Let $\sV_1$ and $\sV_2$ be  semidefinite upper bounds for  $\mathrm{Var}_1(\sS)$ and $\mathrm{Var}_2(\sS)$, respectively. That is,
	$\sV_1 \succeq  \mathrm{Var}_1(\sS) = \Expect (\sS\sS^*) = \sum_i \Expect (\sZ_i\sZ_i^*)$, and $\sV_2 \succeq  \mathrm{Var}_2(\sS) = \Expect (\sS^*\sS) = \sum_i \Expect (\sZ_i^*\sZ_i)$. Define an intrinsic bound $d = \mathrm{intdim}(\sV_1) + \mathrm{intdim}(\sV_2)$ and a variance bound
	$v = \max\{\Vert \sV_1\Vert, \Vert \sV_2\Vert  \}$, then for $t \ge \sqrt{v} + L /3$, 
	$$
	\Prob(\Vert \sS\Vert \ge t) \le 4 d \exp\left(- \frac{t^2/2}{v + Lt/3}\right).
	$$
\end{lemma}

The next lemma provides upper bound for self-adjoint operator.

\begin{lemma} \label{lemma:mainConcentration2}
	For a finite sequence of random self-adjoint  operators $\sZ_i$ mapping from $\sH_1$ to $\sH_1$, they satisfies $\Expect \sZ_i = 0$ and $\Vert \sZ_i \Vert \le L$. Suppose their summation is $\sS = \sum_i \sZ_i$, and $\sV$ is the semidefinite upper bound for 
	$\sum_i \Expect(\sZ_i^2)$. 
	Define an intrinsic bound $d = \mathrm{intdim}(\sV)$ and a variance bound
	$v = \Vert \sV\Vert$, then for $t \ge \sqrt{v} + L /3$, 
	$$
	\Prob(\Vert \sS\Vert \ge t) \le 4 d \exp\left(- \frac{t^2/2}{v + Lt/3}\right).
	$$
\end{lemma}

\section{Proof for Section~\ref{sec:fpca:theory}}\label{appendix:fpca:proof}

\subsection{Proof of Lemma~\ref{lemma:pcaLinearSum} }
Define $\hsP_r = \htheta_r\otimes \htheta_r$. With the definition of the resolvent
$\sR_{\hsC_{XX}}(\eta) = ( \hsC_{XX}-\eta I )^{-1} = -  \sum_{r=1}^\infty \frac{1}{ \eta - \hsigma_{r}^2  } \hsP_r$,
the linear error term can be expanded as
\begin{align*}
	L_R(\sE) & = \frac{1}{2\pi i} \oint_{\Gamma_R} \sR_{\hsC_{XX}}(\eta) \sE \sR_{\hsC_{XX}}(\eta) \mathrm{d}\eta 
	= \sum_{r=1}^R \sum_{s=R+1}^\infty \frac{1}{\hsigma_{r}^2 - \hsigma_{s}^2}
	\left( \hsP_r\sE \hsP_s + \hsP_s\sE \hsP_r \right).
\end{align*}
by Cauchy integration formula. The above results can also be deduced from Theorem~5.1.4 of \cite{hsing2015theoretical}.
Notice that
\begin{align*}
	\hsP_r (x_n \otimes x_n) \hsP_s &= 
	(\htheta_r\otimes \htheta_r) (x_n \otimes x_n) (\htheta_s\otimes \htheta_s)  \\
	&= \langle \htheta_r, x_n\rangle \cdot \langle \htheta_s, x_n \rangle\cdot (\htheta_s\otimes \htheta_r) 
	= \hsigma_r \hsigma_s \hxi_{ns}\hxi_{nr} (\htheta_s\otimes \htheta_r),
\end{align*}
and $\hsP_r \hsC_{XX} \hsP_s = 0$ for any $r\neq s$. From these, we know for $r\neq s$ that
\begin{align*}
	\hsP_r \sE \hsP_s = \hsP_r (\tsC_{XX} - \hsC_{XX}) \hsP_s = \frac{\hsigma_r\hsigma_s}{CN} \sum_{c=1}^C
	\frac{\txi_{cr}\txi_{cs} }{\tp_c} (\htheta_s\otimes \htheta_r).
\end{align*}
In the above, $\txi_{cs}$ is the $s$-th score of the subsampled $\tx_c$. 
It holds $\txi_{cs}  = \langle \tx_c, \htheta_s\rangle/\hsigma_{s}$ when $\hsigma_s>0$; and $\txi_{cs}  =1$ when $\hsigma_s=0$.
Notice it is computed from the full sample eigenfunction $\htheta_s$ and the subsampled $\tx_c$. As a result, the linear error term can be expressed as
$L_R(\sE) = \frac{1}{C} \sum_{c=1}^C \frac{\sZ_c}{N\tp_c}$, with
$$ \sZ_c =  \sum_{r=1}^R
\sum_{s=R+1}^\infty \frac{\hsigma_r\hsigma_s}{\hsigma_r^2 - \hsigma_s^2} \times
\txi_{cr}\txi_{cs} \times 
\big[\htheta_r \otimes \htheta_s +
\htheta_s \otimes \htheta_r\big].$$

\subsection{Proof of Theorem \ref{thm:projectionOperator}} \label{sec:proof:fpca}
By Lemma~\ref{lemma:lsbound}, we have
$\|\tsP_R - \hsP_R\|  \le \| L_R(\sE) \| +
\| S_R(\sE) \|  \le \| L_R(\sE) \|  +K_R \| \sE \| ^2/ g_R^2$.
It follows that
$$
\Prob\big(\|\tsP_R - \hsP_R\| \ge \epsilon + \hsigma_1^4K_R\epsilon^2/ g_R^2\big) 
\le \Prob \big( \| L_R(\sE) \| \ge \epsilon \big)
+ \Prob \big( \| \sE \| \ge \hsigma_{1}^2\epsilon \big).
$$
The two probabilities on the right hand side are controlled by the next two lemmas. The conclusion of  Theorem~\ref{thm:projectionOperator} follows by combing the next two lemmas and noting that
$\Delta_0 \le R + \Delta_R$.

\begin{lemma} \label{lemma:projection:part1}
	Under the conditions of Theorem \ref{thm:projectionOperator}, it holds that
	\begin{align*}
		\Prob \big( \| L_R(\sE) \| \ge \epsilon \big)
		\le 8R \times 
		\exp\Big( -\frac{C\epsilon^2/2}{\frac{\beta G_R^2}{2} \big(R + \Delta_R\big)^2+ 
			[G_R\beta \big( R + \Delta_R\big) /\sqrt{2}] \times  \epsilon /3} \Big)
	\end{align*}
	for $\epsilon\cdot  C\ge\sqrt{\beta C/2}G_R(R+\Delta_R) + [G_R\beta \big( R + \Delta_R\big) /\sqrt{2}]/3$.
\end{lemma}

\begin{lemma} \label{lemma:projection:part2}
	Denote  $\Delta_0=\mathrm{intdim}(\hsC_{XX})$ as the intrinsic dimension of $\hsC_{XX}$.
	Under the conditions of Theorem \ref{thm:projectionOperator}, we have
	$$
	\Prob\left(\left\Vert \tsC_{XX} - \hsC_{XX}\right\Vert \ge \epsilon \hsigma_{1}^2\right) \le 
	4 \Delta_0 \exp\Big(- \frac{C\epsilon^2/2}{\beta (R + \Delta_R) + [\beta (R + \Delta_R)+1]\epsilon/3}\Big),$$
	for $\epsilon$ satisfying 	$\epsilon \cdot C > \sqrt{C\beta (R + \Delta_R)} + [\beta(R+\Delta_R) + 1] /3$
\end{lemma}

\subsection{Proof of Lemma~\ref{lemma:projection:part1} }
Recall from Lemma~\ref{lemma:pcaLinearSum}, we have the relation 
$L_R(\sE)  = \frac{1}{C} \sum_{c=1}^C \frac{\sZ_c}{N\tp_c}$ with
$$ \frac{\sZ_c}{N\tp_c} =  \sum_{r=1}^R
\sum_{s=R+1}^\infty \frac{\hsigma_r\hsigma_s}{\hsigma_r^2 - \hsigma_s^2} \times
\frac{\txi_{cr}\txi_{cs}}{\tp_c N} \times 
\big[\htheta_r \otimes \htheta_s +
\htheta_s \otimes \htheta_r\big].$$
We will apply Lemma~\ref{lemma:mainConcentration2} to the summation $\sS=\sum_{c=1}^C \frac{\sZ_c}{N\tp_c}$. We first bound the operator norm of each summand by
\begin{align*}
	\|  \frac{\sZ_c}{N\tp_c} \|^2& \le \| \frac{\sZ_c}{N\tp_c}\|_{HS}^2  = 
	2 \sum_{r=1}^R
	\sum_{s=R+1}^\infty \frac{\hsigma_r^2\hsigma_s^2}{[\hsigma_r^2 - \hsigma_s^2]^2} \times
	\frac{\txi_{cr}^2\txi_{cs}^2}{\tp_c^2 N^2}   \\
	&= 
	2 \sum_{r=1}^R
	\sum_{s=R+1}^\infty \frac{\hsigma_r^2\hsigma_R^2}{[\hsigma_r^2 - \hsigma_s^2]^2} \times
	\frac{\txi_{cr}^2(\hsigma_s^2\txi_{cs}^2 / \hsigma_R^2)}{\tp_c^2 N^2}   \\
	&\stackrel{(i)}{\le} \frac{2 G_R^2}{\tp_c^2 N^2} \Big( \sum_{r=1}^R\txi_{cr}^2\Big) 
	\Big( \sum_{s=R+1}^\infty\hsigma_s^2 \txi_{cs}^2 /\hsigma_R^2 \Big) \\
	&\stackrel{(ii)}{\le} \frac{ G_R^2}{2\tp_c^2 N^2} \Big( \sum_{r=1}^R\txi_{cr}^2
	+ \sum_{s=R+1}^\infty\hsigma_s^2 \txi_{cs}^2 /\hsigma_R^2 \Big)^2. 
\end{align*}
In the above, (i) used that
$\frac{\hsigma_r^2}{\hsigma_r^2 - \hsigma_s^2} \le G_R$ and $\hsigma_r^2 - \hsigma_s^2 \ge \hsigma_R^2 - \hsigma_{R+1}^2$ for any $r\le R$ and $s>R$. Inequality (ii) used that
$ab \le (\frac{a+b}{2})^2$ for real numbers $a$ and $b$.
For the sampling probability $\{p_n\}_{n=1}^N$ satisfying $p_n \ge p_n^{\text{Exact}} / \beta$, we have
$$
p_n \ge \frac{1}{\beta} \Big( \sum_{r=1}^R\hxi_{nr}^2
+ \sum_{s=R+1}^\infty\hsigma_s^2 \hxi_{ns}^2 /\hsigma_R^2 \Big) /
\sum_{m=1}^N \Big( \sum_{r=1}^R\hxi_{mr}^2
+ \sum_{s=R+1}^\infty\hsigma_s^2 \hxi_{ms}^2 /\hsigma_R^2 \Big),
$$
as $\sum_{n=1}^N \hxi_{nr}^2/N = 1$, we have
\begin{align}
	p_n &\ge \frac{1}{\beta\cdot N} \Big( \sum_{r=1}^R\hxi_{nr}^2
	+ \sum_{s=R+1}^\infty\hsigma_s^2 \hxi_{ns}^2 /\hsigma_R^2 \Big) /
	\big( R
	+ \sum_{s=R+1}^\infty\hsigma_s^2  /\hsigma_R^2 \big) \nonumber\\
	&  \ge \frac{1}{\beta\cdot N \big( R
		+ \Delta_R \big)} \Big( \sum_{r=1}^R\hxi_{nr}^2
	+ \sum_{s=R+1}^\infty\hsigma_s^2 \hxi_{ns}^2 /\hsigma_R^2 \Big).
	\label{eqn:proof:projprob}
\end{align}
As a result, 
$\|  \frac{\sZ_c }{N\tp_c}\|^2 \le  \frac{ G_R^2\beta^2}{2} \big( R + \Delta_R\big)^2$, which is
\begin{equation}\label{eqn:projectLinearOperatorUpper}
	\|  \frac{\sZ_c}{N\tp_c} \| \le L = G_R\beta \big( R + \Delta_R\big) /\sqrt{2}.
\end{equation}

Next, we bound the variance of the summation $\sS=\sum_{c=1}^C \frac{\sZ_c }{N\tp_c}$. Note that
\begin{align*}
	\frac{\sZ_c\sZ_c^*}{N^2\tp_c^2} = \frac{\sZ_c^*\sZ_c}{N^2\tp_c^2}=&\sum_{r,r'=1}^R \htheta_r \otimes \htheta_{r'} \Big[
	\sum_{s=R+1}^\infty \frac{\hsigma_r\hsigma_{r'}\hsigma_s^2}{(\hsigma_r^2 - \hsigma_s^2)(\hsigma_{r'}^2 - \hsigma_s^2)}
	\times \frac{\txi_{cr}\txi_{cr'} \txi_{cs}^2}{\tp_c^2 N^2}\Big] \\ 
	& + \sum_{s,s'=R+1}^\infty \htheta_s \otimes \htheta_{s'} \Big[
	\sum_{r=1}^R \frac{\hsigma_s\hsigma_{s'}\hsigma_r^2}{(\hsigma_r^2 - \hsigma_s^2)(\hsigma_{r}^2 - \hsigma_{s'}^2)}
	\times \frac{\txi_{cr}^2\txi_{cs} \txi_{cs'}}{\tp_c^2 N^2}\Big].
\end{align*}
Because all the $\sZ_c$'s are independent with zero mean, and $ \sS^*\sS$ is positive semi-definite and self-adjoint, its
largest eigenvalue is upper bounded by its trace 
\begin{align}
	&\|\Expect\sS^2\|  \le C\cdot \mathrm{tr} (\Expect \frac{\sZ_c\sZ_c^*}{N^2\tp_c^2}) \nonumber\\
	= &C \sum_{r=1}^R \Big[
	\sum_{s=R+1}^\infty \frac{\hsigma_r^2\hsigma_s^2}{(\hsigma_r^2 - \hsigma_s^2)^2}
	\times \sum_{n=1}^N \frac{\hxi_{nr}^2\hxi_{ns}^2}{p_n N^2}\Big] 
	+ C \sum_{s=R+1}^\infty \Big[
	\sum_{r=1}^R \frac{\hsigma_r^2\hsigma_s^2}{(\hsigma_r^2 - \hsigma_s^2)^2}
	\times \sum_{n=1}^N \frac{\hxi_{nr}^2\hxi_{ns}^2}{p_n N^2}\Big] \nonumber\\ 
	= &2C \sum_{r=1}^R 
	\sum_{s=R+1}^\infty \frac{\hsigma_r^2\hsigma_s^2}{(\hsigma_r^2 - \hsigma_s^2)^2}
	\times \sum_{n=1}^N \frac{\hxi_{nr}^2\hxi_{ns}^2}{p_n N^2}. \label{eqn::projectTraceUpper}
\end{align}
It follows that
\begin{align}
	\|\Expect\sS^2\| & \le 2CG_R^2  \sum_{n=1}^N \sum_{r=1}^R  \sum_{s=R+1}^\infty\frac{\hxi_{nr}^2(\hsigma_s^2\hxi_{ns}^2/\hsigma_R^2)}{p_n N^2}\nonumber\\ 
	&\le 2CG_R^2  \sum_{n=1}^N  \frac{1}{p_n N^2} \Big(\sum_{r=1}^R\hxi_{nr}^2  \Big)
	\Big( \sum_{s=R+1}^\infty \hsigma_s^2\hxi_{ns}^2/\hsigma_R^2 \Big) \nonumber\\ 
	&\le CG_R^2  \sum_{n=1}^N  \frac{1}{2p_n N^2} \Big(\sum_{r=1}^R\hxi_{nr}^2 +
	\sum_{s=R+1}^\infty \hsigma_s^2\hxi_{ns}^2/\hsigma_R^2 \Big)^2 \nonumber\\ 
	&\le \frac{\beta CG_R^2}{2N^2}  \Bigg[\sum_{n=1}^N  \Big(\sum_{r=1}^R\hxi_{nr}^2 +
	\sum_{s=R+1}^\infty \hsigma_s^2\hxi_{ns}^2/\hsigma_R^2 \Big)^2\Bigg]^2  \nonumber\\ 
	&\le  \frac{\beta CG_R^2}{2} \big(R + \Delta_R\big)^2.
	\label{eqn:projectVarUpper}
\end{align}
The eigenvalue of $\Expect \sS^2$ is lower bounded by any of the coefficient of $\htheta_r\otimes \htheta_r$
for $r = 1,\cdots, R$. It can also be lower bounded by their average
\begin{align*}
	\|\Expect \sS^2\| \ge 
	\frac{C}{R} \sum_{r=1}^R \langle\htheta_r\otimes \htheta_r,
	\Expect \sZ_c\sZ_c^*\rangle
	=
	\frac{C}{R} \sum_{r=1}^R \Big[
	\sum_{s=R+1}^\infty \frac{\hsigma_r^2\hsigma_s^2}{(\hsigma_r^2 - \hsigma_s^2)^2}
	\times \sum_{n=1}^N \frac{\hxi_{nr}^2\hxi_{ns}^2}{p_n N^2}\Big]
\end{align*}
Combining this with the trace evaluation in~\eqref{eqn::projectTraceUpper}, the intrinsic dimension has upper bound $ \mathrm{intdim} ( \Expect \sS^2)\le 2R$.
Therefore, applying Lemma~\ref{lemma:mainConcentration2} with~\eqref{eqn:projectLinearOperatorUpper} and~\eqref{eqn:projectVarUpper}, we get
\begin{align*}
	\Prob \big( \| L_R(\sE) \| \ge \epsilon \big) &=
	\Prob \big( \| \sS \| \ge  C\epsilon \big)  \\
	\le & 8R \times 
	\exp\Big( -\frac{C\epsilon^2/2}{\frac{\beta G_R^2}{2} \big(R + \Delta_R\big)^2+ 
		[G_R\beta \big( R + \Delta_R\big) /\sqrt{2}] \times  \epsilon /3} \Big).
\end{align*}

\subsection{Proof of Lemma~\ref{lemma:projection:part2} }
Notice that
$$
\frac{1}{\hsigma_{1}^2}\left[ \tsC_{XX} - \hsC_{XX} \right]
= 	\frac{1}{\hsigma_{1}^2} \times
\frac{1}{C}\sum_{c=1}^C \Big[\frac{1}{\tp_cN}\tx_c\otimes \tx_c
- \hsC_{XX} \Big] = \frac{1}{C} \sum_{c=1}^C \sZ_c = \frac{1}{C}\sS\,. 
$$
In the above, we have set $\sZ_c  :=\frac{1}{\hsigma_{1}^2}  \Big[\frac{1}{\tp_cN}\tx_c\otimes \tx_c - \hsC_{XX} \Big] $.
Our goal is to apply Lemma~\ref{lemma:mainConcentration2} to 
the summation $\sS =  \sum_{c=1}^C \sZ_c$. 
Each summand has the operator norm upper bound
\begin{equation} \label{proof:corrCx:ZcL}
	\| \sZ_c \| \le \frac{1}{N\tp_c} \cdot \frac{\| \tx_c\|^2}{\hsigma_{1}^2}+1.
\end{equation}
According to the operator norm bound~\eqref{proof:corrCx:ZcL} and the probability bound~\eqref{eqn:proof:projprob}, we can derive that
\begin{align*}
	\| \sZ_c\| &\le \beta (R + \Delta_R) \frac{\| \tx_c\|^2/\hsigma_1^2}{\sum_{r=1}^R\txi_{cr}^2
		+ \sum_{s=R+1}^\infty\hsigma_s^2 \txi_{cs}^2 /\hsigma_R^2 } +  1 \\
	&= \beta (R + \Delta_R) \frac{
		\sum_{r=1}^\infty\hsigma_r^2 \txi_{cr}^2 /\hsigma_1^2
	}{\sum_{r=1}^R\txi_{cr}^2
		+ \sum_{s=R+1}^\infty\hsigma_s^2 \txi_{cs}^2 /\hsigma_R^2 } +  1 \\
	&\le \beta (R + \Delta_R)+ 1.
\end{align*}
with the sampling probability $\{p_n\}_{n=1}^N$ satisfying $p_n \ge p_n^{\text{Exact}} / \beta$.

For the variance of $\sS$, it holds that
\begin{align}
	\Expect [\sZ_c^*\sZ_c]  &= \frac{1}{\hsigma_{1}^4 } \Big[
	\sum_{n=1}^N 	\frac{1}{p_nN^2} \langle x_n, x_n \rangle x_n \otimes x_n
	-\hsC_{XX}^*\hsC_{XX}	\Big] \nonumber \\
	& \preceq \frac{1}{\hsigma_{1}^4 }  \sum_{n=1}^N 	\frac{1}{p_nN^2} \langle x_n, x_n \rangle x_n \otimes x_n.
	\label{corrCx:Vbound}
\end{align}
Based on~\eqref{corrCx:Vbound} and~\eqref{eqn:proof:projprob}, we can get that
$$
\Expect [\sZ_c^*\sZ_c] \preceq \frac{1}{N} \sum_{n=1}^N \beta (R + \Delta_R) \times \frac{x_n \otimes x_n}{\hsigma_{1}^2}
=  [\beta  (R + \Delta_R) \ /\hsigma_{1}^2] \hsC_{XX}.
$$
Set $\sV = C[\beta  (R + \Delta_R) \ /\hsigma_{1}^2] \hsC_{XX}$, 
we have $\|\sV\| = \beta C (R + \Delta_R)$, and 
$ \mathrm{intdim}(\sV) = \Delta_0$.
Therefore, by Lemma~\ref{lemma:mainConcentration2}, for $\epsilon$ satisfying 
$\epsilon \cdot C > \sqrt{C\beta (R + \Delta_R)} + [\beta(R+\Delta_R) + 1] /3$,
\begin{align*}
	\Prob\Big(\left\Vert \tsC_{XX} - \hsC_{XX}\right\Vert \ge \epsilon \hsigma_{1}^2\Big) =&
	\Prob\Big(\left\Vert \sS\right\Vert \ge  C\epsilon\Big)  \\
	\le & 
	4 \Delta_0 \exp\Big(- \frac{\epsilon^2C/2}{\beta (R + \Delta_R) + [\beta (R + \Delta_R)+1]\epsilon/3}\Big).
\end{align*}

\section{Proof for Section~\ref{sec:flr} }\label{appendix:flr:proof}

\subsection{Proof of Lemma~\ref{lemma:regErrorDeco} }

In classical linear regression,  the residual is perpendicular to the columns of the design matrix. For functional regression with 
scalar response, this residual $\mY^\perp$ is perpendicular to the span of the leadning R scores. In particular,  
define $\vxi_r = ( \hxi_{1r},\cdots, \hxi_{Nr})^T$ to be a  vector of the $r$-th scores for all samples, 
then $ \mY^{\perp}$ is perpendicular to the span of the first $R$ score vectors, which is 
$\text{span}\{\vxi_1, \cdots, \vxi_R\}$. 
This implies $\mY$ can be decomposed  into the two parts
$\mY = \mY^{\parallelsum}+ \mY^{\perp}$,
with $ \mY^{\parallelsum}$  inside the span  $\text{span}\{\vxi_1, \cdots, \vxi_R\}$. 

Now we verify that the $ \mY^{\perp}$ is perpendicular to the span of the first $R$ score vectors. Note $\mY^{\perp} =\mY - \sT\hPsi= (I - \sT(\sT^*\sT)^+\sT^* ) \mY$. Also, for any $\ma\in\bbR^N$, it holds that
\begin{align*}
	\sT(\sT^*\sT)^+\sT^* \ma &=
	\sT(\sT^*\sT)^+ \Big(\sum_{n=1}^{N} a_n x_n/N\Big) \\
	&= \sT \Big(\sum_{r=1}^R\sum_{n=1}^{N} a_n 
	\langle  \htheta_r, x_n\rangle \htheta_r /\hsigma_r^2/N\Big)  \\
	&= \sT \Big(\sum_{r=1}^R\sum_{n=1}^{N} a_n \hxi_{nr}\htheta_r /(\hsigma_rN)\Big)  \\
	&= \sT \Big(\sum_{r=1}^R\langle \ma, \vxi_r\rangle_N\htheta_r /\hsigma_r \Big)   = \sum_{r=1}^R\langle \ma, \vxi_r\rangle_N \vxi_r.
\end{align*}
As a result, $\sT(\sT^*\sT)^+\sT^*  =  \sum_{r=1}^R \vxi_r \otimes \vxi_r$ is the projection onto $\text{span}\{\vxi_1, \cdots, \vxi_R\}$.
It is also idempotent $(\sT(\sT^*\sT)^+\sT^* )^2 = \sT(\sT^*\sT)^+\sT^* $. It follows that
$$
\sT\hsC_{XX}^+\sT^* \mY^{\perp} = 
\sT(\sT^*\sT)^+\sT^* \mY^{\perp} = 
\sT(\sT^*\sT)^+\sT^* (I - \sT(\sT^*\sT)^+\sT^* ) \mY = \vzero.
$$
This is the last conclusion stated in Lemma~\ref{lemma:regErrorDeco}.

For each predictor $x_n$, we define a related element $x_n^+ = \sum_{r=1}^R \hxi_{nr} \htheta_r/\hsigma_r$. Define a new mapping 
$\sT^+:\bbR^N\to \sH_X$ such that for any  $\ma\in\bbR^N$,
$\sT^+$ maps $\ma$  to $
\sT^+\ma := \frac{1}{N}\sum_{n=1}^N a_n x_n^+ $. 
Similarly, we can also deduce that $\sT\sT^+ = \sum_{r=1}^R \vxi_r \otimes \vxi_r$. In fact, for any $\ma\in\bbR^N$, we have $\sT\sT^+\ma\in \bbR^N$.
The $m$-th element of $\sT\sT^+\ma$ is
\begin{align*}
	\langle x_m, \sT^+\ma\rangle &= 
	\langle x_m, \frac{1}{N}\sum_{n=1}^N a_n x_n^+\rangle
	=\frac{1}{N}\sum_{n=1}^N a_n\langle x_m,  x_n^+\rangle
	= \frac{1}{N}\sum_{n=1}^N a_n\Big(\sum_{r=1}^{R}\hxi_{nr}\hxi_{mr}\Big)\\
	&=\sum_{r=1}^{R}\Big(\frac{1}{N}\sum_{n=1}^N a_n\hxi_{nr}\Big)\hxi_{mr}
	= \sum_{r=1}^{R}\langle \ma, \vxi_r\rangle_N \hxi_{mr}.
\end{align*}
The above equation verifies that $\sT\sT^+ = \sum_{r=1}^R \vxi_r \otimes \vxi_r$ is exactly the projection operator onto $\text{span}\{\vxi_1,
\cdots, \vxi_R\}$. Therefore, we can get
$\mY^{\parallelsum} = \sT\sT^+\mY$,  $ \mY^{\perp} = (I -\sT\sT^+)\mY$. 
and $\sT\sT^+  \mY^{\perp}=\vzero$.

Our target is to analyze the prediction error $\mF =  \sT\widetilde{\Psi}-\sT\hPsi$, which is
$$
\mF = \sT\tsC_{XX}^+\tilde{z}
- \sT\hsC_{XX}^+\hat{z} =
\sT(\sT^*\sD^*\sD\sT)^+\sT^*\sD^*\sD\mY
- \sT(\sT^*\sT)^+\sT^*\mY.
$$
Plug in the decomposition $\mY = \mY^{\parallelsum}+ \mY^{\perp}$, we have
\begin{align*}
	\mF = &\sT(\sT^*\sD^*\sD\sT)^+\sT^*\sD^*\sD\mY^{\parallelsum} + \sT(\sT^*\sD^*\sD\sT)^+\sT^*\sD^*\sD\mY^{\perp} \nonumber\\
	& \qquad  -\sT(\sT^*\sT)^+\sT^*\mY^{\parallelsum}
	-\sT(\sT^*\sT)^+\sT^*\mY^{\perp}\nonumber \\
	\stackrel{(i)}{=} &\sT(\sT^*\sD^*\sD\sT)^+\sT^*\sD^*\sD\sT\sT^+\mY + \sT(\sT^*\sD^*\sD\sT)^+\sT^*\sD^*\sD\mY^{\perp} \nonumber\\
	& \qquad  -\sT(\sT^*\sT)^+\sT^*\sT\sT^+\mY\nonumber \\
	\stackrel{(ii)}{=} &\sT\tsP_R\sT^+\mY + \sT\tsC^+_{XX}\sT^*\sD^*\sD\mY^{\perp}   -\sT\hsP_R\sT^+\mY\nonumber \\
	= &\sT(\tsP_R-\hsP_R)\sT^+\mY \\
	&\qquad + \sT(\tsC^+_{XX} - \hsC^+_{XX})\sT^*(\sD^*\sD)\mY^{\perp} 
	+  \sT \hsC^+_{XX}\sT^*(\sD^*\sD)\mY^{\perp} .
\end{align*}
In the above, equality~(i) uses the relations $\sT(\sT^*\sT)^+\sT^* \mY^{\perp} = \vzero$ and $\mY^{\parallelsum} = \sT\sT^+\mY$.
The  equation~(ii) is based on the following equations 
\begin{align*}
	\tsP_R &=\tsC_{XX}^+ \tsC_{XX}= (\sT^*\sD^*\sD\sT)^+\sT^*\sD^*\sD\sT, \\
	\hsP_R &= \hsC_{XX}^+ \hsC_{XX}=(\sT^*\sT)^+\sT^*\sT.
\end{align*}

\subsection{Proof of Theorem~\ref{thm:regression}}

We prove a stronger result of Theorem~\ref{thm:regression2} below. Theorem~\ref{thm:regression} directly follows from 
Theorem~\ref{thm:regression2} by simplifying notations. 

\begin{theorem} \label{thm:regression2}
	Under Condition~\ref{assumption:eigenvalue},
	set $Z_1 =\hsigma_1^6K_R / ( g_R ^2\hsigma_R^2) $
	and $Z_2 = \beta(R+\Delta_R)$. In addition, 
	\begin{align*}
		V_1 &= G_R^2 Z^2_2/\beta, \quad 
		L_1 = G_RZ_2/2, \\
		V_2 &=  2 +   (2+G_R^2/2)Z_2^2/\beta,\quad 
		L_2 = \sqrt{2R}+ ( \sqrt{2}+ G_R/\sqrt{2} )Z_2,\\
		V_3 &= Z_2,\quad L_3 = Z_2+1.
	\end{align*}
	Suppose $C$ subsamples are obtained according to the  probability satisfying
	\begin{equation} 
		p_n \ge \frac{1}{\beta}
		\frac{\sum_{r=1}^R \hxi_{nr}^2 + \|(I - \hsP_R) x_n \|^2  / \hsigma_{R}^2}{\sum_{m=1}^N \big[ \sum_{r=1}^R \hxi_{mr}^2 + \|( I - \hsP_R) x_m \|^2  / \hsigma_{R}^2\big] }
	\end{equation}
	with $\beta\ge 1$. Then for large enough $C$, 
	with probability at least
	\begin{align*}
		&1 - 4(R + \Delta_R)  \exp\Big(-\frac{C\epsilon_1^2}{V_1 +L_1\epsilon_1/3 }\Big)
		-  8R  \exp\Big(-\frac{C\epsilon_1^2/2}{V_2 +L_2\epsilon_1/3 }\Big) \\
		& \qquad \qquad \qquad - 4 \Delta_0 \exp\Bigg(- \frac{C\epsilon^2_2/2}{V_3+ (V_3 + 1)\epsilon_2/3}\Bigg) - 3\epsilon_3^2/C
	\end{align*}
	it holds that
	\begin{align} 
		\| \sT\hPsi - \sT\widetilde{\Psi} \|_N & \le 
		\big( \epsilon_1 + (\hsigma_R/\hsigma_1) Z_1\cdot 
		\epsilon_2^2\big)  \|\mY \|_N \nonumber \\
		& \qquad +
		\big(\epsilon_1+  2Z_1 \cdot  \epsilon_2^2\big)  (1 + \sqrt{Z_2}/\epsilon_3)  \|\mY^{\perp}\|_N
		+(\sqrt{Z_2}/\epsilon_3) \|\mY^{\perp}\|_N .	\label{thm:regression:boundLong}
	\end{align}
	In the above, $\epsilon_1$ satisfies $ \epsilon_1\ge  \max_{j=1,2}\{\sqrt{V_j/C} + L_j/(3C)\}$. Additionally, $\epsilon_2$ satisfies $ \epsilon_2\ge  \sqrt{V_3/C} + L_3/(3C)$ and $\epsilon_2\le  g_R /3$. Lastly,  $\epsilon_3$ is any strictly positive number.
\end{theorem}

Theorem~\ref{thm:regression2} provides characterization of the first order and higher order error by expanding $\tsP_R-\hsP_R$ and $\tsC_{XX}^+-\hsC_{XX}^{+}$.
Theorem~\ref{thm:regression} follows  by setting $\epsilon=\epsilon_1=\epsilon_2=\sqrt{Z_2}/\epsilon_3$ with $\epsilon\le 1$. 
We can also find that bound~\eqref{thm:regression:boundLong} simplifies to
\begin{align*} 
	\| \sT\widetilde{\Psi}-\sT\hPsi  \|_N & \le 
\big( \epsilon+ (\hsigma_R/\hsigma_1) Z_1 \epsilon^2\big)  \|\mY \|_N 	+
	2\big(\epsilon+  2Z_1\epsilon^2 \big)   \|\mY^{\perp}\|_N +\epsilon \|\mY^{\perp}\|_N \\
	&\le \big( \epsilon+ \hsigma_RZ_1/\hsigma_1 \epsilon^2 \big)  \|\mY \|_N 
	+\big(3\epsilon+  4Z_1\epsilon^2 \big)   \|\mY^{\perp}\|_N .
\end{align*}
Also, by noticing that
\begin{align*}
	\max\{V_1,V_2,V_3\} \le V:=&  2 +   (2+G_R^2)Z_2^2/\beta, \\
	\max\{L_1,L_2,L_3\} \le L:=& \sqrt{2R}+ ( \sqrt{2}+ G_R/\sqrt{2} )Z_2,
\end{align*}
we will be able to reach the conclusion of Theorem~\ref{thm:regression}.

Theorem~\ref{thm:regression2} is proved by a detailed analysis of the prediction error terms~\eqref{eqn:regression:errorterm23} in Lemma~\ref{lemma:regErrorDeco}. The first two terms are bounded by the next two lemmas.

\begin{lemma} \label{lemma:regression:part1}
	Define quantities
	\begin{align*}
		&V_1 = \beta  G_R^2 \big(R + \Delta_R\big)^2, \quad 
		L_1 = G_R\beta \big( R + \Delta_R\big)/2 , \\
		&V_3  = \beta (R + \Delta_R),  \quad
		L_3  = \beta (R + \Delta_R)+1.
	\end{align*}
	For $\epsilon_1,\epsilon_2$ satisfying $\epsilon_1  \ge \sqrt{V_1/C} + L_1/(3C)$, and
	$\epsilon_2  \ge \sqrt{V_3/C} + L_3/(3C)$,
	with probability at least 
	\begin{align*}
		1 - 4(R + \Delta_R)  \exp\Big(-\frac{C\epsilon^2_1}{V_1 +L_1\epsilon_1/3 }\Big)
		- 4 \Delta_0 \exp\Bigg(- \frac{\epsilon^2_2C/2}{V_3+ L_3\epsilon_2/3}\Bigg).
	\end{align*}
	it holds that
	\begin{align*}
		\|\sT(\tsP_R-\hsP_R)\sT^+\mY\|_N \le \big[ \epsilon_1 +K_R\hsigma_1^5  /(\hsigma_R  g_R^2)\cdot 
		\epsilon_2^2\big] 
		\|\mY \|_N
	\end{align*}
\end{lemma}

\begin{lemma}  \label{lemma:regression:part2}
	Define quantities
	\begin{align*}
		V_2 &= 2 +  \beta (2+G_R^2/2)(R+\Delta_R)^2, \\
		L_2 &= \sqrt{2R}+ ( \sqrt{2}+ G_R /\sqrt{2})\beta(R + \Delta_R) \\
		V_3 & = \beta (R + \Delta_R),  \quad
		L_3  = \beta (R + \Delta_R)+1.
	\end{align*}
	For large enough $C$ and $\epsilon_1,\epsilon_2$ satisfying $\epsilon_1  \ge \sqrt{V_2/C} + L_2/(3C)$,
	$\epsilon_2  \ge \sqrt{V_3/C} + L_3/(3C)$ and $\epsilon_2\le  g_R/3$,
	with probability at least 
	\begin{align*}
		1 - 8R  \exp\Big(-\frac{C\epsilon^2/2}{V_2 +L_2\epsilon/3 }\Big)
		- 4 \Delta_0 \exp\Bigg(- \frac{\epsilon^2_2C/2}{V_3+ L_3\epsilon_2/3}\Bigg).
	\end{align*}
	it holds that
	\begin{align*}
		\| \sT(\tsC_{XX}^+ - \hsC_{XX}^+)\sT^*\sD^*\sD\mY^{\perp}\|_N \le 
		\big(\epsilon_1+ 2Z_1\cdot  \epsilon^2_2\big)\big(1 + \sqrt{Z_2}/\epsilon_3\big)   \|\mY^{\perp}\|_N.
	\end{align*}
\end{lemma}

To control the last term of the prediction error term in Lemma~\ref{lemma:regErrorDeco}, we need further decomposition of the two operators
$\sT$ and $\sT^*$.  Specifically, we make the decomposition
$\sT = \sU_1\sV_1$ with $\sU_1$ and $\sV_1$ defined as follows. For any $u\in\sH_X$, $\sV_1 u = (\hsigma_1 \langle \htheta_1, u\rangle, \cdots, 
\hsigma_N \langle \htheta_N, u\rangle)^T$. Meanwhile, 
the operator $\sU_1$ maps $\ma \in\bbR^N$ to
$\sU_1\ma = \sum_{n=1}^N a_n\vxi_n \in\bbR^N$. They have operator norm $\|\sU_1\| =\sup_{\| \ma\|_N\le 1} \| \sU_1\ma\|_N  =\sqrt{N}$ and $\|\sV_1\| =\sup_{\| u\|\le 1} \| \sV_1 u\|_N= \hsigma_1/\sqrt{N}$, respectively.

In addition, the decomposition that $\sT^* = \sV_2 \sU_2$ is constructed. For any $\ma \in\bbR^N$, $\sU_2$ maps $\ma$ to another vector
$\sU_2\ma = \mb \in\bbR^N$, whose vector elements are
$$B_r = \begin{cases}
	N\langle \vxi_r, \ma\rangle_N, & \text{ if } r\le R; \\
	N(\hsigma_r/\hsigma_{R})\langle \vxi_r, \ma\rangle_N, & \text{ if } r > R. \\
\end{cases}
$$ 
In the above, its element for $r>R$ is reweighted by the eigenvalue such that
$\| \sU_2\sD^*\sD \mY^{\perp} \|_N$ can be controlled as in~\eqref{eqn:lemma:yperpResult2}. 
In addition, $\sV_2$ maps $\ma\in \bbR^N$ to an element in $\sH_X$, such that
$$\sV_2 \ma =(1/N) \sum_{r=1}^R\hsigma_r a_r \htheta_r +
(\hsigma_{R} /N)\sum_{r=R+1}^N  a_r \htheta_r.$$ 
It is easy to verify their operator norm as $\|\sU_2\| = \sqrt{N}$ and $\|\sV_2\| = \hsigma_1/\sqrt{N}$.

\begin{lemma}    \label{lemma:regression:part3}
	With the sampling probability $\{p_n\}_{n=1}^N$ satisfying $p_n \ge p_n^{\text{Exact}} / \beta$, each of the following inequality holds with probability at least $1-\epsilon_3^2/C$,
	\begin{align}
		\| \sT^* \sD^*\sD \mY^{\perp} \| &\le \hsigma_1 \Big(1+
		\sqrt{Z_2} /\epsilon_3 \Big)\| \mY^{\perp}\|_N,
		\label{eqn:lemma:yperpResult1}\\
		\| \sU_2\sD^*\sD \mY^{\perp} \|_N /\sqrt{N} &\le\Big(1+
		\sqrt{Z_2} /\epsilon_3 \Big) \| \mY^{\perp}\|_N,  \label{eqn:lemma:yperpResult2}  \\
		\| \sT \hsC_{XX}^+ \sT^*\sD^*\sD \mY^{\perp} \|_N  &\le 
		\sqrt{Z_2} \| \mY^{\perp}\|_N /\epsilon_3,  \label{eqn:lemma:yperpResult3} 
	\end{align}
	where  $Z_2 = \beta(R+\Delta_R)$. 
\end{lemma}

Theorem~\ref{thm:regression2} is established by combining the above three lemmas. Their proof is presented in the following subsections. The proof of Lemma~\ref{lemma:regression:part3} is presented firstly.

\subsection{Proof of Lemma~\ref{lemma:regression:part3} }

We only present the proof of \eqref{eqn:lemma:yperpResult1}. The other two results~\eqref{eqn:lemma:yperpResult2} and~\eqref{eqn:lemma:yperpResult3}
can be derived similarly. 
Notice that
$$
\sT^*\sD^*\sD\mY^{\perp} - \sT^*\mY^{\perp} = 
\frac{1}{NC} \sum_{c=1}^C\frac{\tx_c \tilde{Y}_c^{\perp} }{\tp_c} - 
\frac{1}{N}\sum_{n=1}^N x_n Y_n^{\perp}.
$$
It has zero expectation with respect to the subsampling process.
We  can also compute its expected squared norm as
\begin{align*}
	\Expect \|\sT^*\sD^*\sD\mY^{\perp} - \sT^*\mY^{\perp} \|^2 &=
	\frac{1}{C}\cdot  \Expect \Big\| 
	\frac{1}{N}\frac{\tx_1 \tilde{Y}_1^{\perp} }{\tp_1} - 
	\frac{1}{N}\sum_{n=1}^N x_n Y_n^{\perp} \Big\|^2 \\
	& \le \frac{1}{CN^2} \sum_{n=1}^N \frac{\| x_n\|^2 (Y_n^{\perp})^2 }{p_n}  \\
	& \le \frac{\hsigma_1^2\beta(R+\Delta_R)}{CN} \sum_{n=1}^N 
	\frac{(\| x_n\|^2/\hsigma_1^2) (Y_n^{\perp})^2 }{\sum_{r=1}^R \hxi_r^2 + \| (I - \hsP_R) x_n\|^2/\hsigma^2_R}  \\
	& = \hsigma_1^2\beta(R+\Delta_R) \| \mY^{\perp}\|_N^2 / C.
\end{align*}
As a direct consequence of Markov inequality, for $\epsilon_3 > 0$, with probability at least $1 - \epsilon_3^2/C$ 
it holds that
$$
\|\sT^*\sD^*\sD\mY^{\perp} - \sT^*\mY^{\perp} \|^2 \le (1/\epsilon^2_3)\times
\hsigma_1^2\beta(R+\Delta_R) \| \mY^{\perp}\|_N^2.
$$
Together with $\|\sT^* \mY^{\perp}\| \le \|\sT^*\| \cdot \|  \mY^{\perp}\|_N = \hsigma_1 \|  \mY^{\perp}\|_N$, 
with probability at least $1 - \epsilon_3^2/C$  it holds that
\begin{align*}
	\| \sT^* \sD^*\sD \mY^{\perp} \| & \le \|\sT^* \mY^{\perp}\| +
	\| \sT^* \sD^*\sD \mY^{\perp} - \sT^* \mY^{\perp} \|  \\
	& \le \hsigma_1 \|  \mY^{\perp}\|_N
	+ (\hsigma_1/\epsilon_3) \| \mY^{\perp}\|_N \sqrt{\beta(R+\Delta_R)} \\
	& \le  \hsigma_1 \|  \mY^{\perp}\|_N 
	(1 + \sqrt{Z_2}/\epsilon_3)
\end{align*}
Thereby, we have proven~\eqref{eqn:lemma:yperpResult1}.

\subsection{Proof of Lemma~\ref{lemma:regression:part1} }
Recall the difference $\tsP_R-\hsP_R = L_R(\sE) + S_R(\sE)$ is expressed as the sum of a linear part and higher order part. By Lemma~\ref{lemma:lsbound}, 
$\Vert S_R(\sE)  \Vert \le K_R\left( \Vert \sE\Vert  / g_R  \right)^2$, it follows that
\begin{align*}
	\|\sT(\tsP_R-\hsP_R)\sT^+\mY\|_N \le& \| \sT L_R(\sE)\sT^+\mY\|_N + \| \sT S_R(\sE)\sT^+ \mY\|_N \\
	\le &  \| \sT L_R(\sE)\sT^+\| \|\mY\|_N + \| \sT\|  \|S_R(\sE)\| \| \sT^+\| \| \mY\|_N \\
	\le &  \| \sT L_R(\sE)\sT^+\| \|\mY\|_N +  \big[\hsigma_1K_R /(\hsigma_R  g_R^2)\big] \|\sE\|^2 \| \mY\|_N. 
\end{align*}
From here we know that
\begin{align}
	&\Prob\Big(\|\sT(\tsP_R-\hsP_R)\sT^+\mY\|_N 
	\ge \big[ \epsilon_1 +K_R\hsigma_1^5  /(\hsigma_R  g_R^2)\cdot 
	\epsilon_2^2\big]  \|\mY \|_N  \Big) \nonumber \\
	\le & \Prob\Big(\| \sT L_R(\sE)\sT^+\| \ge \epsilon_1\Big)+
	\Prob\Big(\|\sE\| \ge \epsilon_2\hsigma_1^2 \Big). \label{eqn:regPart1:2prob}
\end{align}
The second probability in~\eqref{eqn:regPart1:2prob} is controled by Lemma~\ref{lemma:projection:part2}.
We now proceed to control the first probability. 

Based on the result of Lemma~\ref{lemma:pcaLinearSum}, 
$$
\sT L_R(\sE)\sT^+ = \frac{1}{C} \sum_{c=1}^C \sT  \sZ_c\sT^+ 
= \frac{1}{C} \sum_{c=1}^C \sZ'_c, $$
with $\sZ'_c = \sT  \sZ_c\sT^+$
and
$ \sZ_c =  \sum_{r=1}^R
\sum_{s=R+1}^\infty \frac{\hsigma_r\hsigma_s}{\hsigma_r^2 - \hsigma_s^2} \times
\frac{\txi_{cr}\txi_{cs}}{\tp_c N} \times 
\big[\htheta_r \otimes \htheta_s +
\htheta_s \otimes \htheta_r\big]$. We will apply Lemma~\ref{lemma:mainConcentration} to the summation 
$\sS' = \sum_{c=1}^C \sZ'_c$ to conclude the proof of this lemma.

Note that for $r\le R$, $s > R$ and any $\ma=(a_1,\cdots, a_N)\in\bbR^N$, it holds that
\begin{align}
	& \sT (\htheta_r\otimes \htheta_s) \sT^+ \ma = \sT (\htheta_r\otimes \htheta_s) 
	\Big(\frac{1}{N}\sum_{n=1}^N a_n x_n^+\Big) =
	(\sT \theta_s) \Big(\frac{1}{N}\sum_{n=1}^N a_n \langle\htheta_r, x_n^+\rangle\Big) \nonumber \\
	= &\begin{pmatrix} \langle x_1, \htheta_s\rangle \\ \vdots \\ \langle x_N, \htheta_s\rangle \end{pmatrix}
	\Big(\frac{1}{N}\sum_{n=1}^N a_n \hxi_{nr} /\hsigma_r\Big)
	= \begin{pmatrix} \hsigma_s \hxi_{1s} \\ \vdots \\\hsigma_s \hxi_{Ns} \end{pmatrix} 
	\Big(\frac{1}{N} \sum_{n=1}^N a_n \hxi_{nr} /\hsigma_r\Big) \nonumber\\
	=& \frac{\hsigma_s}{\hsigma_r} \langle \vxi_r, \ma\rangle_N \times \vxi_s
	= \frac{\hsigma_s}{\hsigma_r} (\vxi_r \otimes \vxi_s) \ma, \label{eqn:regression:lemma1:tthetat}
\end{align}
where $\vxi_r = (\hxi_{1r}, \cdots, \hxi_{Nr} )^T$ is the vector of the $r$-th score for all samples.
Note $\vxi_r$ is a unit vector in $\bbR^N$ with respect to the norm 
$\|\cdot\|_N$. 
For $r\le R$ and $s > R$,
the above implies that $\sT (\htheta_r\otimes \htheta_s) \sT^+
=  \frac{\hsigma_s}{\hsigma_r} (\vxi_r \otimes \vxi_s)$ and by similar derivation we can find that 
$\sT (\htheta_s\otimes \htheta_r) \sT^+ =  0$.
As a result,
\begin{equation} \label{eqn:regressionProjIndividual}
	\sZ'_c = \sT \sZ_c \sT^+=  \sum_{r=1}^R
	\sum_{s=R+1}^\infty \frac{\hsigma_s^2}{\hsigma_r^2 - \hsigma_s^2} \Big(
	\frac{\txi_{cr}\txi_{cs}}{\tp_c N}\Big) \vxi_r \otimes \vxi_s
\end{equation}
Apply almost identical argument that leading to~\eqref{eqn:projectLinearOperatorUpper}, we have
$$\|  \sZ'_c \| \le L = G_R\beta \big( R + \Delta_R\big) /2.$$
This upper bound is smaller than  that in~\eqref{eqn:projectLinearOperatorUpper} by a factor of $\sqrt{2}$, as 
we only have one tensor product term in~\eqref{eqn:regressionProjIndividual}. For the variance terms, we need to derive bounds for $\Expect \sS' (\sS')^*$ and
$\Expect (\sS')^*\sS' $. For the latter, we have
\begin{align}
	\Expect(\sZ'_c)^*  \sZ'_c &= 
	\sum_{r,r'=1}^R \Bigg[\sum_{s=R+1}^\infty 
	\frac{\hsigma_s^4}{(\hsigma_r^2 - \hsigma_s^2)(\hsigma_{r'}^2 - \hsigma_{s}^2) }
	\Big( \sum_{n=1}^N \frac{\hxi_{ns}^2\hxi_{nr}\hxi_{nr'} }{p_n N^2}\Big) \Bigg] \vxi_r \otimes \vxi_{r'}
\end{align}
Similar to the derivation of~\eqref{eqn:projectVarUpper}, we get the upper bound,
\begin{align*}
	\| \Expect (\sS')^* \sS'\| &\le  
	C\cdot \mathrm{tr}( (\sZ'_c)^*  \sZ'_c ) \\
	& = C\sum_{r=1}^R \sum_{s=R+1}^\infty 
	\frac{\hsigma_s^4}{(\hsigma_r^2 - \hsigma_s^2)^2 }
	\Big( \sum_{n=1}^N \frac{\hxi_{nr}^2\hxi_{ns}^2 }{p_n N^2}\Big)  \\
	&\le  CG_R^2 \sum_{r=1}^R \sum_{s=R+1}^\infty 
	\Big( \sum_{n=1}^N \frac{\hxi_{nr}^2 ( \hsigma_s^2\hxi_{ns}^2/\hsigma_R^2) }{p_n N^2}\Big)  \\
	&   \le  \beta C G_R^2 (R + \Delta_R)^2/4. 
\end{align*}
We can also verify  that $\mathrm{intdim} \big(\Expect (\sS')^* \sS'\big)\le R$. 
As for  $\Expect \sS' (\sS')^* $, we will show that there exists a positive semidefinite operator $\sV$ such that
\begin{equation} \label{eqn:reg:ssupper2}
	\Expect \sS' (\sS')^* \preceq \sV,\,
	\mathrm{intdim} \big(\sV \big) \le \Delta_R, \,
	\text{ and } \|\sV\| \le C\beta RG^2_R (R + \Delta_R) .
\end{equation}
The proof of this result is deferred to latter part of this subsection. 

Note 
\begin{align*}
	\max\{\|\Expect \sS' (\sS')^*\|, \|\Expect (\sS')^*\sS' \| \} 
	\le &\max\{\beta C G_R^2 (R + \Delta_R)^2/4,\beta C RG^2_R (R + \Delta_R)\} \\
	\le & \beta CG_R^2 (R + \Delta_R)^2.
\end{align*}
Now, collecting all of the above results, according to Lemma~\ref{lemma:mainConcentration},  we arrive at the conclusion of this lemma
\begin{align*}
	&\Prob \big( \| \sT L_R(\sE) \sT^+ \| \ge \epsilon_1 \big) \\
	=& 
	\Prob \big( \| \sS' \| \ge C\epsilon_1 \big) \\
	\le& 4(R + \Delta_R) \cdot
	\exp\Big( -\frac{C\epsilon_1^2}{\beta  G_R^2 \big(R + \Delta_R\big)^2+ 
		G_R\beta \big( R + \Delta_R\big)   \epsilon_1 /6} \Big).
\end{align*}

It remains to prove the claim of~\eqref{eqn:reg:ssupper2}, notice that
\begin{align}
	\Expect \sZ'_c (\sZ'_c)^* = 
	\sum_{s,s'=R+1}^\infty \Bigg[\sum_{r=1}^R 
	\frac{\hsigma_s^2\hsigma_{s'}^2}{(\hsigma_r^2 - \hsigma_s^2)(\hsigma_r^2 - \hsigma_{s'}^2) }
	\Big( \sum_{n=1}^N \frac{\hxi_{nr}^2\hxi_{ns}\hxi_{ns'} }{p_n N^2}\Big) \Bigg] \vxi_s \otimes \vxi_{s'} 
\end{align}
Define $f_{rs} = \hsigma_s^2/ (\hsigma_r^2 - \hsigma_s^2)$.
For any $\ma\in\bbR^N$, we have
\begin{align*}
	\langle\ma, \Expect \sZ'_c (\sZ'_c)^* \ma\rangle
	=& 
	\sum_{s,s'=R+1}^\infty \Big[\sum_{r=1}^R 
	f_{rs} f_{rs'}
	\Big( \sum_{n=1}^N \frac{\hxi_{nr}^2\hxi_{ns}\hxi_{ns'} }{p_n N^2}\Big) \Big] \langle \ma, \vxi_s\rangle \langle \ma, \vxi_{s'}\rangle \\
	=& \sum_{r=1}^R \sum_{n=1}^N \frac{\hxi^2_{nr}}{p_n N^2} 
	\Big(\sum_{s=R+1}^\infty f_{rs} \hxi_{ns}\langle \ma, \vxi_s\rangle \Big)^2 \\
	\stackrel{(i)}{\le}   & \frac{\beta (R + \Delta_R)}{N} \sum_{r=1}^R \sum_{n=1}^N
	\Big(\sum_{s=R+1}^\infty f_{rs} \hxi_{ns}\langle \ma, \vxi_s\rangle \Big)^2 \\
	= & \frac{\beta (R + \Delta_R)}{N} \sum_{r=1}^R \sum_{n=1}^N
	\sum_{s,s'=R+1}^\infty f_{rs} f_{rs'}\hxi_{ns}\hxi_{ns'}\langle \ma, \vxi_s\rangle
	\langle \ma, \vxi_s'\rangle\\
	\stackrel{(ii)}{=}  & \beta (R + \Delta_R) \sum_{r=1}^R 
	\sum_{s=R+1}^\infty f^2_{rs} \langle \ma, \vxi_s\rangle^2\\
	\stackrel{(iii)}{\le} & \beta RG_R(R + \Delta_R)/g_R 
	\sum_{s=R+1}^\infty \hsigma^2_{s} \langle \ma, \vxi_s\rangle^2.
\end{align*}
In the above, the inequality~(i) uses the property of the sampling probability $\{p_n\}_{n=1}^N$ satisfying $p_n \ge p_n^{\text{Exact}} / \beta$. Equation~(ii) uses the orthonormality between the score vectors, i.e.,  $\langle \vxi_{s}, \vxi_{s'}\rangle_N = I(s=s')$ where $I(\cdot)$ is the indicator function. Inequality~(iii) uses the fact that
$f_{rs} \le G_R$.

From the above, we can conclude that
$$C\Expect \sZ'_c (\sZ'_c)^* \preceq \sV: = 
C\beta RG_R(R + \Delta_R)/g_R  \sum_{s=R+1}^{\infty}
\hsigma^2_{s}  \vxi_{s} \otimes \vxi_{s}.
$$
For $\sV$, the intrinsic dimension $\mathrm{intdim} (\sV) = \Delta_R$, and its bound for the operator norm  is
$$
\|\sV\| = C\beta RG_R (R + \Delta_R) (\hsigma_{R+1}^2  /g_R) 
\le C\beta RG^2_R (R + \Delta_R) .
$$

\subsection{Proof of Lemma~\ref{lemma:regression:part2} }
For the truncated inverse operator
\begin{align*}
	\tsC_{XX}^+ - \hsC_{XX}^+ & = -\frac{1}{2\pi i} \oint_{\Gamma_R} \frac{1}{\eta} ( \tsC_{XX}-\eta I)^{-1} \mathrm{d}\eta
	+  \frac{1}{2\pi i} \oint_{\Gamma_R} \frac{1}{\eta} (\hsC_{XX}-\eta I )^{-1} \mathrm{d}\eta \\
	&= -\frac{1}{2\pi i} \oint_{\Gamma_R} \frac{1}{\eta} ( \hsC_{XX}-\eta I)^{-1}
	\big[I + \sE  (\hsC_{XX}-\eta I )^{-1} \big]^{-1} \mathrm{d}\eta
	+ \frac{1}{2\pi i} \oint_{\Gamma_R} \frac{1}{\eta} ( \hsC_{XX}-\eta I)^{-1} \mathrm{d}\eta.
\end{align*}
Recall $\sE =  \tsC_{XX} -  \hsC_{XX}$ and  $\sR_{\hsC_{XX}}(\eta) = ( \hsC-\eta I )^{-1}= -\sum_{r=1}^\infty \frac{1}{\eta - \hsigma_{r}^2  } \hsP_r$, then
\begin{align*}
	\tsC_{XX}^+ - \hsC_{XX}^+ 
	&= -\frac{1}{2\pi i} \oint_{\Gamma_R} \frac{1}{\eta} \sR_{\hsC_{XX}}(\eta)
	\big[I + \sE \sR_{\hsC_{XX}}(\eta) \big]^{-1} \mathrm{d}\eta
	+  \frac{1}{2\pi i} \oint_{\Gamma_R} \frac{1}{\eta} \sR_{\hsC_{XX}}(\eta) \mathrm{d}\eta  \\
	& =  \frac{1}{2\pi i} \oint_{\Gamma_R} \frac{1}{\eta} \sR_{\hsC_{XX}}(\eta) \sE \sR_{\hsC_{XX}}(\eta)\mathrm{d}\eta -
	\frac{1}{2\pi i}
	\sum_{k\ge 2} \oint_{\Gamma_R} \frac{1}{\eta} \sR_{\hsC_{XX}}(\eta) \big[-\sE \sR_{\hsC_{XX}}(\eta)\big]^k\mathrm{d}\eta.
\end{align*}
The above states the difference is a summation of the linear term 
$P_R(\sE)$ and  the remainder term $Q_R(\sE)$ with
\begin{align}
	P_R(\sE) &=  \frac{1}{2\pi i} \oint_{\Gamma_R} \frac{1}{\eta} \sR_{\hsC_{XX}}(\eta) \sE \sR_{\hsC_{XX}}(\eta)\mathrm{d}\eta, \label{eqn:inverror:P}\\
	Q_R(\sE) &=  -\frac{1}{2\pi i} \sum_{k\ge 2} \oint_{\Gamma_R} \frac{(-1)^k}{\eta} \sR_{\hsC_{XX}}(\eta) \big[-\sE \sR_{\hsC_{XX}}(\eta)\big]^k\mathrm{d}\eta. \label{eqn:inverror:Q}
\end{align}
Given $\|\sE\| <  g_R/3$ the latter can be controled by $\|Q_R(\sE) \| \le 2K_R \|\sE\|^2 / ( g_R^2\hsigma_R^2)$. It follows from here that
\begin{align*}
	&\| \sT(\tsC_{XX}^+ - \hsC_{XX}^+)\sT^*\sD^*\sD\mY^{\perp}\|_N \\
	\le& 
	\| \sT P_R(\sE) \sT^*  \sD^*\sD\mY^{\perp}\|_N +
	\| \sT\| \|Q_R(\sE)\|   \|\sT^* \sD^*\sD\mY^{\perp}\|_N \\
	\le & (\|\sU_1 \|/\sqrt{N})(N \| \sV_1 P_R(\sE) \sV_2\|)( \|\sU_2\sD^*\sD\mY^{\perp}\|_N/\sqrt{N} ) \\
	&\qquad + 2\hsigma_1K_R / ( g_R^2\hsigma_R^2)
	\|\sE\|^2\cdot \|\sT^* \sD^*\sD\mY^{\perp}\|_N .
\end{align*}
Recall we have defined $Z_1 =\hsigma_1^6K_R / ( g_R^2\hsigma_R^2) $
and $Z_2 = \beta(R+\Delta_R)$.  
In order to prove the lemma, we   need to control 
\begin{align}
	&\Prob \Big( \| \sT(\tsC_{XX}^+ - \hsC_{XX}^+)\sT^*\sD^*\sD\mY^{\perp}\|_N \ge 
	\big(\epsilon_1+ 2Z_1\cdot  \epsilon^2_2\big)\big(1 + \sqrt{Z_2}/\epsilon_3\big)   \|\mY^{\perp}\|_N\Big) \nonumber \\
	\le & \Prob\Big( N\cdot \| \sV P_R(\sE) \sV^*\| \ge \epsilon_1\Big) + \Prob\Big( \|  \sE \|  \ge \hsigma_1^2\epsilon_2\Big)\label{eqn:regPart2:prob21} \\
	& \quad 
	+\Prob\Big( \|\sU_2 \sD^*\sD\mY^{\perp}\|_N/\sqrt{N}  \ge \big(1 + \sqrt{Z_2}/\epsilon_3\big)  \|\mY^{\perp}\|_N \Big)
	\label{eqn:regPart2:prob22} \\
	& \quad +\Prob\Big( \|\sT^* \sD^*\sD\mY^{\perp}\|_N \ge \hsigma_1 \big(1 + \sqrt{Z_2}/\epsilon_3\big)   \|\mY^{\perp}\|_N \Big)
	. \label{eqn:regPart2:prob23}
\end{align}
The second probability in~\eqref{eqn:regPart2:prob21} is controled by 
Lemma~\ref{lemma:projection:part2},  and the probability in~\eqref{eqn:regPart2:prob22} and~\eqref{eqn:regPart2:prob23} is controled by Lemma~\ref{lemma:regression:part3}. Note we need the additional requirement that $\epsilon_2 \le  g_R/3$.

In the rest of the proof, we only need to  control the first probability in~\eqref{eqn:regPart2:prob21}. Recall $\hsP_s=\htheta_r\otimes \htheta_r$, it holds that
\begin{align*}
	P_R(\sE) &= \frac{1}{2\pi i} \oint_{\Gamma_R} \frac{1}{\eta} \sR_{\hsC_{XX}}(\eta) \sE \sR_{\hsC_{XX}}(\eta)\mathrm{d}\eta
	= \frac{1}{2\pi i} \oint_{\Gamma_R}\sum_{r,s=1}^\infty \frac{1}{\eta(\eta - \hsigma_r^2)(\eta - \hsigma_s^2)} 
	\hsP_r \sE \hsP_s\mathrm{d}\eta \\
	& =  \frac{1}{2\pi i} \oint_{\Gamma_R}\sum_{r=1}^\infty \frac{1}{\eta(\eta - \hsigma_r^2)^2} \hsP_r \sE \hsP_r\mathrm{d}\eta+
	\frac{1}{2\pi i} \oint_{\Gamma_R}\sum_{r\neq s} \frac{1}{\eta(\eta - \hsigma_r^2)(\eta - \hsigma_s^2)} 
	\hsP_r \sE \hsP_s\mathrm{d}\eta \\
	& = -\sum_{r=1}^R \frac{1}{\hsigma_r^4} \hsP_r \sE \hsP_r
	+ \sum_{r=1}^{R-1} \sum_{s=r+1}^R \Big( \frac{1}{\hsigma_r^2(\hsigma_r^2 - \hsigma_s^2)} +
	\frac{1}{\hsigma_s^2(\hsigma_s^2 - \hsigma_r^2)}\Big) (\hsP_r \sE \hsP_s+\hsP_s \sE \hsP_r) \\
	&\qquad \qquad+ \sum_{r=1}^R \sum_{s=R+1}^\infty  \frac{1}{\hsigma_r^2(\hsigma_r^2 - \hsigma_s^2)} 
	(\hsP_r \sE \hsP_s + \hsP_s \sE \hsP_r) \\
	& = -\sum_{r=1}^R \frac{1}{\hsigma_r^4} \hsP_r \sE \hsP_r
	-\sum_{r=1}^{R-1} \sum_{s=r+1}^R  \frac{1}{\hsigma_r^2\hsigma_s^2} (\hsP_r \sE \hsP_s+\hsP_s \sE \hsP_r) \\
	&\qquad \qquad+ \sum_{r=1}^R \sum_{s=R+1}^\infty  \frac{1}{\hsigma_r^2(\hsigma_r^2 - \hsigma_s^2)} 
	(\hsP_r \sE \hsP_s + \hsP_s \sE \hsP_r).
\end{align*}
In the above, the third equation employs Cauchy integral formula. Recall the relation
$$
\sE = \tsC_{XX} - \hsC_{XX} = \frac{1}{C\cdot N}\sum_{c=1}^C \frac{1}{\tp_c }\tx_c\otimes \tx_c
- \hsC_{XX},
$$
we find that
\begin{align}
	P_R(\sE) &=  -\sum_{r=1}^R \frac{1}{C N \hsigma_r^2} \Big[\sum_{c=1}^C 
	\Big(\frac{\txi_{cr}^2}{\tp_c } -  N\Big)\Big] \htheta_r \otimes \htheta_r\nonumber \\
	&\qquad \qquad-\sum_{r=1}^{R-1} \sum_{s=r+1}^R  \frac{1}{CN\hsigma_r\hsigma_s}
	\Big[\sum_{c=1}^C 
	\frac{\txi_{cr}\txi_{cs}}{\tp_c } \Big]  ( \htheta_r \otimes \htheta_s +  \htheta_s \otimes \htheta_r) \nonumber\\
	&\qquad \qquad+ \sum_{r=1}^R \sum_{s=R+1}^\infty  \frac{\hsigma_s}{CN\hsigma_r(\hsigma_r^2 - \hsigma_s^2)} 
	\Big[\sum_{c=1}^C 
	\frac{\txi_{cr}\txi_{cs}}{\tp_c } \Big]  ( \htheta_r \otimes \htheta_s +  \htheta_s \otimes \htheta_r). \label{eqn:inversefirstorder}
\end{align}
In the above, $\txi_{cr}  = \langle \tx_c, \htheta_r\rangle/\hsigma_r$ is the $r$-th score of the subsampled observation $\tx_c$.

Define $\ve_s = (0,\cdots, \sqrt{N}, \cdots, 0)^T$ as a unit vector (w.r.t $\|\cdot\|_N$) in $\bbR^N$, whose $s$-th element is $\sqrt{N} $ and all the other elements are zero. Due to the similar reasoning as~\eqref{eqn:regression:lemma1:tthetat}, we can verify that
$$\sV_1  ( \htheta_r \otimes \htheta_s ) \sV_2 =
\begin{cases}
	(\hsigma_r \hsigma_s/N) \ve_r \otimes \ve_s & \text{ if } r \le R, \\
	(\hsigma_{R} \hsigma_s/N) \ve_r \otimes \ve_s & \text{ if }   r > R.
\end{cases} 
$$
It follows that
\begin{align*}
	N(\sV_1 P_R(\sE)\sV_2) &=  -\sum_{r=1}^R \frac{1}{C N } \Big[\sum_{c=1}^C 
	\Big(\frac{\txi_{cr}^2}{\tp_c } -  N\Big)\Big] \ve_r \otimes \ve_r \\
	&\qquad \qquad-\sum_{r=1}^{R-1} \sum_{s=r+1}^R  \frac{1}{CN}
	\Big[\sum_{c=1}^C \frac{\txi_{cr}\txi_{cs}}{\tp_c } \Big]  (\ve_r \otimes \ve_s+  \ve_s \otimes \ve_r) \\
	&\qquad \qquad+ \sum_{r=1}^R \sum_{s=R+1}^\infty  \frac{\hsigma_s\hsigma_{R}}{CN(\hsigma_r^2 - \hsigma_s^2)} 
	\Big[\sum_{c=1}^C \frac{\txi_{cr}\txi_{cs}}{\tp_c } \Big] 
	(\ve_r \otimes \ve_s+  \ve_s \otimes \ve_r) . \end{align*}
In order to apply Lemma~\ref{lemma:mainConcentration}, let $\sS:=C N(\sV_1 P_R(\sE)\sV_2)  = \sum_{c=1}^C \sZ_c$ with
\begin{align}
	\sZ_c &:=  -\sum_{r=1}^R \frac{1}{N }
	\Big(\frac{\txi_{cr}^2}{\tp_c } -  N\Big) \ve_r \otimes \ve_r 
	-\sum_{r=1}^{R-1} \sum_{s=r+1}^R  \frac{1}{N}
	\frac{\txi_{cr}\txi_{cs}}{\tp_c }  (\ve_r \otimes \ve_s+  \ve_s \otimes \ve_r) \nonumber \\
	&\qquad \qquad+ \sum_{r=1}^R \sum_{s=R+1}^\infty  \frac{\hsigma_s\hsigma_{R}}{N(\hsigma_r^2 - \hsigma_s^2)} 
	\frac{\txi_{cr}\txi_{cs}}{\tp_c } (\ve_r \otimes \ve_s+  \ve_s \otimes \ve_r) . \label{eqn:regression2:zc}
\end{align}
The operator norm of $\sZ_c$ is bounded by its Hilbert-Schmidt norm as
\begin{align*}
	N^2\| \sZ_c\|^2 \le  &  N^2\| \sZ_c\|^2_{HS} \\
	= &  \sum_{r=1}^R\Big(\frac{\txi_{cr}^2}{\tp_c } -  N\Big)^2 +
	2 \sum_{r=1}^{R-1} \sum_{s=r+1}^R \Big( \frac{\txi_{cr}\txi_{cs}}{\tp_c } \Big)^2
	+ 2\sum_{r=1}^R \sum_{s=R+1}^\infty  \frac{\hsigma_s^2\hsigma_{R}^2}{(\hsigma_r^2 - \hsigma_s^2)^2} 
	\Big(\frac{\txi_{cr}\txi_{cs}}{\tp_c } \Big)^2 \\
	\le &   \sum_{r=1}^R\Big(\frac{2\txi_{cr}^4}{\tp_c^2}  + 2N^2\Big) +
	4 \sum_{r=1}^{R-1} \sum_{s=r+1}^R  \frac{\txi_{cr}^2\txi_{cs}^2 }{\tp_c^2}
	+  \frac{2\hsigma_R^4}{ g_R^2} \sum_{r=1}^R \sum_{s=R+1}^\infty 
	\frac{\txi_{cr}^2(\hsigma_s^2\txi_{cs}^2/\hsigma_R^2) }{\tp_c^2 } \\
	\le & 2RN^2 + \frac{2}{\tp_c^2} \Big( \sum_{r=1}^R \txi_{cr}^2 \Big)^2 +
	\frac{2\hsigma_R^4}{ g_R^2\tp_c^2}  \Big( \sum_{r=1}^R \txi_{cr}^2 \Big)
	\Big(  \sum_{s=R+1}^\infty  \hsigma_s^2\txi_{cs}^2/\hsigma_R^2 \Big) \\
	\le & 2 RN^2+ 2 \beta^2 N^2 (R + \Delta_R)^2 
	+   G_R^2 \beta^2 N^2(R + \Delta_R)^2/2,
\end{align*}
where $ g_R = \hsigma_R^2 - \hsigma_{R+1}^2$ is the eigengap. From here, we get the upper bound that
$ \| \sZ_c\|^2 \le 2 R +  2 \beta^2  (R + \Delta_R)^2 
+   G_R^2 \beta^2 (R + \Delta_R)^2/2$. That is
\begin{align}
	\| \sZ_c\| & \le \sqrt{2R} +  \sqrt{2} \beta  (R + \Delta_R)
	+   G_R \beta (R + \Delta_R)/\sqrt{2} \nonumber \\
	& \le  \sqrt{2R}+ ( \sqrt{2}+ G_R/\sqrt{2} )\beta(R + \Delta_R).
\end{align}

Next, we manage to bound the variance of  $\sS = \sum_{c=1}^C \sZ_c$ with $\sZ_c$ defined in~\eqref{eqn:regression2:zc}. Note that  
$\|\Expect\sS^2\|  \le C\cdot \mathrm{tr} (\Expect \sZ_c^2) $. Evaluating the trace requires computing
$\langle \Expect \sZ_c^2, \ve_r \otimes \ve_r\rangle $ for $r=1,2,\cdots$. Firstly, for $r=1,\dots, R$, it holds that
\begin{align*}
	N^2 \langle \sZ_c^2, \ve_r \otimes \ve_r\rangle 
	& = \Big(\frac{\txi_{cr}^2}{\tp_c } -  N\Big)^2
	+  \sum_{\substack{s\le R\\ s\neq r} } \frac{\txi_{cr}^2\txi_{cs}^2}{\tp_c^2 } 
	+\sum_{s=R+1}^\infty  \frac{\hsigma_s^2\hsigma_{R}^2}{(\hsigma_r^2 - \hsigma_s^2)^2} 
	\frac{\txi_{cr}^2\txi_{cs}^2}{\tp_c^2 }  \\
	& \le 2 N^2+ \frac{2\txi_{cr}^2}{\tp_c^2 } 
	+ 4 \sum_{\substack{s\le R\\ s\neq r} } \frac{\txi_{cr}^2\txi_{cs}^2}{\tp_c^2} 
	+ G_R^2 \sum_{s=R+1}^\infty \frac{ \txi_{cr}^2(\hsigma_s^2\txi_{cs}^2/\hsigma_R^2)}{\tp_c^2 }.
\end{align*}
Take expectation with respect to the subsampling process
\begin{align*}
	N^2 \langle \Expect \sZ_c^2, \ve_r \otimes \ve_r\rangle 
	& \le 2N^2 + \sum_{n=1}^N\frac{2\hxi_{nr}^4}{p_n } 
	+ 4 \sum_{n=1}^N \sum_{\substack{s\le R\\ s\neq r} } \frac{\hxi_{nr}^2\hxi_{ns}^2}{p_n} 
	+ G_R^2 \sum_{n=1}^N \sum_{s=R+1}^\infty 
	\frac{\hxi_{nr}^2(\hsigma_s^2\hxi_{ns}^2/\hsigma_R^2)}{p_n }.
\end{align*}
Similarly, for $s=R+1,\cdots$ we have
\begin{align*}
	N^2 \langle \Expect \sZ_c^2, \ve_s \otimes \ve_s\rangle 
	& \le 
	G_R^2 \sum_{n=1}^N \sum_{r=1}^R \frac{ \hxi_{nr}^2(\hsigma_s^2\hxi_{ns}^2/\hsigma_R^2)}{p_n }.
\end{align*}
Combining the above, we get
\begin{align*}
	N^2\|\Expect\sS^2\| & \le CN^2\cdot \mathrm{tr} (\Expect \sZ_c^2) \le
	CN^2 \sum_{r=1}^\infty \langle \Expect \sZ_c^2, \ve_r \otimes \ve_r\rangle \\
	& \le  2CN^2+ \sum_{n=1}^N \frac{2C}{p_n} \Big( \sum_{r=1}^R \hxi_{nr}^2\Big)^2 +
	2C\cdot G_R^2 \sum_{n=1}^N \sum_{r=1}^R \sum_{s=R+1}^\infty\frac{ \hxi_{nr}^2(\hsigma_s^2\hxi_{ns}^2/\hsigma_R^2)}{p_n }.
\end{align*}
Plug in $p_n \ge\big(\sum_{r=1}^R \hxi_{nr}^2 + \sum_{s=R+1}^\infty\hsigma_s^2 \hxi_{ns}^2/\hsigma_R^2\big) /
[\beta N(R + \Delta_R)]$ to get that
\begin{align*}
	N^2\|\Expect\sS^2\| 
	& \le  2CN^2+ 2C \beta N(R + \Delta_R) \sum_{n=1}^N \frac{\Big(\sum_{r=1}^R \hxi_{nr}^2\Big)^2}{\sum_{r=1}^R \hxi_{nr}^2 + \sum_{s=R+1}^\infty\hsigma_s^2 \hxi_{ns}^2/\hsigma_R^2 } \\
	&\qquad   +
	2C \beta N(R + \Delta_R) \cdot G_R^2 \sum_{n=1}^N
	\frac{ \Big( \sum_{r=1}^R \hxi_{nr}^2\Big)\Big( \sum_{s=R+1}^\infty \hsigma_s^2\hxi_{ns}^2/\hsigma_R^2 \Big)}{\sum_{r=1}^R \hxi_{nr}^2 + \sum_{s=R+1}^\infty\hsigma_s^2 \hxi_{ns}^2/\hsigma_R^2 } \\
	& \le  2CN^2+ 2C \beta N(R + \Delta_R) \sum_{n=1}^N \sum_{r=1}^R \hxi_{nr}^2\\
	&\qquad   +
	2C \beta N(R + \Delta_R) \cdot (G_R^2/4) \sum_{n=1}^N
	\Big( \sum_{r=1}^R \hxi_{nr}^2 + \sum_{s=R+1}^\infty\hsigma_s^2 \hxi_{ns}^2/\hsigma_R^2\Big) \\
	& \le  2CN^2+ 2C \beta N^2(R + \Delta_R) R
	+ (1/2)C \beta N^2G_R^2 (R + \Delta_R)^2 .
\end{align*}
From here we get the upper bound  
\begin{align}
	\|\Expect\sS^2\|\le& 2C +
	2C \beta (R + \Delta_R) R   +C \beta G_R^2 (R + \Delta_R)^2/2 \nonumber \\
	\le & 2C +   C\beta (2+G_R^2/2)(R+\Delta_R)^2.
\end{align}
Similarly as in the proof of Lemma~\ref{lemma:projection:part1}, we can establish that
$ \mathrm{intdim} ( \Expect \sS^2)\le 2R$. Finally, it follows that for $\epsilon_1  \ge \sqrt{V_2/C} + L_2/(3C)$,
\begin{align*}
	\Prob\Big( N\cdot \| \sT P_R(\sE) \sT^*\|  \ge \epsilon_1\Big)
	\le  8R  \exp\Big(-\frac{C\epsilon^2/2}{V_2 +L_2\epsilon/3 }\Big).
\end{align*}
with $V_2 = 2 +  \beta (2+G_R^2/2)(R+\Delta_R)^2$
and $L_2 =  \sqrt{2R}+ ( \sqrt{2}+ G_R /\sqrt{2})\beta(R + \Delta_R)$.

\section{Proof of Proposition~\ref{proposition:problowerbound}}
\label{proposition:problowerbound:proof}

For the $n$-th sample $x_n$, we evaluate the difference between the full sample probability and pilot probability estimated by Algorithm~\ref{alg:sampleProb}. Our target is to control the difference
\begin{align}
	&\Big[ \sum_{r=1}^R \hxi_{nr}^2 + \|(I - \hsP_R) x_n \|^2  / \hsigma_{R}^2\Big] -\Big[\sum_{r=1}^R (\hxi_{nr}')^2 + \|(I - \tsP_R') x_n \|^2  / (\tsigma_{R}')^2 \Big] \nonumber\\
	=& \langle \hsC_{XX}^+ - (\tsC_{XX}')^+,  x_n\otimes x_n\rangle + 
	\frac{1}{(\tsigma_{R}')^2} \langle \tsP_R'-\hsP_R , x_n\otimes x_n\rangle \nonumber \\
	&\qquad\qquad+\|(I - \hsP_R) x_n \|^2 \times \Big(\frac{1}{\hsigma_{R}^2}-\frac{1}{(\tsigma_{R}')^2} \Big)
	\nonumber\\
	=&- \langle P_R(\sE'),  x_n\otimes x_n\rangle+
	\langle L_R(\sE')/(\tsigma_{R}')^2,  x_n\otimes x_n\rangle+
	\langle -Q_R(\sE') + S_R(E')/(\tsigma_{R}')^2,  x_n\otimes x_n\rangle \nonumber\\
	&\qquad \qquad +\|(I - \hsP_R') x_n \|^2 \times\Big(\frac{1}{\hsigma_{R}^2}-\frac{1}{(\tsigma_{R}')^2} \Big) \label{eqn:last:bound0}
\end{align}
The last equality relies on the decomposition~\eqref{eqn:projerror:S}, the decomposition~\eqref{eqn:inverror:P} and~\eqref{eqn:inverror:Q}, and the error $\sE'= \tsC_{XX}'-\hsC_{XX}$.

With a little abuse of notation, in the rest of this proof, we will use some simplified notations. In particular, we will denote the pilot estimates $\tsP_R'$, $\tsigma_{nr}'$, $\tsC_{XX}'$ and $\sE'$ from Line~2 of Algorithm~\ref{alg:sampleProb} by $\tsP_R$, $\tsigma_{r}$, $\tsC_{XX}$ and $\sE$, respectively. 
This should not cause confusion as we are only comparing the pilot estimates against the full sample estimates in this proof. With almost identical argument as Lemma~\ref{lemma:projection:part2}, we can establish the next result.

\begin{lemma} \label{lemma:operatorXX}
	Let  $\Delta_0= \mathrm{intdim}(\hsC_{XX} )$ and  $\hsigma_{1}^2 = \Vert \hsC_{XX} \Vert$. Suppose the sampling probability satisfies
	\begin{equation} \label{eqn:probmixed}
		p_n=\frac{1}{2} \frac{\Vert x_n \Vert^2}{\sum_{m=1}^{N}  \Vert x_m \Vert^2}+\frac{1}{2N}.
	\end{equation}  Then, 
	we have that
	$$\Prob\Big(\Vert \tsC_{XX} - \hsC_{XX} \Vert \ge \epsilon\cdot \hsigma_{1}^2 \Big)\le     4 \Delta_0 \exp\Big(- \frac{C\epsilon^2/2}{2 \Delta_0 + (2 \Delta_0+1)\epsilon/3}\Big)  \,,$$
	for $\epsilon$ satisfying $C\cdot\epsilon > \sqrt{2C\cdot \Delta_0}
	+ (2\Delta_0+ 1) /3$.
\end{lemma}
From the above lemma,  for $\epsilon =(8\Delta_0 +2)\log(120\Delta_0)/(\hsigma_{1}^2C^{1/2})\le  g_R/(3\hsigma_1^2 )$, we know that
$$
\Prob\Big(\Vert \tsC_{XX} - \hsC_{XX} \Vert \ge 
\frac{\hsigma_{1}^2\sqrt{(8\Delta_0 +2)\log(120\Delta_0)}}{C^{1/2}}
\Big)\le \frac{1}{30}.
$$
This means the event
\begin{equation} \label{eqn:last:event}
	E = \Big\{\Vert \tsC_{XX} - \hsC_{XX} \Vert \le 
	\frac{\hsigma_{1}^2\sqrt{(8\Delta_0 +2)\log(120\Delta_0)}}{C^{1/2}}\Big\}
\end{equation}
holds with probability at least $1-1/30$.

For large enough $C \ge 144\hsigma_{1}^4\gamma_0^2/g_R^2$ for $\gamma_0 = \sqrt{(\Delta_0 +1)\log(120\Delta_0)}$. Then, under the event $E$~\eqref{eqn:last:event}, 
it holds that
$$\Vert \tsC_{XX} - \hsC_{XX} \Vert  \le g_R/3\le \hsigma_R^2/2,$$ 
and it further holds that
\begin{equation} \label{eqn:last:sigmasigma}
	| \hsigma_{R}^2 - \tsigma_R^2| \le \hsigma_{R}^2/2 \implies
	-\hsigma_{R}^2/2  \le  \hsigma_{R}^2 - \tsigma_R^2 \le \hsigma_{R}^2/2 
	\implies \hsigma_{R}^2/2  \le   \tsigma_R^2 \le 3\hsigma_{R}^2/2.
\end{equation}
Under the event $E$~\eqref{eqn:last:event}, it also holds that
\begin{align} 
	\Big|\frac{1}{\tsigma_{R}^2} - 
	\frac{1}{\hsigma_{R}^2}\Big| \le
	\frac{\Vert \tsC_{XX} - \hsC_{XX} \Vert}{\tsigma_{R}^2\hsigma_{R}^2} \le
	\frac{2\hsigma_{1}^2\sqrt{(8\Delta_0 +2)\log(120\Delta_0)}}{C^{1/2}\hsigma_{R}^4},
\end{align}
where the last inequality uses the relation~\eqref{eqn:last:sigmasigma}. It follows  for the last term in~\eqref{eqn:last:bound0} that
\begin{align}
	&\|(I - \hsP_R) x_n \|^2 \times \Big|\frac{1}{\tsigma_{R}^2} - 
	\frac{1}{\hsigma_{R}^2}\Big| \nonumber \\
	\le&(\|(I - \hsP_R) x_n \|^2/\hsigma_R^2 )	\frac{2\hsigma_{1}^2\sqrt{(8\Delta_0 +2)\log(120\Delta_0)}}{C^{1/2}\hsigma_{R}^2} \nonumber\\
	\le&(\sum_{r=1}^R \hxi_{nr}^2 +\|(I - \hsP_R) x_n \|^2/\hsigma_R^2 )	\frac{2\hsigma_{1}^2\sqrt{(8\Delta_0 +2)\log(120\Delta_0)}}{C^{1/2}\hsigma_{R}^2}.\label{eqn:last:bound1}
\end{align}

Under the event $E$, it directly follows for the second order error term in~\eqref{eqn:last:bound0} that
\begin{align}
	&\Big|\langle S_R(\sE)/ \tsigma_{R}^2 -Q_r(\sE), x_n\otimes x_n \rangle\Big| \nonumber\\
	\stackrel{(i)}{\le }& K_R/g_R^2 \big[1/ \tsigma_{R}^2 +2  / \hsigma_R^2\big]\|\sE\|^2\times
	\Vert x_n\Vert^2 \nonumber\\
	\stackrel{(ii)}{\le }& 4K_R  /  (\hsigma_R^2g_R^2)\times\|\sE\|^2\times \|x_n\|^2
	\nonumber \\
	\stackrel{(iii)}{\le }& 4K_R \hsigma_{1}^2 /(g_R^2\hsigma_{R}^2 )\times\|\sE\|^2\times
	\Big[\sum_{r=1}^R  \hxi_{nr}^2 +
	\sum_{s=R+1}^\infty  \hsigma_s^2\hxi_{ns}^2/\hsigma_R^2\Big]\nonumber \\
	\le& \frac{32K_R\hsigma_{1}^6(\Delta_0 +1)\log(120\Delta_0)}{Cg_R^2\hsigma_{R}^2 } \times
	\Big[\sum_{r=1}^R \hxi_{nr}^2 + \|(I - \hsP_R) x_n \|^2  / \hsigma_{R}^2 \Big]\label{eqn:last:bound2}\\
\end{align}

In the above, the first inequality~(i) uses Lemma~\ref{lemma:lsbound} and the bound for~\eqref{eqn:inverror:Q}. Step~(ii) uses the relation~\eqref{eqn:last:sigmasigma}. 
The inequality~(iii) uses the relation $\|x_n\|^2/\hsigma_{1}^2\le \big[\sum_{r=1}^R  \hxi_{nr}^2 +
\sum_{s=R+1}^\infty  \hsigma_s^2\hxi_{ns}^2/\hsigma_R^2\big]$.

It remains to analyze the first order terms in~\eqref{eqn:last:bound0}.
By the result of Lemma~\ref{lemma:pcaLinearSum}, the first order error term  $L_R(\sE)$	can be expressed as  $L_R(\sE) = \frac{1}{C}\sum_{c=1}^C \sZ_c/(N\tp_c)$. Note in this expression, we use the sampling probability~\eqref{eqn:probmixed}. It can be evaluated that
\begin{align}
	&\frac{1}{\tsigma_{R}^2}\langle \sum_{c=1}^{C} \sZ_c/(N\tp_c) , x_n\otimes x_n\rangle\nonumber\\
	= & \frac{2}{C\tsigma_{R}^2} \sum_{r=1}^R
	\sum_{s=R+1}^\infty \frac{\hsigma_r^2\hsigma_s^2}{\hsigma_r^2 - \hsigma_s^2} \times
	\Big(\sum_{c=1}^{C}\frac{\txi_{cr}\txi_{cs}}{ N\tp_c} \Big) \times \hxi_{nr}\hxi_{ns} \nonumber\\
	\le & \frac{2 \hsigma_R^2}{C\tsigma_{R}^2}\Big[\sum_{r=1}^R
	\sum_{s=R+1}^\infty  \Big(\frac{\hsigma_r^2}{\hsigma_r^2 - \hsigma_s^2} \Big)^2 
	\Big(\sum_{c=1}^{C}\frac{\txi_{cr}\hsigma_s\txi_{cs}/\hsigma_R}{ N\tp_c} \Big)^2 \Big]^{1/2}
	\Big[\sum_{r=1}^R \sum_{s=R+1}^\infty  \hxi_{nr}^2\hsigma_s^2\hxi_{ns}^2/\hsigma_R^2\Big]^{1/2}\nonumber\\
	\le & \frac{2 G_R\hsigma_R^2}{C\tsigma_{R}^2} \Big[\sum_{r=1}^R \sum_{s=R+1}^\infty 
	\Big(\sum_{c=1}^{C}\frac{\txi_{cr}\hsigma_s\txi_{cs}/\hsigma_R}{ N\tp_c} \Big)^2 \Big]^{1/2}
	\Big[\sum_{r=1}^R  \hxi_{nr}^2\Big]^{1/2}
	\Big[\sum_{s=R+1}^\infty  \hsigma_s^2\hxi_{ns}^2/\hsigma_R^2\Big]^{1/2} \nonumber\\
	\le &\frac{G_R\hsigma_R^2}{C\tsigma_{R}^2} \Big[\sum_{r=1}^R
	\sum_{s=R+1}^\infty 
	\Big(\sum_{c=1}^{C}\frac{\txi_{cr}\hsigma_s\txi_{cs}/\hsigma_R}{ N\tp_c} \Big)^2 \Big]^{1/2}
	\Big[\sum_{r=1}^R  \hxi_{nr}^2 +
	\sum_{s=R+1}^\infty  \hsigma_s^2\hxi_{ns}^2/\hsigma_R^2\Big] .\label{eqn:last:LrEbound}
\end{align}
As the $C$ subsamples are independent with zero mean, it holds that
\begin{align*}
	\Expect \sum_{r=1}^R
	\sum_{s=R+1}^\infty 
	\Big(\sum_{c=1}^{C}\frac{\txi_{cr}\hsigma_s\txi_{cs}/\hsigma_R}{ N\tp_c} \Big)^2 
	=& C\sum_{n=1}^{N}\sum_{r=1}^R\sum_{s=R+1}^\infty 
	\frac{\hxi_{nr}^2 \hsigma_s^2 \hxi_{ns}^2 /\hsigma_R^2}{ N^2  p_n} \\
	=& C\sum_{n=1}^{N} \frac{1}{N^2 p_n}
	\Big(\sum_{r=1}^R\hxi_{nr}^2\Big)
	\Big(\sum_{s=R+1}^\infty \hsigma_s^2 \hxi_{ns}^2/\hsigma_R^2\Big)\\
	\le & \frac{C}{4}\sum_{n=1}^{N} \frac{1}{N^2 p_n}
	\Big(\sum_{r=1}^R\hxi_{nr}^2 +\sum_{s=R+1}^\infty \hsigma_s^2 \hxi_{ns}^2 /\hsigma_R^2\Big)^2.
\end{align*}
For the sampling probability~\eqref{eqn:probmixed}, we use the lower bound that
\begin{equation}
	p_n=\frac{1}{2} \frac{\Vert x_n \Vert^2}{\sum_{m=1}^{N}  \Vert x_m \Vert^2}+\frac{1}{2N} \ge \frac{1}{2} \frac{\Vert x_n \Vert^2}{\sum_{m=1}^{N}  \Vert x_m \Vert^2}=
	\frac{1}{2N} \frac{\Vert x_n \Vert^2}{\tr(\hsC_{XX})}.
\end{equation}
Then, 
\begin{align*}
	\Expect \sum_{r=1}^R
	\sum_{s=R+1}^\infty 
	\Big(\sum_{c=1}^{C}\frac{\txi_{cr}\hsigma_s\txi_{cs}/\hsigma_R}{ N\tp_c} \Big)^2 	\le & \frac{C}{2}\sum_{n=1}^{N} \frac{\tr(\hsC_{XX})}{N \cdot \| x_n\|^2}
	\Big(\sum_{r=1}^R\hxi_{nr}^2 +\sum_{s=R+1}^\infty \hsigma_s^2 \hxi_{ns}^2/\hsigma_R^2\Big)^2\\
	\le & \frac{C}{2}\sum_{n=1}^{N} \frac{\tr(\hsC_{XX})/\hsigma_{1}^2}{N }
	\Big(\sum_{r=1}^R\hxi_{nr}^2 +\sum_{s=R+1}^\infty \hsigma_s^2 \hxi_{ns}^2/\hsigma_R^2\Big)\\
	\le & \frac{C(R+\Delta_R)\tr(\hsC_{XX})}{2\hsigma_{1}^2} \\
	= & \frac{C(R+\Delta_R) \Delta_0 }{2} \le  \frac{C(R+\Delta_R)^2 }{2}. 
\end{align*}
The last inequality uses the relation that $\Delta_0 \le (R+\Delta_R)$.
By Markov inequality and the upper bound~\eqref{eqn:last:LrEbound}, with probability at least $1-\frac{1}{30}$, it holds that
\begin{align}
	\big|	\langle L_R(\sE)/ \tsigma_{R}^2 , x_n\otimes x_n\rangle\big| \le&
	\frac{\sqrt{15}G_R\hsigma_R^2(R+\Delta_R)}{C^{1/2}\tsigma_{R}^2 }\times 
	\Big[\sum_{r=1}^R  \hxi_{nr}^2 +
	\sum_{s=R+1}^\infty  \hsigma_s^2\hxi_{ns}^2/\hsigma_R^2\Big]  \nonumber\\
	\le&
	\frac{8G_R(R+\Delta_R)}{C^{1/2}} \Big[\sum_{r=1}^R \hxi_{nr}^2 + \|(I - \hsP_R) x_n \|^2  / \hsigma_{R}^2 \Big].  \label{eqn:last:bound3}
\end{align}
The last inequality uses again the relation~\eqref{eqn:last:sigmasigma}.

From the first order expression~\eqref{eqn:inversefirstorder} for the truncated inverse operator, it holds that
\begin{align}
	\langle P_R(\sE), x_n\otimes x_n\rangle &=  -\sum_{r=1}^R \frac{1}{C N } \Big[\sum_{c=1}^C 
	\Big(\frac{\txi_{cr}^2}{\tp_c } -  N\Big)\Big] \hxi_{nr}^2 -\sum_{r=1}^{R-1} \sum_{s=r+1}^R  \frac{1}{CN}
	\Big[\sum_{c=1}^C 
	\frac{\txi_{cr}\txi_{cs}}{\tp_c } \Big]  (2\hxi_{nr} \hxi_{ns}) \nonumber\\
	&\qquad \qquad+ \sum_{r=1}^R \sum_{s=R+1}^\infty  \frac{\hsigma_s^2}{CN(\hsigma_r^2 - \hsigma_s^2)} 
	\Big[\sum_{c=1}^C 
	\frac{\txi_{cr}\txi_{cs}}{\tp_c } \Big]  (2\hxi_{nr} \hxi_{ns}).
	\label{eqn:last:P1} \end{align}
The first two terms on the right hand side of the above equation can be bounded
\begin{align}
	&\left|\sum_{r=1}^R \Big[\sum_{c=1}^C 
	\Big(\frac{\txi_{cr}^2}{\tp_c } -  N\Big)\Big] \hxi_{nr}^2 +\sum_{r=1}^{R-1} \sum_{s=r+1}^R  
	\Big[\sum_{c=1}^C 
	\frac{\txi_{cr}\txi_{cs}}{\tp_c } \Big]  (2\hxi_{nr} \hxi_{ns}) \right|\nonumber\\
	\le & \Big\{\sum_{r=1}^R  \hxi_{nr}^4 +2\sum_{r=1}^{R-1}\sum_{s=r+1}^R  
	\hxi_{nr}^2 \hxi_{ns}^2 \Big\}^{1/2}\times
	\Big\{\sum_{r=1}^R \Big[\sum_{c=1}^C 
	\Big(\frac{\txi_{cr}^2}{\tp_c } -  N\Big)\Big]^2 +2\sum_{r=1}^{R-1} \sum_{s=r+1}^R   \Big[\sum_{c=1}^C  \frac{\txi_{cr}\txi_{cs}}{\tp_c } \Big]^2 \Big\}^{1/2} \nonumber\\
	= & \sum_{r=1}^R  \hxi_{nr}^2 \times
	\Big\{\sum_{r=1}^R \Big[\sum_{c=1}^C 
	\Big(\frac{\txi_{cr}^2}{\tp_c } -  N\Big)\Big]^2 +2\sum_{r=1}^{R-1} \sum_{s=r+1}^R   \Big[\sum_{c=1}^C  \frac{\txi_{cr}\txi_{cs}}{\tp_c } \Big]^2 \Big\}^{1/2}\nonumber\\
	\le & \Big\{
	\sum_{r=1}^R \hxi_{nr}^2+ \sum_{s=R+1}^\infty  \hsigma_s^2\hxi_{ns}^2/\hsigma^2_R\Big\}
	\times
	\Big\{\sum_{r=1}^R \Big[\sum_{c=1}^C 
	\Big(\frac{\txi_{cr}^2}{\tp_c } -  N\Big)\Big]^2 +2\sum_{r=1}^{R-1} \sum_{s=r+1}^R   \Big[\sum_{c=1}^C  \frac{\txi_{cr}\txi_{cs}}{\tp_c } \Big]^2 \Big\}^{1/2},\label{eqn:last:P2}
\end{align}
and
\begin{align}
	& \left|\sum_{r=1}^R \sum_{s=R+1}^\infty  \frac{\hsigma_s^2}{(\hsigma_r^2 - \hsigma_s^2)} 
	\Big[\sum_{c=1}^C 
	\frac{\txi_{cr}\txi_{cs}}{\tp_c } \Big]  (2\hxi_{nr} \hxi_{ns})\right| \nonumber\\
	\le & 2\hsigma_{R}\Big\{
	\sum_{r=1}^R \sum_{s=R+1}^\infty  \frac{1}{(\hsigma_r^2 - \hsigma_s^2)^2} 
	\Big[\sum_{c=1}^C 
	\frac{\txi_{cr}\hsigma_s\txi_{cs}}{\tp_c } \Big]^2\Big\}^{1/2}
	\times \Big\{
	\sum_{r=1}^R \sum_{s=R+1}^\infty  \hxi_{nr}^2 \hsigma_s^2\hxi_{ns}^2/\hsigma^2_R\Big\}^{1/2}\nonumber\\
	\le & \Big\{ \sum_{r=1}^R \sum_{s=R+1}^\infty  \frac{\hsigma_{R}^4}{(\hsigma_r^2 - \hsigma_s^2)^2} 
	\Big[\sum_{c=1}^C 
	\frac{\txi_{cr}\hsigma_s\txi_{cs}/\hsigma_{R}}{\tp_c } \Big]^2\Big\}^{1/2}
	\times \Big\{
	\sum_{r=1}^R \hxi_{nr}^2+ \sum_{s=R+1}^\infty  \hsigma_s^2\hxi_{ns}^2/\hsigma^2_R\Big\}\nonumber\\
	\le & \Big\{  G_R^2 \sum_{r=1}^R\sum_{s=R+1}^\infty  \Big[\sum_{c=1}^C 
	\frac{\txi_{cr}\hsigma_s\txi_{cs}/\hsigma_{R}}{\tp_c } \Big]^2\Big\}^{1/2}
	\times \Big\{
	\sum_{r=1}^R \hxi_{nr}^2+ \sum_{s=R+1}^\infty  \hsigma_s^2\hxi_{ns}^2/\hsigma^2_R\Big\}. \label{eqn:last:P3}
\end{align}
Combining the above~\eqref{eqn:last:P1}, ~\eqref{eqn:last:P2} and~\eqref{eqn:last:P3} we get
\begin{align}
	&\big|\langle P_R(\sE), x_n\otimes x_n\rangle \big|\le
	2\Big\{
	\sum_{r=1}^R \hxi_{nr}^2+ \sum_{s=R+1}^\infty  \hsigma_s^2\hxi_{ns}^2/\hsigma^2_R\Big\}
	\times \nonumber\\
	&\qquad \qquad\qquad\qquad\Big\{\sum_{r=1}^R \Big[\sum_{c=1}^C 
	\Big(\frac{\txi_{cr}^2}{\tp_c } -  N\Big)\Big]^2 +2\sum_{r=1}^{R-1} \sum_{s=r+1}^R   \Big[\sum_{c=1}^C  \frac{\txi_{cr}\txi_{cs}}{\tp_c } \Big]^2 \nonumber\\
	&\qquad\qquad\qquad\qquad\qquad\qquad+
	G_R^2 \sum_{r=1}^R \sum_{s=R+1}^\infty  \Big[\sum_{c=1}^C 
	\frac{\txi_{cr}\hsigma_s\txi_{cs}/\hsigma_R}{\tp_c } \Big]^2\Big\}^{1/2}.  \label{eqn:last:P4} \end{align}
Compute the expectation for the last line
\begin{align*}
	&\frac{1}{N^2C^2}\Expect\Big\{\sum_{r=1}^R \Big[\sum_{c=1}^C 
	\Big(\frac{\txi_{cr}^2}{\tp_c } -  N\Big)\Big]^2 +2\sum_{r=1}^{R-1} \sum_{s=r+1}^R   \Big[\sum_{c=1}^C  \frac{\txi_{cr}\txi_{cs}}{\tp_c } \Big]^2  +  G_R^2\sum_{r=1}^R \sum_{s=R+1}^\infty 
	\Big[\sum_{c=1}^C 
	\frac{\txi_{cr}\hsigma_s\txi_{cs}/\hsigma_R}{\tp_c } \Big]^2 
	\Big\} \\
	\le &\frac{1}{N^2C}\Big\{2\sum_{r=1}^R \sum_{n=1}^N
	\frac{\hxi_{nr}^4}{p_n } + 2R N^2 +4\sum_{r=1}^{R-1} \sum_{s=r+1}^R  \sum_{n=1}^N  \frac{\hxi_{nr}^2\hxi_{ns}^2}{p_n }  +4 G_R^2\sum_{r=1}^R \sum_{s=R+1}^\infty  
	\Big[\sum_{n=1}^N 
	\frac{\hxi_{nr}^2 \hsigma_s^2\hxi_{ns}^2/\hsigma^2_R }{p_n } \Big] \Big\}\\
	\le &\frac{2+G_R^2}{N^2C}\sum_{n=1}^N \frac{1}{p_n} \Big[\sum_{r=1}^R \hxi_{nr}^2+
	\sum_{s=R+1}^\infty\hsigma_{s}^2\hxi_{ns}^2/\hsigma^2_R\Big]^2+ 2R/C \\
	\le &\frac{4+2G_R^2}{C}\sum_{n=1}^N \frac{\tr(\hsC_{XX})}{N \cdot \| x_n\|^2} \Big[\sum_{r=1}^R \hxi_{nr}^2+
	\sum_{s=R+1}^\infty\hsigma_{s}^2\hxi_{ns}^2/\hsigma^2_R\Big]^2+ 2R/C \\
	\le &\frac{4+2G_R^2}{C\hsigma_{1}^2}\sum_{n=1}^N \frac{\tr(\hsC_{XX})}{N} \Big[\sum_{r=1}^R \hxi_{nr}^2+
	\sum_{s=R+1}^\infty\hsigma_{s}^2\hxi_{ns}^2/\hsigma^2_R\Big]+ 2R/C \\
	\le &\frac{(4+2G_R^2 )\Delta_0}{C}  \big[R+\Delta_R\big]+ 2R/C \\
	\le &\frac{(6+2G_R^2 ) }{C}  \big[R+\Delta_R\big]^2.
\end{align*}
Again, by Markov inequality and~\eqref{eqn:last:P4}, with probability at least $1-\frac{1}{30}$, it holds that
\begin{align}
	|\langle P_R(\sE), x_n\otimes x_n\rangle |\le 
	&\frac{2(30)^{1/2}(6+2G_R^2 )^{1/2} }{C^{1/2}}  \big[R+\Delta_R\big] \times \Big[\sum_{r=1}^R \hxi_{nr}^2 + \|(I - \hsP_R) x_n \|^2  / \hsigma_{R}^2 \Big]\nonumber\\
	\le 
	&\frac{27(1+G_R ) }{C^{1/2}}  \big[R+\Delta_R\big] \times \Big[\sum_{r=1}^R \hxi_{nr}^2 + \|(I - \hsP_R) x_n \|^2  / \hsigma_{R}^2 \Big]  \label{eqn:last:bound4}
\end{align}

Combining \eqref{eqn:last:bound0}, \eqref{eqn:last:bound1}, 
\eqref{eqn:last:bound2}, \eqref{eqn:last:bound3} and
\eqref{eqn:last:bound4}, we arrive at
\begin{align*}
	&\left|\Big[ \sum_{r=1}^R \hxi_{nr}^2 + \|(I - \hsP_R) x_n \|^2  / \hsigma_{R}^2\Big] -\Big[\sum_{r=1}^R (\hxi_{nr}')^2 + \|(I - \tsP_R') x_n \|^2  / (\tsigma_{R}')^2 \Big] \right|\\
	\le& \gamma\times\Big[\sum_{r=1}^R (\hxi_{nr})^2 + \|(I - \hsP_R) x_n \|^2  / \hsigma_{R}^2 \Big] .
\end{align*}
with probability at least $1-1/10$  and
\begin{align*}
	\gamma = \frac{35(1+G_R ) }{C^{1/2}} \big[R+\Delta_R\big] 
	+\frac{8\hsigma_{1}^2\gamma_0}{C^{1/2}\hsigma_{R}^2}+
	\frac{32K_R\hsigma_{1}^6\gamma_0^2}{Cg_R^2\hsigma_{R}^2 } .
\end{align*}
In the above $\gamma_0 = \sqrt{(\Delta_0 +1)\log(120\Delta_0)}$. From the above, we can conclude the results of Proposition~\ref{proposition:problowerbound}.

\clearpage 
\section{Additional Results for Simulation Study}
\label{sec:supp:simulation}

Section~\ref{sec:simu} conducts the randomized FPCA in Algorithm~\ref{alg:sampleFPCA} and the randomized FLR in Algorithm~\ref{alg:sampleFLR} with three sampling probabilities. These include: 1) the naive uniform sampling probability (UNIF), $p_n = 1/N$; 2) the importance sampling probability \citep[IMPO,][]{he2020randomized}, $p_n  = \| x_n\|^2 / \sum_{m = 1}^N\Vert x_m\Vert^2$; 3) the proposed  functional principal subspace sampling probability (FunPrinSS)  estimated by Algorithm~\ref{alg:sampleProb} with $\alpha = 0.5$.
This section contains additional simulation results of Section~\ref{sec:simu} with the  \texttt{Polynomial Decay (PD)} eigenvalue setting. See Table~\ref{tab:simu:paras} for  details. 

Figure~\ref{fig:ErrorNUpoly}--\ref{fig:ErrorVNpoly} show the results for the estimation of covariance operator estimation and functional principal components. They correspond to the NU, MN and VN distribution settings for the scores, respectively.  The interpretation of these figures is similar to that of Figure~\ref{fig:ErrorNUexp}--\ref{fig:ErrorVNexp}. The IMPO sampling has the highest accuracy in estimating $\hsC_{XX}$ and $\htheta_1$, while our proposed FunPrinSS sampling leads in performance of estimating the other eigenfunctions and the projection operator $\hsP_R$.

Figure~\ref{fig:Regressionpoly} shows the functional linear regression (FLR) result of the PD eigenvalue setting. The figure conveys similar messages as Figure~\ref{fig:Regressionexp}.  Still, our proposed FunPrinSS sampling leads in performance. The UNIF sampling has higher accuracy than IMPO in the NU distribution setting, but fails completely in the VN distribution setting.

\begin{figure}[p]
	\centering
	\includegraphics[width = 1\textwidth, height = 0.5\textwidth]{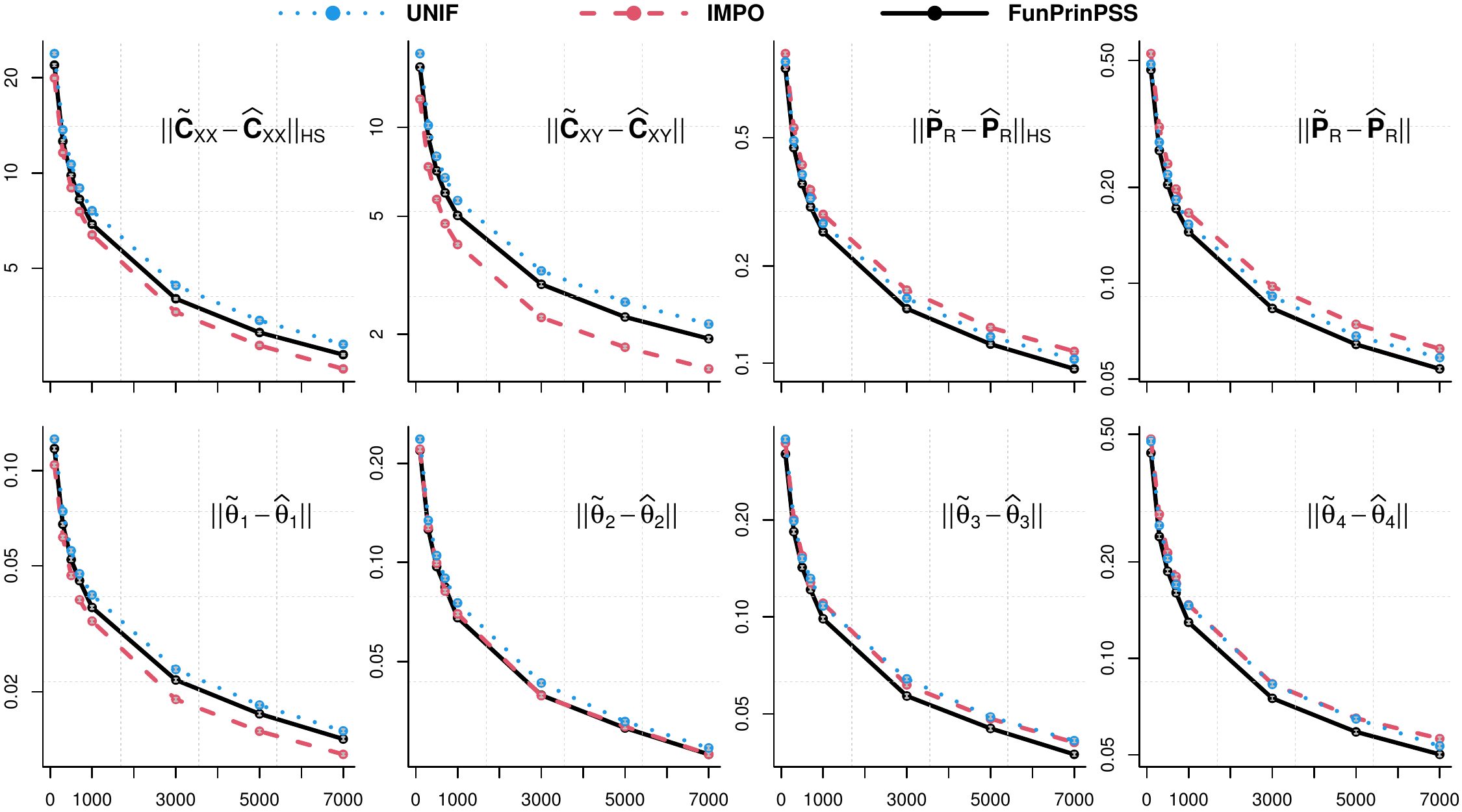} 
	\caption{Randomized covariance operator estimation and FPCA for the PD and NU setting.}\label{fig:ErrorNUpoly}
\end{figure}
\begin{figure}[h!]
	\centering
	\includegraphics[width = 1\textwidth, height = 0.5\textwidth]{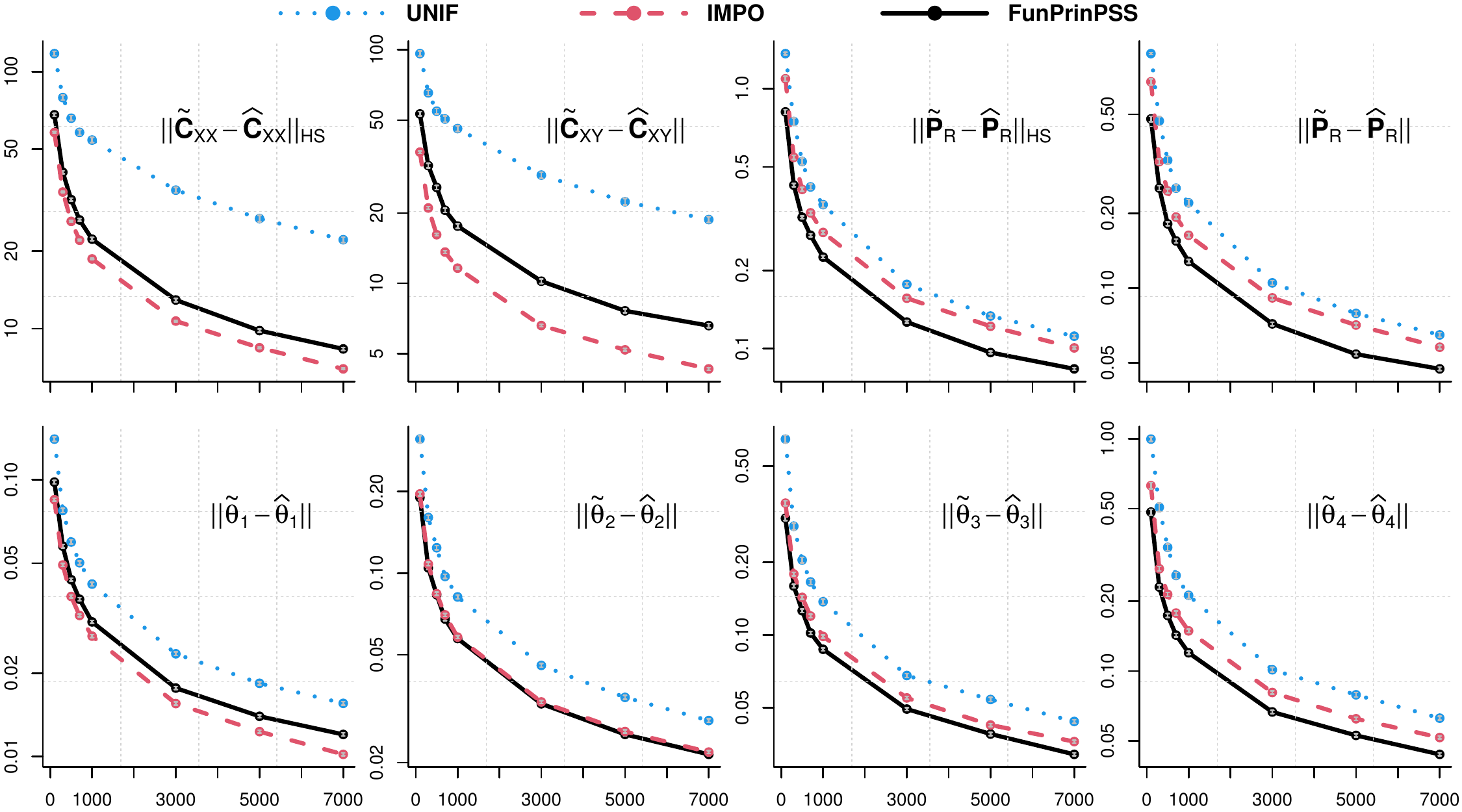} 
	\caption{Randomized covariance operator estimation and FPCA for the PD and MN setting.}\label{fig:ErrorMNpoly}
\end{figure}
\begin{figure}[h!]
	\centering
	\includegraphics[width = 1\textwidth, height = 0.5\textwidth]{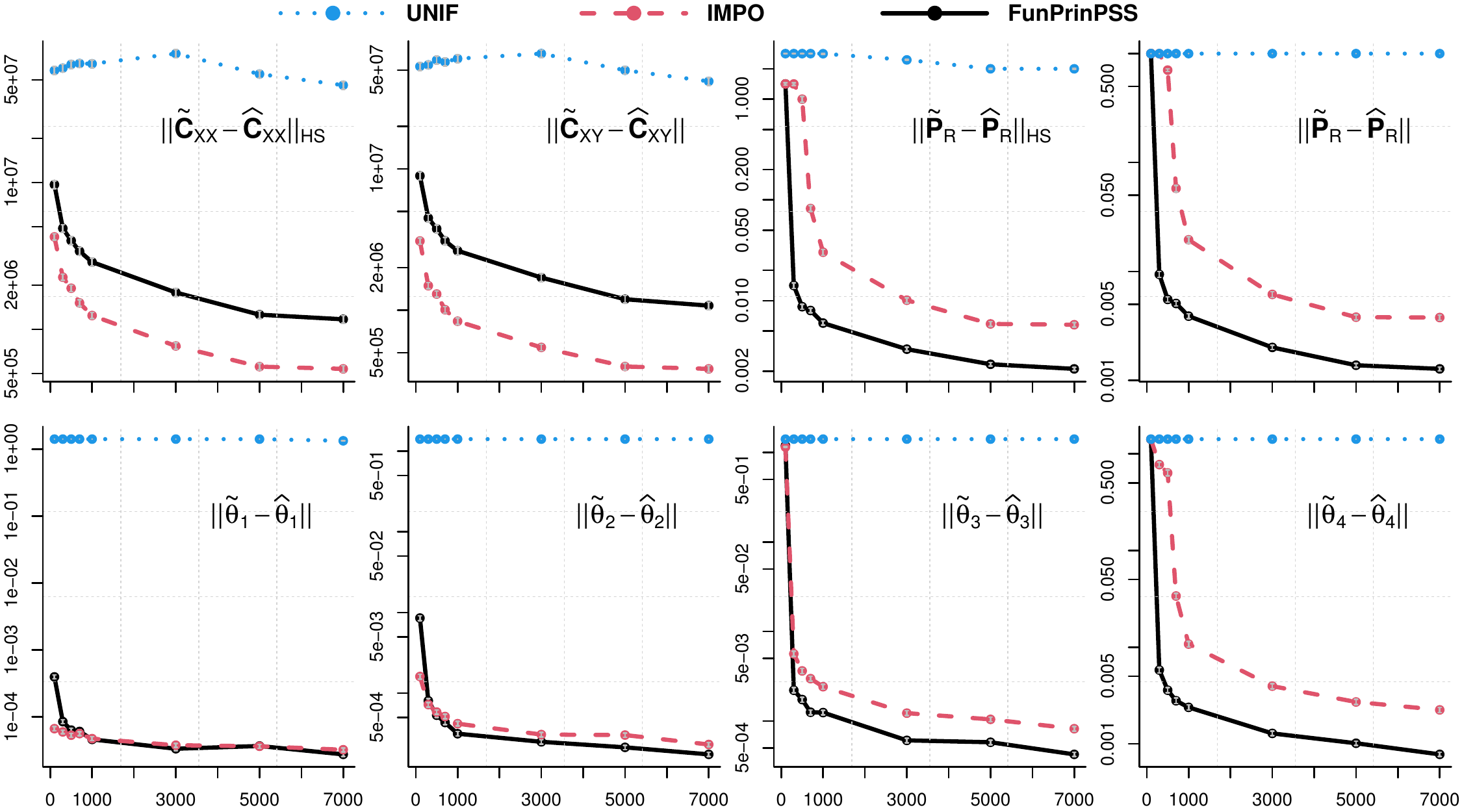} 
	\caption{Randomized covariance operator estimation and FPCA for the PD and VN setting.}\label{fig:ErrorVNpoly}
\end{figure}
\begin{figure}[h!]
	\centering
	\includegraphics[width = 1\textwidth, height = 0.5\textwidth]{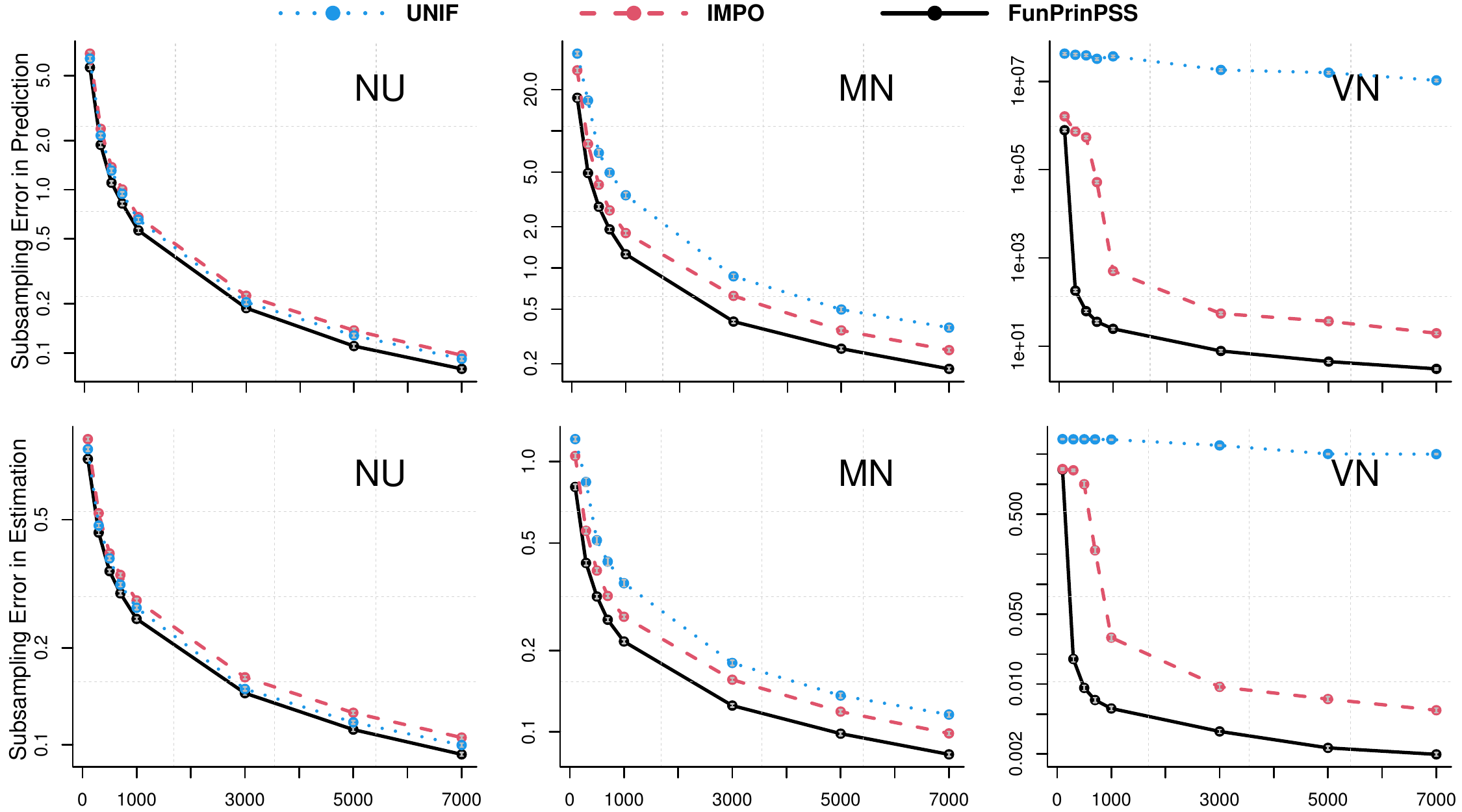} 
	\caption{Randomized FLR for predictors generated with polynomially decaying  eigenvalues.}\label{fig:Regressionpoly}
\end{figure}

\end{document}